\documentclass[journal, onecolumn]{IEEEtran}
\IEEEoverridecommandlockouts
\usepackage{cite}
\usepackage{amsmath,amssymb,amsfonts}
\usepackage{graphicx}
\usepackage{textcomp}
\usepackage{xcolor}
\usepackage[utf8]{inputenc} 
\usepackage[T1]{fontenc}    
\usepackage{hyperref}       
\usepackage{url}            
\usepackage{booktabs}       
\usepackage{amsfonts}       
\usepackage{tikz}
\usepackage{pgfplots}
\pgfplotsset{compat=1.18}
\usetikzlibrary{calc,fit,patterns}
\usepackage{bbm}
\usepackage{amsthm}
\usepackage{stmaryrd}
\usepackage{mathrsfs}
\usepackage{algpseudocode}
\usepackage{amsmath}
\usepackage{algorithm}
\usepackage{amssymb}
\usepackage{pifont}
\usepackage{subcaption}
\usepackage{nicefrac}       
\usepackage{microtype}      
\usepackage{multirow}
\usepackage{xcolor}
\usepackage{lipsum}
\usepackage{ulem} 
\usepackage{mathtools}
\usepackage{float}

\newcommand{\barp}{\bar{p}}

\newcommand{\barepsilon}{\bar{\varepsilon}}

\newcommand{\Fai}[1][]{\mathscr{F}_{\mathfrak{a}^{#1},1}}
\newcommand{\Faii}[1][]{\mathscr{F}_{\mathfrak{a}^{#1},2}}
\newcommand{\Fphi}[1][]{\mathscr{F}_{\varphi}^{#1}}
\newcommand{\Fpsii}[1][]{\mathscr{F}_{\psi,1}^{#1}}
\newcommand{\Fpsiii}[1][]{\mathscr{F}_{\psi,2}^{#1}}
\newcommand{\Ga}[1][]{\mathscr{G}_{|\mathfrak{a}^{#1}|}}
\newcommand{\Ba}[1][]{\mathfrak{B}(|\mathfrak{a}^{#1}|)}
\newcommand{\tilphi}[1][]{\tilde{\varphi}^{#1}}

\def\Xint#1{\mathchoice
   {\XXint\displaystyle\textstyle{#1}}%
   {\XXint\textstyle\scriptstyle{#1}}%
   {\XXint\scriptstyle\scriptscriptstyle{#1}}%
   {\XXint\scriptscriptstyle\scriptscriptstyle{#1}}%
   \!\int}
\def\XXint#1#2#3{{\setbox0=\hbox{$#1{#2#3}{\int}$}
     \vcenter{\hbox{$#2#3$}}\kern-.5\wd0}}

\def\dashint{\Xint-}

\algnewcommand\algorithmicswitch{\textbf{switch}}
\algnewcommand\algorithmiccase{\textbf{case}}
\algnewcommand\algorithmicassert{\texttt{assert}}
\algnewcommand\Assert[1]{\State \algorithmicassert(#1)}%
\algdef{SE}[SWITCH]{Switch}{EndSwitch}[1]{\algorithmicswitch\ #1\ \algorithmicdo}{\algorithmicend\ \algorithmicswitch}%
\algdef{SE}[CASE]{Case}{EndCase}[1]{\algorithmiccase\ #1}{\algorithmicend\ \algorithmiccase}%

\def\BibTeX{{\rm B\kern-.05em{\sc i\kern-.025em b}\kern-.08em
    T\kern-.1667em\lower.7ex\hbox{E}\kern-.125emX}}

\theoremstyle{definition}
\newtheorem{exmp}{Example}[section]
\newtheorem{theorem}{Theorem}[section]
\newtheorem{corollary}{Corollary}[theorem]
\newtheorem{lemma}[theorem]{Lemma}
\newtheorem{prop}{Proposition}[section]
\newtheorem{definition}{Definition}[section]
\theoremstyle{remark}
\newtheorem*{remark}{Remark}

\begin{document}
\usetikzlibrary{arrows.meta}
\usetikzlibrary{positioning}
\title{Complex Analysis of Channel Polarization on Discrete BMS Channels \\
}

\author{
\begin{center}
Dongxiao Xu, Moritz Wiese, Holger Boche \\
\textit{Chair of Theoretical Information Technology}\\
\textit{Technical University of Munich}\\
Munich, Germany\\
Emails: \{dongxiao.xu, wiese, boche\}@tum.de
\end{center}
}

\maketitle

\begin{abstract}
We develop component evolution (CE), a framework based on complex function theory for finite-blocklength channel polarization on discrete binary-input memoryless output-symmetric (BMS) channels. In this view, the Bhattacharyya parameter is treated as a real-valued instance of a broader class of complex-valued channel functionals. CE systematically derives analytic expressions for the Bhattacharyya parameters of the bit-channels of a given discrete BMS channel at arbitrary polarization levels. CE also enables structural analysis, providing new evidence of extremality of the binary erasure channel (BEC) and binary symmetric channel (BSC), and revealing new channel-dependent recursions for a class of BSC bit-channels.
\end{abstract}


\section{Introduction}

Polar code construction for a given \textit{binary-input discrete memoryless channel} (B-DMC) $W$ was formulated in~\cite[Sec. IX]{arikan2009channel} as a decision problem: given the blocklength $N\triangleq 2^k$ at polarization level $k$, an index $i \in [N]$, and a threshold $\gamma \in [0,1]$, determine whether $i\in \mathcal{A}_\gamma$, where $\mathcal{A}_\gamma \triangleq\{i \in[N]: Z(W_N^{(i)})<\gamma\}$ and $Z(W_N^{(i)})$ denotes the Bhattacharyya parameter of the bit-channel $W_N^{(i)}$. This problem can be formulated via density evolution (DE)~\cite{mori2009performance, mori2009density}. While exact for the \textit{binary erasure channel} (BEC), DE becomes intractable for a rigorous mathematical analysis for general \textit{binary-input memoryless output-symmetric} (BMS) channels, including the \textit{binary symmetric channel} (BSC). Tal and Vardy’s degrading and upgrading quantizations~\cite{tal2013construct} enable efficient construction with provable bounds, while practical methods such as the Gaussian approximation~\cite{trifonov2012efficient,dai2017does} and the partial-order~\cite{schurch2016partial, mondelli2018construction, he2017beta} further reduce complexity. However, despite its importance for code construction, the structure of $\mathcal{A}_\gamma$ remains poorly understood beyond degradation relations~\cite{renes2015alignment, korada2009polar}. In the infinite-blocklength limit, polar codes exhibit a fractal-like structure~\cite{geiger2018fractality}, but it remains unclear how these asymptotic properties translate to finite-length codes.

This motivates us to investigate the deterministic mechanism of channel polarization on discrete BMS channels at finite blocklength via complex analysis, in contrast to the martingale-based perspective from measure theory. From this perspective, channel polarization is essentially the evolution of entire functions on variable- and check-domains; boundary cases like the BEC arise only as limiting instances and are not representative of the general behavior. Viewing the Bhattacharyya parameter as a real-valued instance of a broader class of complex-valued channel functionals, we introduce a framework of \textit{component evolution} (CE). CE models channel polarization as a complex dynamic system, offers a frequency-domain viewpoint complementary to DE, and provides a systematic procedure for deriving analytic expressions of $Z(W_N^{(i)})$ for every $i \in [N]$ and a given discrete BMS channel $W$. Although the number of additive terms in $Z(W_N^{(i)})$ grows exponentially with polarization level $k$, these expressions are derived deterministically, and can be pre-evaluated and stored as lookup tables for classes of discrete BMS channels.

More importantly, CE is not only a machinery for deriving $Z(W_N^{(i)})$, but also a framework for structural analysis, providing new evidence of extremality~\cite{alsan2012channel,richardson2008modern} of the BEC and BSC through the lens of complex analysis, revealing new channel-dependent recursions (e.g., for a class of BSC bit-channels)—thereby serving as both a symbolic derivation tool and a unifying analytical framework for analyzing polarization phenomenon.

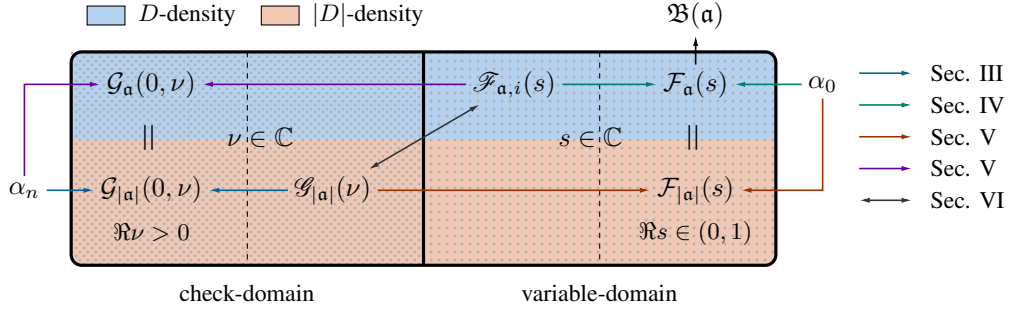
\begin{figure}[!ht]
\centering
\begin{tikzpicture}[x=1.2cm,y=1.4cm]
\tikzset{
  m/.style={inner sep=2pt, text opacity=1},
  dot/.style={circle,fill=black,inner sep=1.2pt},
  solid arrow/.style={arrows = {-Latex[width'=0pt .5, length=4pt]}, shorten >=.25pt, shorten <=.25pt, line width=0.6pt},
  double arrow/.style={arrows = {Latex[width'=0pt .5, length=4pt]-Latex[width'=0pt .5, length=4pt]}, shorten >=.25pt, shorten <=.25pt, line width=0.6pt},
}

\coordinate (l1-0) at (0,1);
\coordinate (l1-1) at (4,1);
\coordinate (l1-2) at (6,1);
\coordinate (l1-3) at (6,1.65);

\coordinate (l2-0) at (0,0);
\coordinate (l2-1) at (2,0);
\coordinate (l2-2) at (6,0);

\node[inner sep=0pt, fit=(l1-0) (l1-1) (l1-2) (l2-0) (l2-1) (l2-2), inner ysep=20pt, inner xsep=30pt, yshift=-8pt] (box) {};

\definecolor{topcolor}{RGB}{180,210,240}    
\definecolor{bottomcolor}{RGB}{240,200,180} 

\begin{scope}
    \clip (box.north west) rectangle ($(box.south east)!0.59!(box.north east)$);
    \fill[topcolor, fill opacity=0.3, rounded corners] (box.north west) rectangle (box.south east);
\end{scope}

\begin{scope}
    \clip ($(box.south west)!0.59!(box.north west)$) rectangle (box.south east);
    \fill[bottomcolor, fill opacity=0.3, rounded corners] (box.north west) rectangle (box.south east);
\end{scope}

\begin{scope}
    \clip ($(box.north west)!0.5!(box.north east)$) rectangle (box.south east);
    \fill[pattern=dots, pattern color=gray!60, opacity=0.5] (box.north west) rectangle (box.south east);
\end{scope}

\begin{scope}
    \clip (box.north west) rectangle ($(box.south west)!0.5!(box.south east)$);
    \fill[pattern=crosshatch dots, pattern color=gray!60, opacity=0.5] (box.north west) rectangle (box.south east);
\end{scope}

\draw[line width=1.2pt] 
    ($(box.north west)!1/2!(box.north east)$) -- ($(box.south west)!1/2!(box.south east)$);

\draw[dash pattern=on 2pt off 2pt] 
    ($(box.north west)!1/4!(box.north east)$) -- ($(box.south west)!1/4!(box.south east)$);

\draw[dash pattern=on 2pt off 2pt] 
    ($(box.north west)!3/4!(box.north east)$) -- ($(box.south west)!3/4!(box.south east)$);

\draw[line width=1.2pt, rounded corners] 
    (box.north west) rectangle (box.south east);

\node[m] (L1-0) at (l1-0) {$\mathcal{G}_{\mathfrak{a}}(0,\nu)$};
\node[m] (L1-1) at (l1-1) {$\mathscr{F}_{\mathfrak{a},i}(s)$};
\node[m] (L1-2) at (l1-2) {$\mathcal{F}_{\mathfrak{a}}(s)$};
\node[m] (L1-3) at (l1-3) {$\mathfrak{B}(\mathfrak{a})$};

\node[m] (L2-0) at (l2-0) {$\mathcal{G}_{|\mathfrak{a}|}(0,\nu)$};
\node[m] (L2-1) at (l2-1) {$\mathscr{G}_{|\mathfrak{a}|}(\nu)$};
\node[m] (L2-2) at (l2-2) {$\mathcal{F}_{|\mathfrak{a}|}(s)$};

\node[m] (C1) at (4.84,0.5) {$s\in\mathbb{C}$};
\node[m] (C2) at (1.2,0.5) {$\nu\in\mathbb{C}$};

\node[m, font=\small] (SOA1) at (6,-0.4) {$\Re s\in(0,1)$};
\node[m, font=\small] (SOA2) at (0,-0.4) {$\Re \nu >0$};

\node[m] (A1) at (-1.4,0) {$\alpha_n$};
\node[m] (A2) at (7.4,1) {$\alpha_0$};

\definecolor{arrow1}{RGB}{0,102,153}    
\definecolor{arrow2}{RGB}{153,51,0}     
\definecolor{arrow3}{RGB}{0,128,102}    
\definecolor{arrow4}{RGB}{102,0,153}    
\definecolor{arrow5}{RGB}{51,51,51}     

\draw[solid arrow, color=arrow2] (A2) |- (L2-2);
\draw[solid arrow, color=arrow4] (A1) |- (L1-0);

\draw[solid arrow, color=arrow4] (L1-1) -- (L1-0);
\draw[solid arrow, color=arrow3] (L1-1) -- (L1-2);
\draw[solid arrow, color=arrow3] (A2) -- (L1-2);

\draw[solid arrow, color=arrow1] (L2-1) -- (L2-0);
\draw[solid arrow, color=arrow2] (L2-1) -- (L2-2);
\draw[solid arrow, color=arrow1] (A1) -- (L2-0);

\draw[solid arrow] (L1-2) -- (L1-3);

\draw[double arrow, color=arrow5] (L1-1) -- (L2-1);

\node at ($(L1-0)!0.5!(L2-0)$) {\Large\rotatebox{90}{$=$}};
\node at ($(L1-2)!0.5!(L2-2)$) {\Large\rotatebox{90}{$=$}};

\node[draw, fill=topcolor, fill opacity=0.3, minimum width=0.5cm, minimum height=0.25cm] at (-0.5,1.65) {};
\node[anchor=west, font=\small] at (-0.25,1.65) {$D$-density};

\node[draw, fill=bottomcolor, fill opacity=0.3, minimum width=0.5cm, minimum height=0.25cm] at (1.42,1.65) {};
\node[anchor=west, font=\small] at (1.67,1.65) {$|D|$-density};


\node[anchor=north, font=\small] 
    at ($(box.south west)!1/4!(box.south east) + (0,-0.1)$) 
    {check-domain};

\node[anchor=north, font=\small] 
    at ($(box.south west)!3/4!(box.south east) + (0,-0.1)$) 
    {variable-domain};


\begin{scope}[xshift=1.2cm] 

\node[anchor=west] at ($(box.north east)+(0.8,-0.2)$) {
    \tikz\draw[solid arrow, color=arrow1] (0,0) -- (0.6,0);
};
\node[anchor=west, font=\small] at ($(box.north east)+(1.6,-0.2)$) {Sec.~\ref{sec:check_conv_domain}};

\node[anchor=west] at ($(box.north east)+(0.8,-0.8)$) {
    \tikz\draw[solid arrow, color=arrow2] (0,0) -- (0.6,0);
};
\node[anchor=west, font=\small] at ($(box.north east)+(1.6,-0.8)$) {Sec.~\ref{sec:convert_G/F}};
---
\node[anchor=west] at ($(box.north east)+(0.8,-0.5)$) {
    \tikz\draw[solid arrow, color=arrow3] (0,0) -- (0.6,0);
};
\node[anchor=west, font=\small] at ($(box.north east)+(1.6,-0.5)$) {Sec.~\ref{sec:var_conv_domain}};

\node[anchor=west] at ($(box.north east)+(0.8,-1.1)$) {
    \tikz\draw[solid arrow, color=arrow4] (0,0) -- (0.6,0);
};
\node[anchor=west, font=\small] at ($(box.north east)+(1.6,-1.1)$) {Sec.~\ref{sec:convert_G/F}};

\node[anchor=west] at ($(box.north east)+(0.8,-1.4)$) {
    \tikz\draw[double arrow, color=arrow5] (0,0) -- (0.6,0);
};
\node[anchor=west, font=\small] at ($(box.north east)+(1.6,-1.4)$) {Sec.~\ref{sec:bi_convert_GF}};

\end{scope}

\end{tikzpicture}
\caption{Relations between the check-domain and variable -domain transforms on their SOAs and the corresponding entire entities in $D$- and $|D|$-densities.}
\label{fig:entities_domains}
\end{figure}

\section{Preliminaries and notation}~\label{sec:prelim}

Throughout, the letter $\mathrm{a}$, $a$, $\mathfrak{a}$ and $|\mathfrak{a}|$ denote $L$-, $G$-, $D$- and $|D|$-density~\cite{richardson2008modern}, respectively. Let $\bar{x}\triangleq 1-x$, $\forall x\in[0,1]$. The complex conjugate of $z\in\mathbb{C}$ is denoted by $z^{*}$. We consider $W_{\mathrm{BEC}}(\varepsilon)$ with erasure $\varepsilon\in(0,1)$ and $W_{\mathrm{BSC}}(p)$ with crossover $p\in(0,1/2)$, and define $\Delta \triangleq \log (p / \bar{p})<0$, $C \triangleq p \bar{p}$, $M_i \triangleq (1-2p)^i$, $S_i\triangleq \barp^i+p^i$, and $D_i\triangleq \barp^i-p^i$ with $i\in \mathbb{Z}_{> 0}$. Define $A_s\triangleq A(s) \triangleq \bar{p}^{1-s} p^s$, $B_s \triangleq B(s) \triangleq p^{1-s} \bar{p}^s$, $\forall s\in\mathbb{C}$. It follows that: \textit{i}) $C=A_sB_s$; \textit{ii}) $B_s=A_{1-s}$.


\subsection{Deterministic DE view}~\label{sub:comp_chain}
Let $b_1\ldots b_k$ denote the $k$-bit binary expansion of $i-1$ for $i\in[N]$, i.e., $(i-1)_{10}=(b_1\ldots b_k)_2$, and let $W_{b_1, \ldots, b_k} \triangleq W_N^{(i)}$, then $W_N^{(i)}$ can be obtained by repeating the following channel operations~\cite[Section 2.3]{korada2009polar}, i.e.,
\begin{equation*}
    W_{b_1, \ldots, b_{m-1}, b_m}= \begin{cases}W_{b_1, \ldots, b_{m-1}} \boxast W_{b_1, \ldots, b_{m-1}}, & \text { if } b_m=0, \\ W_{b_1, \ldots, b_{m-1}} \varoast W_{b_1, \ldots, b_{m-1}}, & \text { if } b_m=1 ,\end{cases}
\end{equation*}
with the initial condition
\begin{equation*}
    W_{b_1}=\begin{cases}
        W \boxast W, & \text { if } b_1=0, \\ 
        W \circledast W, & \text { if } b_1=1. \\ 
    \end{cases}
\end{equation*}
In the original work of Arıkan~\cite[Section IV]{arikan2009channel}, such recursive relationship is regarded as a \textit{random tree process} $\{K_m(\xi)=W_{\xi_1 \ldots \xi_m}; m\geq 0\}$ where the bit-channel as random variable $K_m$ is driven by a sequence of i.i.d. equiprobable Bernoulli random variables $\left\{B_m \in\{ 0,1 \}; m\geq 0\right\}$. Equivalently, conditioned on $B_1=b_1,\ldots, B_m=b_m$, the bit-channel $W_{b_1, \ldots, b_m}$ up to the current level $m$ is a realization of random variable $W_{\xi_1 \ldots \xi_m}$ as outcome of the Bernoulli experiments. The apparent "randomness" of $K_m$ therefore arises only from uniformly sampling a path through the deterministic binary tree of bit-channels; in fact, the nature of each bit-channels $W_{b_1, \ldots, b_m}$ is completely determined by: \textit{i}) the initial channel $W$; and \textit{ii}) the ordered application of the $\boxast$- and $\varoast$-convolutions. In this regard, for a given discrete BMS channel $W$, the solution of the Polar code construction defined in~\cite[Sec. IX]{arikan2009channel} is unique once the threshold $\gamma$ and blocklength $N$ are fixed.

We assign operational meaning to $0$ and $1$ in terms of the the $\boxast$- and $\varoast$-convolution, respectively, i.e., $b_m\in\{0,1\}\equiv \{\boxast,\varoast\}$, and also adopt the shorthand, e.g., $W \boxast W = W^{\boxast}$. For instance, since $(6-1)_{10}=(101)_2$, we write $W_8^{(6)}$ as $W_{101} \equiv W_{\oast\boxast\oast} \triangleq W_{\oast\boxast}^\oast=(W_{\oast\boxast})\oast (W_{\oast\boxast})$. Such an expression is clean and convenient to read: at level $m=3$, the bit-channel $W_{\oast\boxast\oast}$ is obtained by conducting a self-$\varoast$-convolution on its preceding bit-channel $W_{\oast\boxast}$. Moreover, we write $W_8^{(2)}=W_{001}=W_{\boxast 2\oast}$ which compactly denotes two successive self-$\boxast$-convolutions followed by one self-$\varoast$-convolution.

\subsection{Channel models and densities}~\label{sub:density}
In this paper, we consider discrete BMS channels, although a few definitions and results can be extended to the non-symmetric case. A discrete BMS channel has binary input alphabet $\mathcal{X}=\{\pm1\}$, finite output alphabet $\mathcal{Y}$, and transition probabilities $p_{Y|X}(y|x)$, for which there exists a permutation $\pi_1$ on $\mathcal{Y}$ such that $\pi_1^{-1}=\pi_1$ and $p_{Y|X}(y|1)=p_{Y|X}(\pi_1(y)|-1)$ for all $y\in\mathcal{Y}$. It can be equivalently described via induced densities. One is the $L$-distribution~\cite[p.177]{richardson2008modern} of the random variable (r.v.) $L=l(Y) \triangleq \log \frac{p_{Y \mid X}(Y \mid 1)}{p_{Y \mid X}(Y \mid-1)}$ conditioned on $X=1$. A related r.v. is $D\triangleq \tanh \left(L/2\right)$ whose distribution conditioned on $X=1$ is called the $D$-distribution~\cite[p.180]{richardson2008modern}. Another r.v. as tuple is $(\mathfrak{H}(L), \ln \coth(|L|/2))$ whose distribution conditioned on $X=1$ is called the $G$-distribution~\cite[p.182]{richardson2008modern}, where $\mathfrak{H}(\cdot)\in\{\pm 1\}$ denotes the \textit{hard-decision} function~\cite[(4.18)]{richardson2008modern}. Their absolute values induce distributions too, e.g., the absolute value $|D|$ induces the $|D|$-distribution. Accordingly, we denote the associated densities by the $L$-density $\mathrm{a}$, $D$-density $\mathfrak{a}$, $G$-density $a$, $|D|$-density $|\mathfrak{a}|$. In this paper, we primarily work with the $D$- and $|D|$-densities, and the Bhattacharyya parameter $Z(W_N^{(i)})$ is written interchangeably as $\mathfrak{B}(\mathfrak{a}^{b_1 \ldots b_k})$ or $\mathfrak{B}(|\mathfrak{a}^{b_1 \ldots b_k}|)$.

\begin{theorem}~\cite[Theorem 4.27]{richardson2008modern}
The $L$-, $D$- and $G$-distributions corresponding to a BMS channel are symmetric.
\end{theorem}

Under symmetry, any density among $(L,D,G)$ is completely determined by its absolute-value-induced density, e.g., the $D$-density is determined by the corresponding $|D|$-density~\cite[p.180]{richardson2008modern}.

\begin{lemma}[$|D|$-density]~\label{lemma:absD_dens}
The $|D|$-density of a discrete BMS channel admits the form
\begin{equation}
    |\mathfrak{a}|(z)=\sum_{i=0}^{n} \alpha_i \delta\left(z-z_i\right),
\end{equation}
where $\alpha_i\in[0,1]$, $\sum_{i=0}^{n} \alpha_i=1$ and $0 \equiv z_0 <z_1<\ldots <z_{n} \equiv 1$, while $\delta(\cdot)$ is the \textit{(Dirac) delta function} and $n \in \mathbb{Z}_{>0}$.
\end{lemma}

\begin{lemma}[$D$-density]~\label{lemma:D_dens}
The $D$-density of a B-DMC admits the form
\begin{equation}
    \mathfrak{a}(z)=\sum_{j=-n}^{n}\beta_j\delta(z-z_j),
\end{equation}
where $\beta_j\in[0,1]$, $\sum_{j=-n}^{n} \beta_j=1$, $-1 \equiv z_{-n}<\ldots < 0 \equiv z_0 <\ldots <z_{n} \equiv 1$, and $z_{-j}=-z_j$ for $j=-n,\ldots,n\in \mathbb{Z}_{>0}$. The relation of Dirac masses between the $D$- and $|D|$-density is given by
\begin{align}
    \beta_{\pm i} &= (1\pm z_i)/2\cdot \alpha_i, \quad i=1,\ldots, n-1, \label{eq:beta_alpha_1_n-1} \\
    \beta_{-n} &= 0, \quad \beta_0 = \alpha_0, \quad \beta_n = \alpha_n.
\end{align}
\end{lemma}
\begin{proof}
Under symmetry, the $D$-density is determined by the corresponding $|D|$-density, i.e., $|\mathfrak{a}|(z)=\mathfrak{a}(z)+\mathfrak{a}(-z)$, $\forall z\in(0,1]$ and $|\mathfrak{a}|(0)=\mathfrak{a}(0)$. For a symmetric $\mathfrak{a}$, $\mathfrak{a}(z)/\mathfrak{a}(-z)=(1+z)/(1-z)$, $\forall z\in[0,1]$, cf.~\cite[(4.13)]{richardson2008modern}. Combining both, we have $\mathfrak{a}(\pm z)=(1\pm z)/2\cdot |\mathfrak{a}|(z)$, $\forall z\in(0,1]$, thus, $\mathfrak{a}(1)=|\mathfrak{a}|(1)$ and $\mathfrak{a}(-1)=0$. By Lemma~\ref{lemma:absD_dens}, we obtain $\beta_{ \pm i}=\left(1 \pm z_i\right) / 2 \cdot \alpha_i$ for $i=1, \ldots, n-1$, while $\beta_{-n}=0$, $\beta_0=\alpha_0$, and $\beta_n=\alpha_n$.
\end{proof}

\subsection{Variable- and check-domain transforms}

This subsection unifies the variable- and check-domain transforms in $D$- and $|D|$-densities and characterizes their strips of analyticity (SOA). The variable- resp. check-domain transforms correspond to a \textit{bilateral} resp. \textit{unilateral Laplace transform} in the $L$- resp. $G$-densities, cf. Definitions~\ref{def:var_node_double_LT} and~\ref{def:check_node_double_LT}. These two definitions, together with Lemma~\ref{lemma:var_node_double_LT} and~\ref{lemma:check_node_double_LT} do not require symmetry; this is the only exception in this paper. A summary is provided in Table~\ref{tab:transforms_summary}. To avoid ambiguity, the term \textit{Fourier transform} refers to its standard definition. Throughout the remainder of this paper, let $X\sim \mathrm{a}$, $Z \triangleq \tanh (X / 2) \sim \mathfrak{a}$, and $(S,Y)\sim a$ where $S\triangleq (1-\mathfrak{H}(X))/2\in\{0,1\}$ under the standard mapping~\cite[p.175]{richardson2008modern}, and $Y\triangleq \ln\coth (|X|/2)$; moreover, $|Z|=\tanh (|X| / 2) \sim|\mathfrak{a}|$. 

\begin{table}[h!]
    \centering
    \begin{tabular}{|c||c|c|}
     \hline
     & Variable-domain transform & Check-domain transform ($\mu=0$) \\
     \hline
        $L$-density $\mathrm{a}$ & $\mathcal{F}_{\mathrm{a}}(s)$ & -- \\
     \hline
        $D$-density $\mathfrak{a}$ & $\mathcal{F}_{\mathfrak{a}}(s)$ & $\mathcal{G}_{\mathfrak{a}}(0, \nu)$ \\
     \hline
        $|D|$-density $|\mathfrak{a}|$ & $\mathcal{F}_{|\mathfrak{a}|}(s)$ & $\mathcal{G}_{|\mathfrak{a}|}(0, \nu)$ \\
     \hline
        $G$-density $a$ & -- & $\mathcal{G}_a(0, \nu)$ \\
     \hline
        SOA & $\Re s\in(0,1)$ & $\Re \nu > 0$ \\
     \hline
    \end{tabular}
    \vspace{1mm}
    \caption{Comparison of variable- and check-domain transforms across $L$-, $D$-, $|D|$-, and $G$-densities.}
    \label{tab:transforms_summary}
\end{table}

\begin{definition}[Variable-domain transform{\cite[Definition 4.52]{richardson2008modern}}]\label{def:var_node_double_LT}
The variable-domain transform of $\mathrm{a}$ is
\begin{equation}~\label{eq:var_node_double_LT}
    \mathcal{F}_{\mathrm{a}}(s) \triangleq \int_{-\infty}^{+\infty} \mathrm{a}(x) \mathrm{e}^{-s x} \mathrm{~d} x,
\end{equation}
for $s \in \mathbb{C}$ where the integral exists.
\end{definition}

\begin{lemma}\label{lemma:var_node_double_LT}
Let $\mathfrak{a}$ be the $D$-density corresponding to $\mathrm{a}$, then $\mathcal F_{\mathrm{a}}(s)=\int_{-1}^1 \mathfrak{a}(z)\left(\frac{1-z}{1+z}\right)^s \mathrm{~d} z \triangleq\mathcal{F}_{\mathfrak{a}}(s)$.
\end{lemma}

\begin{definition}[Check-domain transform{\cite[Definition 4.53]{richardson2008modern}}]~\label{def:check_node_double_LT}
The check-domain transform of $a$ is
\begin{equation}~\label{eq:check_node_double_LT}
    \mathcal{G}_a(\mu, \nu)\triangleq\int_{0}^{+\infty}\sum_{\sigma \in\{0,1\}}a(\sigma,y)\mathrm{e}^{-\mu \sigma-\nu y}\mathrm{~d} y,
\end{equation}
for $\mu\in\{0,\mathrm{i}\pi\}$ and $\nu\in\mathbb{C}$ where the integral exists.
\end{definition}
\begin{lemma}\label{lemma:check_node_double_LT}
Let $\mathfrak{a}$ be the $D$-density corresponding to $a$, then $\mathcal{G}_a(\mu, \nu) = \int_{-1}^1 \mathfrak{a}(z) \exp \left(-\mu \frac{1-\mathfrak{H}(z)}{2}\right)|z|^\nu \mathrm{~d} z \triangleq\mathcal{G}_{\mathfrak{a}}(\mu, \nu)$. In particular, $\mathcal{G}_{\mathfrak{a}}(0, \nu)=\int_{-1}^1 \mathfrak{a}(z)|z|^\nu \mathrm{~d} z$.
\end{lemma}
\begin{remark}
The expression in the $D$-density is more convenient than that in the $G$-density, since: \textit{i}) $\tanh(\cdot)$ preserves the sign, i.e., $\mathfrak{H}(X)=\mathfrak{H}(Z)$; \textit{ii}) and by its oddness, $\mathrm{e}^{-Y}=|Z|$. Hence, $\mathcal{G}_{\mathfrak{a}}(0, \nu)$ is independent of $\mathfrak{H}(\cdot)$, which differs from the $D$-$k$\textit{-moment}~\cite[Definition 4.57]{richardson2008modern} with $\mathcal{G}_{\mathfrak{a}}(0,k)=\mathfrak{D}_k(\mathfrak{a})$ iff $k$ is even.
\end{remark}

We next present the variable-domain transforms in the $D$- and $|D|$-density within the SOA, i.e., $s\in\mathbb{C}$ with real part $\Re s\in(0,1)$.

\begin{lemma}~\label{lemma:var_node_holo}
$\mathcal{F}_{\mathfrak{a}}(s)$ is holomorphic for every $\Re s\in(0,1)$, and by symmetry, $\mathcal{F}_{\mathfrak{a}}(s)=\mathcal{F}_{\mathfrak{a}}(1-s)$, cf.~\cite[p.199]{richardson2008modern}.
\end{lemma}
\begin{proof}
Please refer to Section~\ref{sec:var_node_holo}.
\end{proof}

\begin{corollary}~\label{corol:var_node_double_LT}
Let $|\mathfrak{a}|$ be the $|D|$-density corresponding to $\mathfrak{a}$, then $\mathcal{F}_{\mathfrak{a}}(s)=\int_0^1|\mathfrak{a}|(z)\Bigl(\frac{1-z}{2}\Bigl(\frac{1+z}{1-z}\Bigr)^s+\frac{1+z}{2}\Bigl(\frac{1-z}{1+z}\Bigr)^s\Bigr)\mathrm{~d} z\triangleq\mathcal{F}_{|\mathfrak{a}|}(s)$, for every $\Re s\in(0,1)$.
\end{corollary}

\begin{corollary}~\label{corol:var_node_half}
The \textit{Bhattacharyya parameter} is obtained by letting $s=1/2$, cf.~\cite[Proof of Lemma 4.64]{richardson2008modern}, i.e., $\mathfrak{B}(|\mathfrak{a}|)=\mathcal{F}_{|\mathfrak{a}|}(1/2)=\int_0^1|\mathfrak{a}|(z) \sqrt{1-z^2} \mathrm{~d} z$.
\end{corollary}

We now extend this analysis to the check-domain transform originally defined in the $G$-density, relating it back to the $D$- and $|D|$-densities within the SOA, i.e., $\nu\in\mathbb{C}$ with real part $\Re \nu>0$.

\begin{lemma}\label{lemma:check_node_holo}
$\mathcal{G}_{\mathfrak{a}}(0, \nu)$ is holomorphic for every $\Re \nu>0$, and by symmetry, $\mathcal{G}_{\mathfrak{a}}(\mathrm{i}\pi, \nu)=\mathcal{G}_{\mathfrak{a}}(0, \nu+1)$, cf.~\cite[p.200]{richardson2008modern}.
\end{lemma}
\begin{proof}
Please refer to Section~\ref{sec:check_node_holo}.
\end{proof}
Lemma~\ref{lemma:check_node_holo} shows that, under symmetry, evaluating at $\mu=\mathrm{i}\pi$ shifts the $\nu$-argument of $\mathcal{G}_{\mathfrak{a}}(0, \nu)$ by $1$. Hence, we restrict our attention to $\mu=0$ in the sequel.
\begin{corollary}\label{corol:check_node_double_LT}
Let $|\mathfrak{a}|$ be the $|D|$-density corresponding to $\mathfrak{a}$, then $\mathcal{G}_{\mathfrak{a}}(0, \nu)=\int_0^1 |\mathfrak{a}|(z) z^\nu \mathrm{d} z \triangleq\mathcal{G}_{|\mathfrak{a}|}(0, \nu)$, for every $\Re \nu>0$.
\end{corollary}


\section{Component evolution under the \texorpdfstring{$\boxast$}{boxast}-convolution}~\label{sec:check_conv_domain}
This section considers the $|D|$-density on the check-domain. We introduce the entire counterpart $\mathscr{G}_{|\mathfrak{a}|}(\nu)$ of $\mathcal{G}_{|\mathfrak{a}|}(0, \nu)$, relate it to $\mathcal{G}_{|\mathfrak{a}|}(0, \nu)$ and boundary atoms, and analyze the effect of the check-domain $\boxast$-convolution on those components.

\begin{lemma}~\label{lemma:f_a_analy_cont}
Define
\begin{equation}
    f_{|\mathfrak{a}|}(z)\triangleq \begin{cases}
        |\mathfrak{a}|(z)\cdot z, & 0 < z < 1, \\
        0, & z \geq 1,
    \end{cases}
\end{equation}
Its \textit{Mellin transform}~\cite[(3.1.1)]{paris2001asymptotics}
\begin{equation}~\label{eq:atom2G}
    \mathcal{M}\left\{f_{|\mathfrak{a}|}\right\}(\nu)\triangleq \mathscr{G}_{|\mathfrak{a}|}(\nu)=\sum_{i=1}^{n-1}\alpha_i z_i^\nu,
\end{equation}
is an \textit{entire} function on $\mathbb{C}$.
\end{lemma}

\begin{proof}
Please refer to Section~\ref{sec:f_a_analy_cont}.
\end{proof}

\begin{corollary}\label{corol:G2G_connect}
$\mathcal{G}_{|\mathfrak{a}|}(0, \nu)=\mathscr{G}_{|\mathfrak{a}|}(\nu)+\alpha_0 \cdot 0^\nu + \alpha_{n}$, $\forall \Re \nu>0$; while $1\equiv\mathcal{G}_{|\mathfrak{a}|}(0, 0)=\mathscr{G}_{|\mathfrak{a}|}(0)+\alpha_0 + \alpha_{n}$.
\end{corollary}
\begin{remark}
Here, $0^\nu\triangleq \lim_{\delta\to 0^+} \delta^\nu$ defines the \textit{Hausdorff dimension function}~\cite[Definition 1.2]{sudland2004mellin}, i.e., $0^\nu=0$ for $\Re \nu>0$ and $0^\nu=1$ for $\nu=0$. Hence, $\mathcal{G}_{|\mathfrak{a}|}(0, \nu)=\mathscr{G}_{|\mathfrak{a}|}(\nu)+ \alpha_{n}$, $\forall \Re \nu>0$; see blue arrows in Fig.~\ref{fig:entities_domains}.
\end{remark}

Next, we characterize the BEC and BSC via $\mathcal{G}_{|\mathfrak{a}|}(0, \nu)$, $\mathscr{G}_{|\mathfrak{a}|}(\nu)$ and boundary atoms in the $|D|$-density; see Table \ref{tab:channel_comp}.

\begin{table}[h!]
\begin{center}
\begin{tabular}{|c || c c c c c|} 
 \hline
 $W$ & $|\mathfrak{a}|(z)$ & $\mathbbm{1}_{\{\alpha_0 > 0\}}$ & $\mathbbm{1}_{\{\alpha_n > 0\}}$ & $\mathcal{G}_{|\mathfrak{a}|}(0, \nu), \ \forall\Re\nu>0$ & $\mathscr{G}_{|\mathfrak{a}|}(\nu), \ \forall\nu\in\mathbb{C}$ \\ [0.5ex] 
 \hline
 $W_{\operatorname{BEC}(\varepsilon)}$ & $\varepsilon \delta(z)+\bar{\varepsilon} \delta(z-1)$ & 1 & 1 & $\bar{\varepsilon}$ & 0 \\ 
 \hline
 $W_{\operatorname{BSC}(p)}$ & $\delta(z-(1-2p))$ & 0 & 0 & $(1-2p)^\nu$ & $(1-2p)^\nu$ \\
 \hline
\end{tabular}
\vspace{1mm}
\caption{Comparison of the BEC and BSC in terms of $|\mathfrak{a}|(z)$, $\mathcal{G}_{|\mathfrak{a}|}(0, \nu)$ and $\mathscr{G}_{|\mathfrak{a}|}(\nu)$}
\label{tab:channel_comp}
\end{center}
\end{table}

\begin{exmp}[$W_{\operatorname{BEC}(\varepsilon)}$ in the $|D|$-density]~\label{exmp:bec_G}
$|\mathfrak{a}|(z)=\varepsilon \delta(z)+\bar{\varepsilon} \delta(z-1)$, cf.~\cite[Example 4.23]{richardson2008modern}, i.e., $n=1$, $\alpha_0=\varepsilon$ and $\alpha_1=\barepsilon$. $\mathcal{G}_{|\mathfrak{a}|}(0, \nu)=\bar{\varepsilon}$, $\forall \Re \nu>0$. $\mathscr{G}_{|\mathfrak{a}|}(\nu)=0$, $\forall \nu\in\mathbb{C}$.
\end{exmp}
\begin{exmp}[$W_{\operatorname{BSC}(p)}$ in the $|D|$-density]
$|\mathfrak{a}|(z)=\delta(z-(1-2p))$, cf.~\cite[Example 4.24]{richardson2008modern}, i.e., $n=2$, $\alpha_0=\alpha_2=0$ and $z_1=1-2p\in(0,1)$ with $\alpha_1=1$. $\mathcal{G}_{|\mathfrak{a}|}(0, \nu)=\mathscr{G}_{|\mathfrak{a}|}(\nu)$, $\forall \Re \nu>0$. $\mathscr{G}_{|\mathfrak{a}|}(\nu)=(1-2p)^\nu$, $\forall \nu\in\mathbb{C}$.
\end{exmp}

\begin{remark}
For the BEC, only the boundary atom at $z_n\equiv 1$ contributes to $\mathcal{G}_{|\mathfrak{a}|}(0,\nu)$, whereas for the BSC, which possesses no boundary atoms, $\mathcal{G}_{|\mathfrak{a}|}(0,\nu)$ is completely determined by the interior contribution $\mathscr{G}_{|\mathfrak{a}|}(\nu)$ within its SOA.
\end{remark}

Excluding atoms at $z_0\equiv 0$ and $z_n\equiv 1$ from $\mathscr{G}_{|\mathfrak{a}|}(\nu)$ may seem inconvenient, but is necessary. Channel polarization in DE is driven by alternating $\boxast$- and $\varoast$-convolutions~\cite[(2.11)]{korada2009polar}. The proposed CE-framework aims to exploit the multiplicative properties of the variable- resp. check-domain transform under the $\varoast$- resp. $\boxast$-convolution~\cite[p.199]{richardson2008modern}. In fact, since polarization only requires self-convolution in each domain, e.g., $\mathfrak{a}^{\boxast}=\mathfrak{a} \boxast \mathfrak{a}$ in the check-domain~\cite[(2.9)]{korada2009polar}, such multiplicativity reduces to squaring, e.g., the check-domain transform of $\mathfrak{a}^{\boxast}$ is simply the square of that of $\mathfrak{a}$, cf.~\eqref{eq:sqr_G}. Thus, the remaining challenge is to draw bilateral conversions between $\mathcal{F}_{\mathfrak{a}}(s)$ (or equivalently $\mathcal{F}_{|\mathfrak{a}|}(s)$, cf. Corollary~\ref{corol:var_node_double_LT}) and $\mathcal{G}_{\mathfrak{a}}(0,\nu)$ (or equivalently $\mathcal{G}_{|\mathfrak{a}|}(0,\nu)$, cf. Corollary~\ref{corol:check_node_double_LT}) across the check- and variable-domains. A natural candidate is the \textit{Parseval formula for the Mellin transform}~\cite[(3.1.11)]{paris2001asymptotics} which provides an integral relation between transform pairs and thereby leads to a \textit{Mellin-Barnes integral representation}, which can then be evaluated by contour deformation via the residue theorem; see~\cite[Sec. 4]{fikioris2022mellin} for details. However, this approach requires sufficient analyticity of the integrand, which is violated in the presence of endpoint atom. 

Indeed, the Dirac mass $\alpha_0$ at $z_0\equiv 0$ in the $|D|$-density introduces the term $\alpha_0\cdot 0^\nu$ in $\mathcal{G}_{|\mathfrak{a}|}(0, \nu)$, cf. Corollary~\ref{corol:check_node_double_LT}. When converting $\mathcal{G}_{|\mathfrak{a}|}(0, \nu)$ to $\mathcal{F}_{|\mathfrak{a}|}(s)$ within its SOA, i.e., $\Re s \in(0,1)$ via a Mellin–Barnes representation, this term is well defined only for $\Re\nu>0$ and does not admit a single-valued analytic continuation to $\Re\nu\le 0$, thereby obstructing contour deformation across $\Re\nu=0$. To circumvent this issue, we instead work with the entire counterpart $\mathscr{G}_{|\mathfrak{a}|}(\nu)$ to $\mathcal{F}_{|\mathfrak{a}|}(s)$, and account for its contribution $\alpha_0$ separately, cf.~\eqref{eq:G2F_convert}. 

Likewise, the term $\beta_n \cdot 0^s$ in $\mathcal{F}_{\mathfrak{a}}(s)$ introduced by the Dirac mass $\beta_n$ at $z_n\equiv 1$ in the $D$-density, cf. Lemma~\ref{lemma:var_node_double_LT}, would obstruct contour deformation across $\Re s=0$, when converting $\mathcal{F}_{\mathfrak{a}}(s)$ to $\mathcal{G}_{\mathfrak{a}}(0, \nu)$ within its SOA, i.e., $\Re \nu>0$. Recall from Lemma~\ref{lemma:D_dens} that this endpoint corresponds to $\alpha_n=\beta_n$ at $z_n\equiv 1$ in the $|D|$-density. Accordingly, we work with the entire counterpart $\mathscr{F}_{\mathfrak{a}}(s)$ of $\mathcal{F}_{\mathfrak{a}}(s)$, and account for its contribution $\alpha_n$ separately in the reverse conversion, cf.~\eqref{eq:F2G_convert}.

Therefore, both boundary atoms at $z_0\equiv 0$ and $z_n\equiv 1$ are removed from $\mathscr{G}_{|\mathfrak{a}|}(\nu)$ (and later from $\mathscr{F}_{\mathfrak{a}}(s)$ in Lemma~\ref{lemma:f_a_w_analy_cont}) and treated separately. As a result, the relation between $\mathcal{G}_{|\mathfrak{a}|}(0, \nu)$ and $\mathscr{G}_{|\mathfrak{a}|}(\nu)$ must account for the contributions of the Dirac masses $\alpha_0$ and $\alpha_n$ within the SOA, cf. Corollary~\ref{corol:G2G_connect}.

Next, we investigate the effect of the $\boxast$-convolution on $\mathcal{G}_{|\mathfrak{a}|}(0, \nu)$ within its SOA, its entire counterpart $\mathscr{G}_{|\mathfrak{a}|}(\nu)$ on the whole check-domain $\nu\in\mathbb{C}$, and Dirac masses of the boundary atoms at $z_0\equiv 0$ and $z_n\equiv 1$ in the $|D|$-density, respectively.

\begin{lemma}~\label{lemma:G_comp_evo}
Under the check-domain $\boxast$-convolution, 
\begin{gather}
    \mathcal{G}_{|\mathfrak{a}^{\boxast}|}(0, \nu)=\mathcal{G}^2_{|\mathfrak{a}|}(0, \nu), \quad \forall \Re \nu >0, ~\label{eq:sqr_G} \\
    \mathscr{G}_{|\mathfrak{a}^{\boxast}|}(\nu)=\mathscr{G}^2_{|\mathfrak{a}|}(\nu)+2 \alpha_n \mathscr{G}_{|\mathfrak{a}|}(\nu), \quad \forall \nu \in\mathbb{C},~\label{eq:evo_G} \\
    \alpha^\boxast_0 = 2\alpha_0-\alpha_0^2,~\label{eq:evo_G_alpha_0} \\
    \alpha^\boxast_{n} = \alpha^2_{n}.~\label{eq:evo_G_alpha_n}
\end{gather}
\end{lemma}

\begin{proof}
Please refer to Section~\ref{sec:G_comp_evo}.
\end{proof}
\begin{remark}
$n^\boxast$ denotes the number of atoms in the $|D|$-density after the $\boxast$-convolution. For brevity, we write $\alpha^\boxast_{n}\equiv \alpha^\boxast_{n^\boxast}$.
\end{remark}

Next, we illustrate the effect of $\boxast$-convolutions on these components, using the BEC and BSC as representative examples.

\begin{exmp}[$W_{\operatorname{BEC}(\varepsilon)}$ under $\boxast$-convolutions]
Under one $\boxast$-convolution, $\mathcal{G}_{|\mathfrak{a}^{\boxast}|}(0, \nu)=\bar{\varepsilon}^2$, $\mathscr{G}_{|\mathfrak{a}^{\boxast}|}(\nu)=0$, $\alpha^\boxast_0 = 2\varepsilon-\varepsilon^2$ and $\alpha^\boxast_{1} = \bar{\varepsilon}^2$. Thus, a single $\boxast$-convolution maps $W_{\operatorname{BEC}(\varepsilon)}$ to another $W_{\operatorname{BEC}(\varepsilon^{\boxast})}$ with erasure $\varepsilon^{\boxast}=1-\barepsilon^2$. Additionally, $n^{\boxast}=1$, $z_0^{\boxast}\equiv z_0\equiv 0$, and $z_1^{\boxast}\equiv z_1=1$, trivially. By induction, we obtain an iteration formula of the erasure $\varepsilon^{\boxast m}$ of $W_{\operatorname{BEC}(\varepsilon)}$ under the $m$-fold $\boxast$-convolution, i.e., $\varepsilon^{\boxast (m+1)} = 1-(\bar{\varepsilon}^{\boxast m})^2$. Therefore, as $m\to \infty$,
\begin{equation*}
    \varepsilon^{\boxast m} = 1-\barepsilon^{2^m} \to 1^-,
\end{equation*}
which implies that $W_{\operatorname{BEC}(\varepsilon^{\boxast})}$ converges \textit{double-exponentially} in $m$ to a fully erasing channel.
\end{exmp}
\begin{exmp}[$W_{\operatorname{BSC}(p)}$ under $\boxast$-convolutions]~\label{emp:bsc_check}
Under one $\boxast$-convolution, $\mathscr{G}_{|\mathfrak{a}^{\boxast}|}(\nu)=\mathscr{G}^2_{|\mathfrak{a}|}(\nu)=(1-2p)^{2\nu}$, i.e., $\alpha^\boxast_0 =\alpha^\boxast_{2} = 0$. In particular, we have $1-2 p^{\boxast}=(1-2 p)^2=z_1^2$, which implies that one $\boxast$-convolution upon $W_{\operatorname{BSC}(p)}$ results in another $W_{\operatorname{BSC}(p^{\boxast})}$ with crossover $p^{\boxast}= 2p\bar{p}$. Additionally, $n^{\boxast}=2$ and $z_1^{\boxast}=z_1^2$. By induction, we obtain iteration formulae of the crossover $p^{\boxast m}$ and $z_1^{\boxast m}$ of $W_{\operatorname{BSC}(p^{\boxast m})}$ under the $m$-fold $\boxast$-convolution, i.e., $p^{\boxast (m+1)} = 2p^{\boxast m} \bar{p}^{\boxast m}$ and $z_1^{\boxast (m+1)} = (z_1^{\boxast m})^2$, respectively. Therefore, as $m\to \infty$,
\begin{align*}
    p^{\boxast m} &= 1/2\cdot(1-z_1^{2^m})\to 1/2^-, \\
    z_1^{\boxast m} &= z_1^{2^m} \to 0^+,
\end{align*}
which implies that $W_{\operatorname{BSC}(p^{\boxast m})}$ converges \textit{double-exponentially} in $m$ to a purely noisy BSC.
\end{exmp}

\section{Component evolution under the \texorpdfstring{$\varoast$}{varoast}-convolution}~\label{sec:var_conv_domain}
This section considers the $D$-density on the variable-domain. We introduce the entire counterpart $\mathscr{F}_{\mathfrak{a},i}(s)$ with $i\in\{1,2\}$ of $\mathcal{F}_{\mathfrak{a}}(s)$, relate them to $\mathcal{F}_{\mathfrak{a}}(s)$ and the boundary atoms, and analyze the effect of the variable-domain $\varoast$-convolution on those components. The entity $\mathscr{F}_{\mathfrak{a}}(s)$ is based on the new r.v. $\mathfrak{W}\triangleq (1-Z)/(1+Z)=\mathrm{e}^{-X}$, linking $X\sim \mathrm{a}$ and $Z\sim \mathfrak{a}$, which will be used throughout the remainder of the paper.

\begin{lemma}~\label{lemma:f_a_w_analy_cont}
Define
\begin{equation}
    f_{\mathfrak{a}}(w)\triangleq \begin{cases}
        \frac{2w}{(1+w)^2} \cdot \mathfrak{a}\left(\frac{1-w}{1+w}\right), \ w\in (0,1) \cup (1,\infty), \\
        0, \ w=1,
    \end{cases}
\end{equation}
of which the Mellin transform
\begin{equation}
    \mathcal{M}\{f_{\mathfrak{a}}\}(s)\triangleq \mathscr{F}_{\mathfrak{a}}(s) =\sum_{i=1}^{n-1} \beta_i w_i^s+\sum_{i=1}^{n-1} \beta_{-i} w_{-i}^s,
\end{equation}
where $\beta_{\pm i}=(1\pm z_i)/2\cdot \alpha_i$, and $w_{\pm i}=(1\mp z_i)/(1\pm z_i)$ for every $i=1,\ldots,n-1$, is an \textit{entire} function on $\mathbb{C}$.
\end{lemma}

\begin{proof}
Please refer to Section~\ref{sec:f_a_w_analy_cont}.
\end{proof}

\begin{corollary}~\label{corol:f_a_1_2_cont}
$f_{\mathfrak{a}}(w)$ admits the decomposition $f_{\mathfrak{a}}(w)=f_{\mathfrak{a},1}(w)+f_{\mathfrak{a},2}(w)$ with
\begin{equation}
    f_{\mathfrak{a},1}(w) \triangleq \begin{cases}
        \frac{2w}{(1+w)^2} \cdot \mathfrak{a}\left(\frac{1-w}{1+w}\right), \ w\in (0,1), \\
        0, \quad w\in[1,\infty);
    \end{cases} \quad \text{and} \quad
    f_{\mathfrak{a},2}(w) \triangleq \begin{cases}
        0, \quad w\in (0,1], \\
        \frac{2w}{(1+w)^2} \cdot \mathfrak{a}\left(\frac{1-w}{1+w}\right), \ w\in(1,\infty).
    \end{cases}
\end{equation}
The Mellin transform of $f_{\mathfrak{a}}(w)$ admits the decomposition $\mathscr{F}_{\mathfrak{a}}(s)=\mathscr{F}_{\mathfrak{a},1}(s)+\mathscr{F}_{\mathfrak{a},2}(s)$, where
\begin{equation}~\label{eq:atom2F12}
    \mathscr{F}_{\mathfrak{a},1}(s) \triangleq \mathcal{M}\{f_{\mathfrak{a},1}\}(s)=\sum_{i=1}^{n-1} \beta_i w_i^s, \quad \text{and} \quad
    \mathscr{F}_{\mathfrak{a},2}(s) \triangleq \mathcal{M}\{f_{\mathfrak{a},2}\}(s)=\sum_{i=1}^{n-1} \beta_{-i} w_{-i}^s,
\end{equation}
with $w_i\in(0,1)$ and $w_{-i}\in(1,\infty)$ for every $i=1,\ldots,n-1$. Both $\mathscr{F}_{\mathfrak{a},1}(s)$ and $\mathscr{F}_{\mathfrak{a},2}(s)$ are entire functions on $\mathbb{C}$.
\end{corollary}

\begin{corollary}~\label{corol:F12_sym}
$\mathscr{F}_{\mathfrak{a},2}(s)=\mathscr{F}_{\mathfrak{a},1}(1-s)$, $\forall s \in \mathbb{C}$.
\end{corollary}

\begin{proof}
Please refer to Section~\ref{sec:F12_sym}.
\end{proof}


The decomposition of $\mathscr{F}_{\mathfrak{a}}(s)$ into $\mathscr{F}_{\mathfrak{a},1}(s)$ and $\mathscr{F}_{\mathfrak{a},2}(s)$ with symmetry $\mathscr{F}_{\mathfrak{a},2}(s)=\mathscr{F}_{\mathfrak{a},1}(1-s)$, cf. Corollary~\ref{corol:F12_sym}, simplifies the CE description in Theorem~\ref{thm:G2F_convert2} by reducing the analysis to a single entire function, e.g., $\mathscr{F}_{\mathfrak{a},1}(s)$. In analogy with Lemma~\ref{lemma:f_a_analy_cont}, the relation between $\mathcal{F}_{\mathfrak{a}}(s)$ and $\mathscr{F}_{\mathfrak{a}}(s)$ within the SOA must account for the exclusion of boundary atoms.

\begin{corollary}\label{corol:F_trans_parts}
$\mathcal{F}_{\mathfrak{a}}(s)=\mathscr{F}_{\mathfrak{a},1}(s)+\mathscr{F}_{\mathfrak{a},2}(s)+\beta_0 +\beta_n \cdot 0^s$, $\forall \Re s \in(0,1)$. In particular, $\mathfrak{B}(\mathfrak{a})=2\mathscr{F}_{\mathfrak{a},1}(1/2)+\alpha_0$.
\end{corollary}

\begin{remark}
Since $\beta_0=\alpha_0$, cf. Lemma~\ref{lemma:D_dens}, and $0^s=0$, $\forall\Re s>0$, then $\mathcal{F}_{\mathfrak{a}}(s)=\mathscr{F}_{\mathfrak{a},1}(s)+\mathscr{F}_{\mathfrak{a},2}(s)+\alpha_0$, $\forall \Re s\in(0,1)$; see green arrows in Fig.~\ref{fig:entities_domains}. 
\end{remark}


Next, we characterize the BEC and BSC via $\mathcal{F}_{\mathfrak{a}}(s)$, $\mathscr{F}_{\mathfrak{a},1}(s)$, $\mathscr{F}_{\mathfrak{a},2}(s)$ and boundary atoms in the $|D|$-density; see Table \ref{tab:channel_comp_merged}.

\begin{table}[h!]
\centering
\small
\setlength{\tabcolsep}{4pt}
\renewcommand{\arraystretch}{1.15}
\resizebox{\linewidth}{!}{%
\begin{tabular}{|c || c c c c c c|}
\hline
$W$ & $\mathfrak{a}(z)$ & $\mathbbm{1}_{\{\beta_{0} > 0\}}$ & $\mathbbm{1}_{\{\beta_n > 0\}}$ & $\mathcal{F}_{\mathfrak{a}}(s)$, $\forall \Re s\in(0,1)$ & $\mathscr{F}_{\mathfrak{a},1}(s)$, $\forall s\in\mathbb{C}$ & $\mathscr{F}_{\mathfrak{a},2}(s)$, $\forall s\in\mathbb{C}$ \\
\hline
$W_{\operatorname{BEC}(\varepsilon)}$
& $\varepsilon \delta(z)+\bar{\varepsilon}\delta(z-1)$
& $1$ & $1$
& $\varepsilon$
& $0$
& $0$ \\
\hline
$W_{\operatorname{BSC}(p)}$
& $p\,\delta(z+(1-2p))+\bar{p}\,\delta(z-(1-2p))$
& $0$ & $0$
& $\mathscr{F}_{\mathfrak{a},1}(s)+\mathscr{F}_{\mathfrak{a},2}(s)$
& $\bar{p}^{\,1-s}p^s$
& $p^{\,1-s}\bar{p}^{\,s}$ \\
\hline
\end{tabular}%
}
\vspace{1mm}
\caption{Comparison of symmetric channels in terms of $\mathfrak{a}(z)$, $\mathcal{F}_{\mathfrak{a}}(s)$, $\mathscr{F}_{\mathfrak{a},1}(s)$ and $\mathscr{F}_{\mathfrak{a},2}(s)$}
\label{tab:channel_comp_merged}
\end{table}

\begin{exmp}[$W_{\operatorname{BEC}(\varepsilon)}$ in the $D$-density]~\label{exmp:bec_F12}
$\mathfrak{a}(z)=\varepsilon \delta(z)+\bar{\varepsilon} \delta(z-1)$. $\mathcal{F}_{\mathfrak{a}}(s)=\beta_0 w_0^s+\beta_n w_n^s=\varepsilon\cdot 1^s+\bar{\varepsilon}\cdot 0^s=\varepsilon$, $\forall \Re s\in(0,1)$, while $\Fai(s)=\Faii(s)=0$, $\forall s\in\mathbb{C}$.
\end{exmp}
\begin{exmp}[$W_{\operatorname{BSC}(p)}$ in the $D$-density]
$\mathfrak{a}(z)=p\cdot\delta(z+(1-2p))+\bar{p}\cdot\delta(z-(1-2p))$. $\mathcal{F}_{\mathfrak{a}}(s)=\mathscr{F}_{\mathfrak{a},1}(s)+\mathscr{F}_{\mathfrak{a},2}(s)=\bar{p}^{1-s}p^s+p^{1-s}\bar{p}^s$, $\forall \Re s\in(0,1)$, while $\beta_0=\beta_n=0$. Additionally, $\mathfrak{B}(\mathfrak{a})=\mathcal{F}_{\mathfrak{a}}(1/2)=2\sqrt{p\bar{p}}$, cf.~\cite[Proof of Lemma 4.65]{richardson2008modern}.
\end{exmp}

\begin{remark}
For the BEC, only the boundary atom at $z_0\equiv 0$ contributes to $\mathcal{F}_{\mathfrak{a}}(s)$, whereas for the BSC, which possesses no boundary atoms, $\mathcal{F}_{\mathfrak{a}}(s)$ is entirely determined by the interior parts $\mathscr{F}_{\mathfrak{a},1}(s)$ and $\mathscr{F}_{\mathfrak{a},2}(s)$ within the SOA.
\end{remark}

Next, we investigate the effect of the $\varoast$-convolution on these components in analogy with Lemma~\ref{lemma:G_comp_evo}.

\begin{lemma}~\label{lemma:F_comp_evo}
Under the variable-domain $\varoast$-convolution, 
\begin{align}
    \mathcal{F}_{\mathfrak a^{\varoast}}(s) &= \mathcal{F}^2_{\mathfrak{a}}(s), \quad  \forall \Re s\in(0,1), ~\label{eq:sqr_F} \\
    \mathscr{F}_{\mathfrak{a}^\varoast,i}(s) &=\mathscr{F}^2_{\mathfrak{a}, i}(s)+2 \alpha_0 \mathscr{F}_{\mathfrak{a}, i}(s)+2 \mathscr{F}_{\psi, i}(s), \forall s\in\mathbb{C},  \label{eq:evo_F1} \\
    \alpha_0^{\varoast} &= \alpha_0^2+2 \mathscr{F}_{\varphi}, \label{eq:evo_F_alpha_0} \\
    \alpha_n^{\varoast} &= 2\alpha_n-\alpha_n^2, \label{eq:evo_F_alpha_n}
\end{align}
where $i \in\{1,2\}$, $\mathscr{F}_{\psi,1}(s) \triangleq \sum_{1 \le j<i \le n-1} \beta_i \beta_{-j}(w_i/w_j)^s$, $\mathscr{F}_{\psi,2}(s) \triangleq \sum_{1 \le i<j \le n-1} \beta_i \beta_{-j}(w_i/w_j)^s = \mathscr{F}_{\psi,1}(1-s)$, and $\mathscr{F}_{\varphi} \triangleq \sum_{i=1}^{n-1} \beta_i \beta_{-i}$. All functions $\mathscr{F}_{\mathfrak{a}^\varoast,i}(s)$, $\mathscr{F}_{\psi,i}(s)$, $i \in\{1,2\}$, admit analytic continuation to entire functions on $\mathbb{C}$.
\end{lemma}

\begin{proof}
Please refer to Section~\ref{sec:F_comp_evo}.
\end{proof}

By isolating entire components $\Ga(\nu)$ and $\mathscr{F}_{\mathfrak{a},i}(s)$ from $\mathcal{G}_{|\mathfrak{a}|}(0, \nu)$ and $\mathcal{F}_{\mathfrak{a}}(s)$, we highlight the asymmetry between the $\boxast$- and $\varoast$-convolutions: although each squares the domain transform within its own SOA, cf.~\eqref{eq:sqr_G} and~\eqref{eq:sqr_F}, their effects on entire components and boundary atoms are not reciprocal, cf.~\eqref{eq:evo_G}-\eqref{eq:evo_G_alpha_n} and~\eqref{eq:evo_F1}-\eqref{eq:evo_F_alpha_n}. Under the $\varoast$-convolution, the product $\Fai(s)\Faii(s)$ yields additional entire components $\mathscr{F}_{\psi, i}(s)$ and the DC component $\mathscr{F}_{\varphi}$, cf.~Fig.~\ref{fig:variable_conv}, left. The symmetry $\mathscr{F}_{\psi,2}(s)=\mathscr{F}_{\psi,1}(1-s)$ further simplifies the CE description in Theorem~\ref{thm:CE}. Next, we show that these three can be obtained solely from $\mathscr{F}_{\mathfrak{a},1}(s)$, evaluated along $\tau=1/2+\mathrm{i}u$ via the Fourier transform, cf.~Fig.~\ref{fig:variable_conv}, right.


\begin{figure}[!ht]
\centering
\begin{tikzpicture}[x=0.8cm,y=1.75cm]
\tikzset{
  m/.style={inner sep=2pt, text opacity=1},
  dot/.style={circle,fill=black,inner sep=1.2pt},
  solid arrow/.style={arrows = {-Latex[width'=0pt .5, length=4pt]}, shorten >=.25pt, shorten <=.25pt, line width=0.6pt},
  double arrow/.style={arrows = {Latex[width'=0pt .5, length=4pt]-Latex[width'=0pt .5, length=4pt]}, shorten >=.25pt, shorten <=.25pt, line width=0.6pt},
}

\coordinate (l1-0) at (0,1.5);
\coordinate (l1-1) at (2,2.5);
\coordinate (l1-2) at (4,1.5);

\coordinate (lm-0) at (2,1.5);
\coordinate (lm-1) at (2,1);
\coordinate (lm-2) at (2,0.5);

\coordinate (l2-0) at (0,0);
\coordinate (l2-1) at (2,0);
\coordinate (l2-2) at (4,0);

\node[m] (L1-0) at (l1-0) {$\mathscr{F}_{\mathfrak{a},1}(s)$};
\node[m] (L1-2) at (l1-2) {$\mathscr{F}_{\mathfrak{a},2}(s)$};

\node[m] (LM-0) at (lm-0) {$\otimes$};
\node[m] (LM-1) at (lm-1) {$\varphi(s)$};
\node[anchor=west] at ($(LM-1.east)+(-0.15,0)$) {$\in\mathbb{C}$};
\node[m] (LM-2) at (lm-2) {$\oplus$};

\node[m] (L2-0) at (l2-0) {$\mathscr{F}_{\psi,1}(s)$};
\node[m] (L2-1) at (l2-1) {$\mathscr{F}_{\varphi}$};
\node[m] (L2-2) at (l2-2) {$\mathscr{F}_{\psi,2}(s)$};

\draw[solid arrow] (L1-0) -- (LM-0);
\draw[solid arrow] (L1-2) -- (LM-0);
\draw[solid arrow] (LM-0) -- (LM-1);

\draw[solid arrow] (LM-2) -- (LM-1);
\draw[solid arrow] (L2-0) -- (LM-2);
\draw[solid arrow] (L2-1) -- (LM-2);
\draw[solid arrow] (L2-2) -- (LM-2);


\coordinate (r1-0) at (10,1.5);
\coordinate (r1-1) at (12,2.5);
\coordinate (r1-2) at (14,1.5);

\coordinate (rm-0) at (12,1.5);
\coordinate (rm-1) at (12,1);
\coordinate (rm-2) at (12,0.5);

\coordinate (r2-0) at (10,0);
\coordinate (r2-1) at (12,0);
\coordinate (r2-2) at (14,0);

\node[m] (R1-0) at (r1-0) {$\mathscr{F}_{\mathfrak{a},1}(\tau)$};
\node[m] (R1-2) at (r1-2) {$\mathscr{F}_{\mathfrak{a},1}(\tau^*)$};

\node[m] (RM-0) at (rm-0) {$\otimes$};
\node[m] (RM-1) at (rm-1) {$\tilde{\varphi}(u)$};
\node[anchor=west] at ($(RM-1.east)+(-0.15,0)$) {$=\varphi(\tau)\in\mathbb{R}$};
\node[m] (RM-2) at (rm-2) {$\widehat{\varphi}(\omega)$};

\node[m] (R2-0) at (r2-0) {$\mathscr{F}_{\psi,1}(s)$};
\node[m] (R2-1) at (r2-1) {$\mathscr{F}_{\varphi}$};
\node[m] (R2-2) at (r2-2) {$\mathscr{F}_{\psi,2}(s)$};

\draw[solid arrow] (R1-0) -- (RM-0);
\draw[solid arrow] (R1-2) -- (RM-0);
\draw[solid arrow] (RM-0) -- (RM-1);

\draw[double arrow] (RM-1) -- (RM-2);
\draw[solid arrow] (RM-2) -- (R2-0);
\draw[solid arrow] (RM-2) -- (R2-1);
\draw[solid arrow] (RM-2) -- (R2-2);


\begin{scope}[shift={(0,-0.375)}]
\draw[->, line width=0.5pt] (9,0) -- (15,0) node[right] {$\omega$};

\draw[densely dashed] (12,0) -- (12,0.231); 
\node[below, font=\small] at (12,0) {$\omega_k=0$};
\node[below, font=\small] at (10,0) {$\omega_k<0$};
\node[below, font=\small] at (14,0) {$\omega_k>0$};

\foreach \x/\h in {
  9.6/0.0594,
  10.32/0.1056,
  10.8/0.0792,
  11.28/0.1485,
  11.64/0.0924
} {
  \draw[line width=0.5pt] (\x,0) -- (\x,\h);
  \fill (\x,\h) circle (0.6pt);
}

\draw[line width=0.5pt] (12,0) -- (12,0.2244); 
\fill (12,0.2244) circle (0.6pt);

\foreach \x/\h in {
  12.36/0.0924,
  12.72/0.1485,
  13.2/0.0792,
  13.68/0.1056,
  14.4/0.0594
} {
  \draw[line width=0.5pt] (\x,0) -- (\x,\h);
  \fill (\x,\h) circle (0.6pt);
}
\end{scope}

\begin{scope}[shift={(2.8,0.15)}, scale=0.7]

\fill[pattern=dots, pattern color=gray!60] (5.3,0.4) rectangle (6.7,1.65);

\draw[->, line width=0.5pt] (5.3,1) -- (6.7,1) node[right, font=\scriptsize] {$\Re s$};
\draw[->, line width=0.5pt] (6,0.4) -- (6,1.65) node[above, font=\small] {$\Im s$};

\draw[dashed, line width=0.6pt] (6.3,0.4) -- (6.3,1.65);

\node[font=\scriptsize, anchor=west] at (6.3,1.3) {$\tau=\frac{1}{2}+\mathrm{i}u$};

\end{scope}

\begin{scope}[shift={(7.0,0.2)}, scale=0.8]

  \draw[step=0.2, gray!30, very thin] (-1.5,-0.4) grid (1.5,0.1);
  
  \draw[->, line width=0.5pt] (-1.5,-0.25) -- (1.7,-0.25);
  \draw[->, line width=0.5pt] (-1.5,-0.4) -- (-1.5,0.1);

  \node[anchor=north, font=\footnotesize] at (0,-0.35) {$u$};

  \draw[domain=0:3, samples=2]
    plot ({\x-1.5}, {-0.05});

  \draw[domain=0:3, samples=200]
    plot ({\x-1.5}, {-0.25 + 0.15*cos(deg(2*\x))});

  \draw[domain=0:3, samples=200]
    plot ({\x-1.5}, {-0.25 + 0.0725*cos(deg(4*\x))});

  \draw[domain=0:3, samples=200]
    plot ({\x-1.5}, {-0.25 + 0.04*cos(deg(6*\x))});

\end{scope}

\end{tikzpicture}
\caption{Left panel: Relations among $\mathscr{F}_{\mathfrak{a},i}(s)$, $\mathscr{F}_{\psi,i}(s)$, $i\in\{1,2\}$, with $\mathscr{F}_{\varphi}$. Right panel: Their acquisition via the Fourier transform along $\tau=1/2+\mathrm{i}u$.}
\label{fig:variable_conv}
\end{figure}
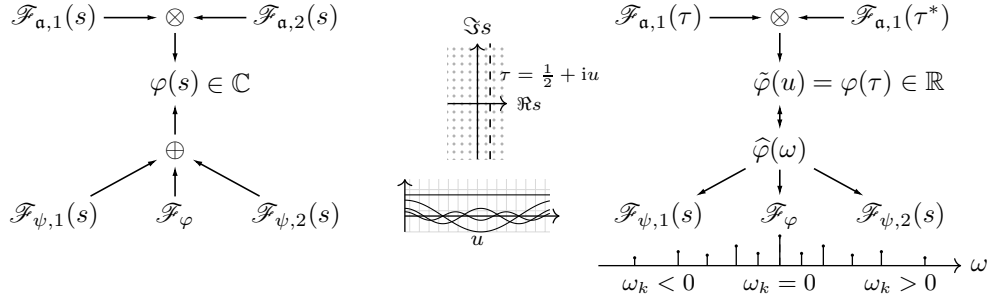

\begin{corollary}~\label{corol:F_dc_comp}
Let $\varphi(s) \triangleq \mathscr{F}_{\mathfrak{a},1}(s)\mathscr{F}_{\mathfrak{a},2}(s)$. Let $\tau\triangleq c+\mathrm{i}u\in\mathbb{C}$ with any $c\in\mathbb{R}$ fixed, then
\begin{equation}~\label{eq:tilde_phi_u}
    \tilde{\varphi}(u)\triangleq \varphi(\tau)= \sum_{k=1}^{(n-1)^2} A_k^{(c)}\cdot \mathrm{e}^{\mathrm{i}\omega_k u} \in\mathbb{C},
\end{equation}
where the $k$-index runs over all \textit{finite} \textit{ordered} pairs $(i,j)$ for every $i,j=1, \ldots, n-1$, with $A_{ij}^{(c)}\triangleq \beta_i \beta_{-j}\left(w_i / w_j\right)^c\in\mathbb{R}$ and $\omega_{ij}\triangleq \log (w_i / w_j)\in\mathbb{R}$. The \textit{Fourier transform} of~\eqref{eq:tilde_phi_u} is given as $\widehat{\varphi}(\omega)=2\pi \sum_{k=1}^{(n-1)^2} A_k^{(c)} \cdot \delta(\omega-\omega_k)$, from which
\begin{equation}~\label{eq:phi_zero}
    \mathscr{F}_{\varphi} =\sum_{k: \omega_k=0} A_k^{(c)}= \frac{1}{2\pi} \widehat{\varphi}(0),
\end{equation}
where $|\{k:\omega_k=0\}|=n-1$.
\end{corollary}

\begin{proof}
Please refer to Section~\ref{sec:F_dc_comp}.
\end{proof}


\begin{corollary}~\label{corol:F_dc_comp_half}
Let $\tau= 1/2+\mathrm{i}u\in\mathbb{C}$, then
\begin{equation}~\label{eq:phi_omega_half}
    \tilde{\varphi}(u) = \mathscr{F}_{\varphi} + 2\sum_{k:\omega_k<0} A_{k}^{(c)} \cdot \cos (\omega_{k} u) \in\mathbb{R},
\end{equation}
where $|\{k:\omega_k<0\}|=|\{k:\omega_k>0\}|=(n-1)(n-2)/2$.
\end{corollary}

\begin{proof}
Please refer to Section~\ref{sec:F_dc_comp_half}.
\end{proof}

\begin{remark}
Corollary~\ref{corol:F_dc_comp} shows that for any $c\in\mathbb{R}$, $\tilde{\varphi}(u)$ is \textit{a finite sum of complex exponentials}, from which $\mathscr{F}_{\varphi}$ can be extracted as the DC component, while Corollary~\ref{corol:F_dc_comp_half} shows that if $c=1/2$, it reduces to \textit{a finite sum of cosines with a DC offset}, where the number of \textit{distinct} cosine frequencies in~\eqref{eq:phi_omega_half} is the cardinality of $\left\{\left|\log \left(w_i / w_j\right)\right|: 1 \leq j<i \leq n-1\right\}$, i.e., the number of \textit{distinct differences} of $\log w$'s, which is up to $(n-1)(n-2)/2$. Moreover, a reduction from the generic upper bounds occurs iff there exist two distinct ordered pairs $(i, j) \neq(i',j')$ with $\omega_{i j}=\omega_{i'j'}>0$.
\end{remark}

\begin{corollary}~\label{corol:F_mellin_comp}
Let $\tau= 1/2+\mathrm{i}u\in\mathbb{C}$, then
\begin{gather}
    \mathscr{F}_{\psi, 1}(s) = \sum_{k:\omega_k <0} A_k^{(c)}\cdot \mathrm{e}^{\omega_k(s-c)}, \quad  \forall s\in\mathbb{C}, \label{eq:F_psi_1_F}\\
    \mathscr{F}_{\psi, 2}(s) = \sum_{k:\omega_k >0} A_k^{(c)}\cdot \mathrm{e}^{\omega_k(s-c)}, \quad  \forall s\in\mathbb{C},\label{eq:F_psi_2_F}
\end{gather}
where $A_k^{(c)}$ and $\omega_k$ are defined in Corollary~\ref{corol:F_dc_comp}.
\end{corollary}

\begin{proof}
Please refer to Section~\ref{sec:F_mellin_comp}.
\end{proof}

Next, we illustrate the effect of $\varoast$-convolutions on these components, using the BEC and BSC as representative examples.

\begin{exmp}[$W_{\operatorname{BEC}(\varepsilon)}$ under $\varoast$-convolutions]
According to Table~\ref{tab:channel_comp_merged}, $\varphi(s)=\mathscr{F}_{\mathfrak{a}, 1}(s) \mathscr{F}_{\mathfrak{a}, 2}(s)=0$ $\forall s\in \mathbb{C}$, thus, $\mathscr{F}_{\varphi}=\mathscr{F}_{\psi, 1}(s)=\mathscr{F}_{\psi, 2}(s)=0$, $\forall s\in \mathbb{C}$. It follows that only $\alpha_0^{\varoast m}$ and $\alpha_n^{\varoast m}$ change along the $m$-fold $\varoast$-convolution, while all other components remain zero for every $m=0,1,2,\ldots$, cf. Table~\ref{tab:bec_variable_evo}.
\end{exmp}

\begin{table}[h!]
\begin{center}
\begin{tabular}{|c || c c c c c || c c c |} 
 \hline
 $m$ & $\mathscr{F}_{\mathfrak{a}^{\varoast m},1}(s)$ & $\mathscr{F}_{\mathfrak{a}^{\varoast m},2}(s)$ & $\alpha_0^{\varoast m}$ & $\alpha_n^{\varoast m}$ & $n^{\varoast m}$ & $\mathscr{F}^{\varoast m}_{\varphi}$ & $\mathscr{F}^{\varoast m}_{\psi, 1}(s)$ & $\mathscr{F}^{\varoast m}_{\psi, 2}(s)$\\ [0.5ex] 
 \hline
 $0$ & $0$ & $0$ & $\varepsilon$ & $\barepsilon$ & $1$ & $0$ & $0$ & $0$ \\ 
 \hline
 $1$ & $0$ & $0$ & $\varepsilon^2$ & $1-\varepsilon^2$ & $1$ & $0$ & $0$ & $0$ \\ 
 \hline
 $2$ & $0$ & $0$ & $\varepsilon^4$ & $1-\varepsilon^4$ & $1$ & $0$ & $0$ & $0$ \\ 
 \hline
 $3$ & $0$ & $0$ & $\varepsilon^8$ & $1-\varepsilon^8$ & $1$ & $0$ & $0$ & $0$ \\ 
 \hline
\end{tabular}
\vspace{1mm}
\caption{Component evolution under the $m$-fold variable-domain $\varoast$-convolution for the BEC}
\label{tab:bec_variable_evo}
\end{center}
\end{table}

\begin{exmp}[$W_{\operatorname{BSC}(p)}$ under $\varoast$-convolutions]~\label{exmp:bsc_var_decomp}
According to Table~\ref{tab:channel_comp_merged}, $\varphi(s)=\mathscr{F}_{\mathfrak{a}, 1}(s) \mathscr{F}_{\mathfrak{a}, 2}(s)=C$ is constant in $s$, thus $\mathscr{F}_{\varphi}=C$ and $\mathscr{F}_{\psi, 1}(s)=\mathscr{F}_{\psi, 2}(s)=0$, $\forall s\in \mathbb{C}$. Define $\alpha\triangleq \bar{p}^{1-c} p^c$ and $\beta\triangleq p^{1-c} \bar{p}^c$ for any $c\in(0,1)$, then
\begin{equation}
    A(s)=\alpha \cdot \mathrm{e}^{\Delta(s-c)}, \quad B(s)=\beta \cdot \mathrm{e}^{-\Delta(s-c)}, \quad \forall s\in \mathbb{C}.
\end{equation}
Note that these are merely convenient factorizations for $W_{\operatorname{BSC}(p)}$, which \textit{per se} do not depend on the choice of $c$. Therefore, we have $\tilde{A}(u)\triangleq A(\tau)=\alpha \mathrm{e}^{\mathrm{i} u \Delta}$ and $\tilde{B}(u)\triangleq B(\tau)=\beta \mathrm{e}^{-\mathrm{i} u \Delta}$. For $m=1$, the component evolution is trivial as given in Table~\ref{tab:bsc_variable_evo}, where $\alpha_0^\varoast=2C$, i.e., a single $\varoast$-convolution does not preserve the BSC class~\cite{wang2023sub}. In what follows, we consider level $m=2$. By definition, $\varphi^{\varoast 2}(\tau)= \mathscr{F}_{\mathfrak{a}^{\varoast 2}, 1}(\tau)\mathscr{F}_{\mathfrak{a}^{\varoast 2}, 2}(\tau)=C^3(17 C+4(\tilde{A}^2(u)+\tilde{B}^2(u)))=17 C^4+4 C^3(\alpha^2 \mathrm{e}^{\mathrm{i} 2 u \Delta}+\beta^2 \mathrm{e}^{-\mathrm{i} 2 u \Delta})\triangleq \tilde{\varphi}^{\varoast 2}(u)$. By applying the Fourier transform, we obtain
\begin{equation}
    \widehat{\varphi}^{\varoast 2}(\omega) =2 \pi \left(17 C^4 \cdot  \delta(\omega)+4 C^3\left(\alpha^2  \delta(\omega-2 \Delta)+\beta^2 \delta(\omega+2 \Delta)\right)\right).
\end{equation}
Hence, $\mathscr{F}^{\varoast 2}_{\varphi}=\frac{1}{2 \pi} \widehat{\varphi}^{\varoast 2}(0)=17C^4$. Moreover, if $c=1/2$, then $\alpha=\beta=C^{1/2}$, thus
\begin{equation}~\label{eq:bsc_2nd_var_conv}
    \widehat{\varphi}^{\varoast 2}(\omega) = 2 \pi C^4\left( 17\delta(\omega) + 4 \delta(\omega-2 \Delta) + 4 \delta(\omega+2 \Delta) \right),
\end{equation}
which corresponds to the \textit{power spectral density} (PSD) of a cosine with a DC offset, cf. Corollary~\ref{corol:F_dc_comp_half}. Therefore, the corresponding amplitudes and frequencies can be identified from~\eqref{eq:bsc_2nd_var_conv}: $A_1=A_2=4C^4$, $\omega_1=-2\Delta<0$, $\omega_2=2\Delta>0$. By substituting these values into~\eqref{eq:F_psi_1_F} and~\eqref{eq:F_psi_2_F} with $c=1/2$, we obtain $\mathscr{F}^{\varoast 2}_{\psi, 1}(s)=4C^4 \mathrm{e}^{2\Delta(s-c)}=4C^3 A^2(s)$ and $\mathscr{F}^{\varoast 2}_{\psi, 2}(s)=4C^4 \mathrm{e}^{-2\Delta(s-c)}=4C^3 B^2(s)$. By repeating the same procedure, we showcase results for $m=3$ as follows.
\begin{align}
    \mathscr{F}_{\mathfrak{a}^{\varoast 3},1}(s) &= A_s^8+8 C A_s^6+28 C^2 A_s^4+56 C^3 A_s^2, \label{eq:F_bsc_1_3} \\
    \mathscr{F}_{\mathfrak{a}^{\varoast 3},2}(s) &= B_s^8+8 C B_s^6+28 C^2 B_s^4+56 C^3 B_s^2, \label{eq:F_bsc_2_3} \\
    \alpha_{0}^{\varoast 3} &= 70C^4, \quad \alpha_{n}^{\varoast 3} = 0, \label{eq:alpha_bsc_0n_3} \\
    \widehat{\varphi}^{\varoast 3}(\omega) &= 2\pi C^8(3985 \delta(\omega)+1800\delta(\omega\pm 2 \Delta)+476\delta(\omega\pm 4 \Delta)+56\delta(\omega\pm 6 \Delta)),  \\
    \mathscr{F}^{\varoast 3}_{\varphi} &= 3985 C^8, \label{eq:F_bsc_phi_3} \\
    \mathscr{F}^{\varoast 3}_{\psi, 1}(s) &= 1800 C^7 A_s^2+476 C^6 A_s^4+56 C^5 A_s^6, \label{eq:F_bsc_psi_1_3} \\
    \mathscr{F}^{\varoast 3}_{\psi, 2}(s) &= 1800 C^7 B_s^2+476 C^6 B_s^4+56 C^5 B_s^6. \label{eq:F_bsc_psi_2_3}
\end{align}
As shown in Table~\ref{tab:bsc_variable_evo}, the $m$-fold $\varoast$-convolution requires participation of all entire components defined in Lemma~\ref{lemma:F_comp_evo}. One observes that a single $\varoast$-convolution does not preserve the BSC class~\cite{wang2023sub}, since the DC component $\mathscr{F}_{\varphi}$ in~\eqref{eq:phi_zero} contributes to the evolution of $\alpha_0^\varoast$ via~\eqref{eq:evo_F_alpha_0}.
\end{exmp}

\begin{table}[h!]
\begin{center}
\begin{tabular}{|c || c c c c c || c c c |} 
 \hline
 $m$ & $\mathscr{F}_{\mathfrak{a}^{\varoast m},1}(s)$ & $\mathscr{F}_{\mathfrak{a}^{\varoast m},2}(s)$ & $\alpha_0^{\varoast m}$ & $\alpha_n^{\varoast m}$ & $n^{\varoast m}$ & $\mathscr{F}^{\varoast m}_{\varphi}$ & $\mathscr{F}^{\varoast m}_{\psi, 1}(s)$ & $\mathscr{F}^{\varoast m}_{\psi, 2}(s)$\\ [0.5ex] 
 \hline
 $0$ & $A_s$ & $B_s$ & $0$ & $0$ & $2$ & $C$ & $0$ & $0$ \\ 
 \hline
 $1$ & $A_s^2$ & $B_s^2$ & $2C$ & $0$ & $2$ & $C^2$ & $0$ & $0$ \\ 
 \hline
 $2$ & $A_s^4+4CA_s^2$ & $B_s^4+4CB_s^2$ & $6C^2$ & $0$ & $3$ & $17C^4$ & $4 C^3 A_s^2$ & $4 C^3 B_s^2$ \\ 
 \hline
 $3$ & \eqref{eq:F_bsc_1_3} & \eqref{eq:F_bsc_2_3} & \eqref{eq:alpha_bsc_0n_3} & \eqref{eq:alpha_bsc_0n_3} & $6$ & \eqref{eq:F_bsc_phi_3} & \eqref{eq:F_bsc_psi_1_3} & \eqref{eq:F_bsc_psi_2_3} \\ 
 \hline
\end{tabular}
\vspace{1mm}
\caption{Component evolution under the $m$-fold variable-domain $\varoast$-convolution for the BSC}
\label{tab:bsc_variable_evo}
\end{center}
\end{table}

The comparisons of the BEC and BSC in Sections~\ref{sec:check_conv_domain} and~\ref{sec:var_conv_domain} further illustrate their roles as extremal channels~\cite{alsan2012channel}~\cite[Theorem 4.143]{richardson2008modern}. In contrast to the BEC, the BSC exhibits a more intricate evolution under convolutions.

\section{Conversion between the variable-domain and the check-domain}~\label{sec:convert_G/F}
This section addresses the conversion between the variable- and check-domain, establishing links between $\mathscr{F}_{\mathfrak{a},i}(s)$ with $i\in\{1,2\}$ and $\mathcal{G}_{\mathfrak{a}}(0,\nu)$ in the $D$-density, and between $\mathscr{G}_{|\mathfrak{a}|}(\nu)$ and $\mathcal{F}_{|\mathfrak{a}|}(s)$ in the $|D|$-density, within the SOA $\Re \nu>0$ and $\Re s\in(0,1)$; see brown and purple arrows in Fig.~\ref{fig:entities_domains}.

\begin{prop}~\label{prop:G2F_convert}
Fix $\Re s \in(0,1)$, then
\begin{gather}
    \mathcal{F}_{|\mathfrak{a}|}(s)=\alpha_0 +\sum_{t=0}^{\infty} \frac{(-1)^t}{t!}\left(A(1-s,t)+A(s,t)\right) \cdot \mathscr{G}_{|\mathfrak{a}|}(t), \label{eq:G2F_convert} \\
    A(\alpha,t) \triangleq 2^{-\alpha} \cdot \Gamma(\alpha+1)_2 \tilde{F}_1(\alpha-1, \alpha+1 ; 1+\alpha-t; 1 / 2), \quad \forall \alpha\in\{s,1-s\}, \label{eq:A_s_nu}
\end{gather}
where ${ }_2 \tilde{F}_1(a, b ; c ; z)\triangleq { }_2 F_1(a, b ; c ; z) / \Gamma(c)$ denotes the \textit{regularized Gauss hypergeometric function}~\cite[(9.02)]{olver1997asymptotics}.
\end{prop}

\begin{proof}
Please refer to Section~\ref{sec:G2F_convert}.
\end{proof}

\begin{corollary}
Let $s=1/2$, then
\begin{gather}
\mathfrak{B}(|\mathfrak{a}|) = \alpha_0 +\sum_{t=0}^{\infty} \frac{(-1)^t}{t!} D(t)\cdot \mathscr{G}_{|\mathfrak{a}|}(t), \label{eq:evo_G2B} \\
        D(t) \triangleq \sqrt{\frac{\pi}{2}} { }_2 \tilde{F}_1\left(-\frac{1}{2}, \frac{3}{2} ; \frac{3}{2}-t ; \frac{1}{2}\right).
\end{gather}
\end{corollary}

\begin{proof}
By Corollary~\ref{corol:var_node_half}, $\mathfrak{B}(|\mathfrak{a}|)=\mathcal{F}_{|\mathfrak{a}|}(1 / 2)$. Let $s=1/2$, it follows from~\eqref{eq:A_s_nu} that
\begin{equation*}
    A(1/2,t)=2^{-1 / 2} \cdot \Gamma(3 / 2){ }_2 \tilde{F}_1\left(-1/2, 3/2 ; 3/2-t ; 1/2\right),
\end{equation*}
with $\Gamma(3 / 2)=\sqrt{\pi}/2$. Moreover, for $s=1/2$, we have $1-s=s$, thus the two coefficients coincide and their sum simplifies to $D(t) = 2\cdot A(1/2,t)=\sqrt{\pi/2}\cdot { }_2 \tilde{F}_1\left(-1/2, 3/2 ; 3/2-t ; 1/2\right)$. 
\end{proof}

\begin{prop}~\label{prop:F2G_convert}
Fix $\Re \nu >0$, then
\begin{gather}
    \mathcal{G}_{\mathfrak{a}}(0, \nu) = \alpha_n +\sum_{t=0}^{\infty} \frac{(-1)^t}{t!} G(\nu,t)\left(\mathscr{F}_{\mathfrak{a}, 1}(t)+\mathscr{F}_{\mathfrak{a}, 2}(-t)\right),~\label{eq:F2G_convert} \\
    G(\nu,t) \triangleq 2^{-\nu}\cdot \Gamma(\nu+1){}_2\tilde{F}_1(\nu, \nu+1 ; 1+\nu-t ; 1 / 2).
\end{gather}
\end{prop}

\begin{proof}
Please refer to Section~\ref{sec:F2G_convert}.
\end{proof}


\section{Component evolution-based channel polarization}~\label{sec:bi_convert_GF}
This section establishes conversions between $\mathscr{F}_{\mathfrak{a},1}(s)$ and $\Ga(\nu)$ over the entire complex planes and across the $D$- and $|D|$-densities; see gray arrows in Fig.~\ref{fig:entities_domains}, and proposes the CE-based channel polarization of a given discrete BMS channel. 

\begin{theorem}~\label{thm:G2F_convert2}
Let $\forall s,\nu\in \mathbb{C}$, then
\begin{align}
    \mathscr{F}_{\mathfrak{a}, 1}(s) &=  \sum_{t=0}^{\infty} \frac{(-1)^t}{t!}A(s,t) \cdot \mathscr{G}_{|\mathfrak{a}|}(t), \quad \forall s\in \mathbb{C}, \label{eq:evo_G2F1} \\
    \mathscr{G}_{|\mathfrak{a}|}(\nu) &=\sum_{t=0}^{\infty} \frac{(-1)^t}{t!} G(\nu,t)\left(\mathscr{F}_{\mathfrak{a}, 1}(t)+\mathscr{F}_{\mathfrak{a}, 1}(t+1)\right), \quad \forall \nu\in \mathbb{C}, \label{eq:evo_F122G}
\end{align}
where $A(s, t)$ resp. $G(\nu,t)$ defined in Proposition~\ref{prop:G2F_convert} resp.~\ref{prop:F2G_convert} are entire functions on $s\in\mathbb{C}$ resp. on $\nu\in\mathbb{C}$ for $t\in\mathbb{Z}_{\geq 0}$.
\end{theorem}

\begin{proof}
Please refer to Section~\ref{sec:G2F_convert2}.
\end{proof}

\begin{remark}
Rather than focusing on the check- and variable transforms, Theorem~\ref{thm:G2F_convert2} extends Proposition~\ref{prop:G2F_convert} and~\ref{prop:F2G_convert} to the entire complex planes and characterizes the relations between their entire counterparts: it derives $\mathscr{F}_{\mathfrak{a}, 1}(s)$ from infinitely many moments of $\mathscr{G}_{|\mathfrak{a}|}(\nu)$ and vice versa. However, this representation does not provide a numerically stable method for recovering one from the other via the alternating series. Instead, we demonstrate in Example~\ref{exmp:bsc_110} resp.~\ref{exmp:bsc_1101} that $\mathscr{G}_{|\mathfrak{a}|}(\nu)$ resp. $\mathscr{F}_{\mathfrak{a}, 1}(s)$  can be identified via
\begin{align}
    (1-w)^\nu(1+w)^{-\nu} &=\sum_{t=0}^{\infty} \frac{(-1)^t}{t!} G(\nu, t) w^t, \quad \forall w\in(0,1), \ \forall \nu\in\mathbb{C},~\label{eq:evo_F2G1_reform2} \\
    1/2 \cdot(1-z)^s(1+z)^{1-s} &= \sum_{t=0}^{\infty} \frac{(-1)^t}{t!} A(s,t) z^t, \quad \forall z\in(0,1), \ \forall s\in\mathbb{C},~\label{eq:evo_G2F1_reform1}    
\end{align}
respectively, provided that the expression of $\mathscr{F}_{\mathfrak{a}, 1}(s)$ resp. $\mathscr{G}_{|\mathfrak{a}|}(\nu)$ is available without evaluating $G(\nu,t)$ resp. $A(s,t)$, $\forall t\in\mathbb{Z}_{\geq 0}$. Such conversion via identification is illustrated in Fig.~\ref{fig:bilateral_convert}. In what follows, we consider several bit-channels of the BSC to illustrate the conversion using this approach.
\end{remark}

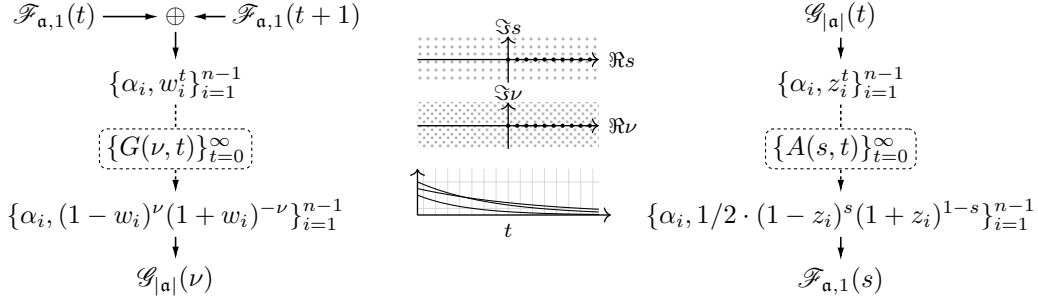
\begin{figure}[!ht]
\centering
\begin{tikzpicture}[x=0.8cm,y=1.75cm]
\tikzset{
  m/.style={inner sep=2pt, text opacity=1},
  dot/.style={circle,fill=black,inner sep=1.2pt},
  solid arrow/.style={arrows = {-Latex[width'=0pt .5, length=5pt]}, shorten >=.25pt, shorten <=.25pt, line width=0.6pt},
  double arrow/.style={arrows = {Latex[width'=0pt .5, length=5pt]-Latex[width'=0pt .5, length=5pt]}, shorten >=.25pt, shorten <=.25pt, line width=0.6pt},
  dot/.style={circle, fill=blue!80, inner sep=0, minimum size=3pt}, 
  dash line/.style={shorten >=.25pt, shorten <=.25pt, dash pattern=on 1.2pt off 1.2pt, line width=0.6pt},
  dash arrow/.style={arrows = {-Latex[width'=0pt .5, length=5pt]}, shorten >=.25pt, shorten <=.25pt, dash pattern=on 1.2pt off 1.2pt, line width=0.6pt},
}


\coordinate (l1-0) at (0,1.5);
\coordinate (l1-2) at (4,1.5);

\coordinate (lm-0) at (2,1.5);
\coordinate (lm-1) at (2,1);
\coordinate (lm-2) at (2,0.5);
\coordinate (lm-3) at (2,0);

\coordinate (l2-1) at (2,-0.5);

\node[m] (L1-0) at (l1-0) {$\mathscr{F}_{\mathfrak{a},1}(t)$};
\node[m] (L1-2) at (l1-2) {$\mathscr{F}_{\mathfrak{a},1}(t+1)$};

\node[m] (LM-0) at (lm-0) {$\oplus$};
\node[m] (LM-1) at (lm-1) {$\{\alpha_i, w_i^t\}_{i=1}^{n-1}$};
\node[
  m,
  draw,
  dashed,
  dash pattern=on 1.2pt off 1.2pt,
  rounded corners=3pt,
  inner sep=2.5pt
] (LM-2) at (lm-2) {$\{G(\nu,t)\}_{t=0}^{\infty}$}; 
\node[m] (LM-3) at (lm-3) {$\{\alpha_i, (1-w_i)^\nu(1+w_i)^{-\nu} \}_{i=1}^{n-1}$};

%

\node[m] (L2-1) at (l2-1) {$\mathscr{G}_{|\mathfrak{a}|}(\nu)$};

\draw[solid arrow] (L1-0) -- (LM-0);
\draw[solid arrow] (L1-2) -- (LM-0);
\draw[solid arrow] (LM-0) -- (LM-1);
\draw[dash line] (LM-1) -- (LM-2);
\draw[dash arrow] (LM-2) -- (LM-3);
\draw[solid arrow] (LM-3) -- (L2-1);


\coordinate (r1-0) at (11,2);
\coordinate (r1-1) at (13,2);
\coordinate (r1-2) at (15,2);

\coordinate (rm-0) at (13,1.5);
\coordinate (rm-1) at (13,1);
\coordinate (rm-2) at (13,0.5);
\coordinate (rm-3) at (13,0);

\coordinate (r2-1) at (13,-0.5);


\node[m] (RM-0) at (rm-0) {$\mathscr{G}_{|\mathfrak{a}|}(t)$};
\node[m] (RM-1) at (rm-1) {$\{\alpha_i, z_i^t\}_{i=1}^{n-1}$};
\node[
  m,
  draw,
  dashed,
  dash pattern=on 1.2pt off 1.2pt,
  rounded corners=3pt,
  inner sep=2.5pt
] (RM-2) at (rm-2) {$\{A(s,t)\}_{t=0}^{\infty}$}; 
\node[m] (RM-3) at (rm-3) {$\{\alpha_i, 1/2\cdot(1-z_i)^s(1+z_i)^{1-s} \}_{i=1}^{n-1}$};

\node[m] (R2-1) at (r2-1) {$\mathscr{F}_{\mathfrak{a},1}(s)$};

\draw[solid arrow] (RM-0) -- (RM-1);
\draw[dash line] (RM-1) -- (RM-2);
\draw[dash arrow] (RM-2) -- (RM-3);
\draw[solid arrow] (RM-3) -- (R2-1);

\begin{scope}[shift={(7.5,1.25)}]

    \fill[pattern=dots, pattern color=gray!60] (-1.5,-0.25) rectangle (1.5,0.1);
    
    \draw[->, line width=0.5pt] (-1.5,-0.075) -- (1.5,-0.075) node[right, font=\footnotesize] {$\Re s$};
    \draw[->, line width=0.5pt] (0,-0.25) -- (0,0.1) node[above, yshift=-2.5pt, font=\footnotesize] {$\Im s$};

    \foreach \x in {0,0.15,...,1.35} {
        \fill (\x,-0.075) circle (0.8pt);
    }

\end{scope}


\begin{scope}[shift={(7.5,0.75)}]

    \fill[pattern=crosshatch dots, pattern color=gray!60] (-1.5,-0.25) rectangle (1.5,0.1);
    
    \draw[->, line width=0.5pt] (-1.5,-0.075) -- (1.5,-0.075) node[right, font=\footnotesize] {$\Re \nu$};
    \draw[->, line width=0.5pt] (0,-0.25) -- (0,0.1) node[above, yshift=-2.5pt, font=\footnotesize] {$\Im \nu$};

    \foreach \x in {0,0.15,...,1.35} {
        \fill (\x,-0.075) circle (0.8pt);
    }

\end{scope}

\begin{scope}[shift={(7.5,0.25)}]

  \draw[step=0.2, gray!30, very thin] (-1.5,-0.25) grid (1.5,0.1);
  
  \draw[->, line width=0.5pt] (-1.5,-0.25) -- (1.7,-0.25);
  \draw[->, line width=0.5pt] (-1.5,-0.25) -- (-1.5,0.1);

  \node[anchor=north, font=\footnotesize] at (0,-0.25) {$t$};

  \draw[domain=0:3, samples=50]
    plot ({\x-1.5}, {0.15*0.3^\x - 0.25});
    
  \draw[domain=0:3, samples=50]
    plot ({\x-1.5}, {0.25*0.45^\x - 0.25});
    
  \draw[domain=0:3, samples=50]
    plot ({\x-1.5}, {0.2*0.6^\x - 0.25});

\end{scope}


\end{tikzpicture}
\caption{Left panel: Conversion from $\mathscr{F}_{\mathfrak{a},1}(s)$ to $\mathscr{G}_{|\mathfrak{a}|}(\nu)$. Right panel: Conversion from $\mathscr{G}_{|\mathfrak{a}|}(\nu)$ to $\mathscr{F}_{\mathfrak{a},1}(s)$.}
\label{fig:bilateral_convert}
\end{figure}

\begin{exmp}[$W_{\operatorname{BSC}(p)}$: from $\mathscr{F}_{\mathfrak{a}^{\varoast 2},1}(s)$ to $\mathscr{G}_{|\mathfrak{a}^{\varoast 2}|}(\nu)$]~\label{exmp:bsc_110}
According to Table~\ref{tab:bsc_variable_evo}, for $m=2$ and $t\in\mathbb{Z}_{\geq 0}$, we have $\mathscr{F}_{\mathfrak{a}^{\varoast 2},1}(t)=A^4(t)+4CA^2(t)$ and by symmetry $\mathscr{F}_{\mathfrak{a}^{\varoast 2},2}(-t)=B^4(-t)+4CB^2(-t)=\mathscr{F}_{\mathfrak{a}^{\varoast 2},1}(t+1)=A^4(t+1)+4CA^2(t+1)$. By definition of $A(s)$, $B(s)$, $C$ and $\Delta<0$ in Section~\ref{sec:prelim}, we obtain
\begin{equation}
    \mathscr{F}_{\mathfrak{a}^{\varoast 2},1}(t) + \mathscr{F}_{\mathfrak{a}^{\varoast 2},1}(t+1) = S_4\mathrm{e}^{4\Delta t} + 4CS_2\mathrm{e}^{2\Delta t},
\end{equation}
from which identifying with~\eqref{eq:evo_F2G1_reform2} and~\eqref{eq:evo_F2G1_reform1} it follows that $n^{\varoast 2}=3$, $\alpha^{\varoast 2}_1=4CS_2$, $\alpha^{\varoast 2}_2=S_4$, and $w^{\varoast 2}_1=\mathrm{e}^{2\Delta}>w^{\varoast 2}_2=\mathrm{e}^{4\Delta}$, which implies that $z^{\varoast 2}_1=(1-w^{\varoast 2}_1)/(1+w^{\varoast 2}_1)=D_2/S_2< z^{\varoast 2}_2=D_4/S_4$. Therefore, we have
\begin{equation}~\label{eq:bsc_var_2}
    \mathscr{G}_{|\mathfrak{a}^{\varoast 2}|}(\nu) = 4CD_2^\nu S_2^{1-\nu}+D_4^\nu S_4^{1-\nu}.
\end{equation}
Moreover, the only atom that does not contributes to $\mathscr{G}_{|\mathfrak{a}^{\varoast 2}|}(\nu)$ is at $z=z^{\varoast 2}_0\equiv 0$ with $\alpha^{\varoast 2}_0=6 C^2$, while $\alpha^{\varoast 2}_3=0$ at $z=z^{\varoast 2}_3\equiv 1$, trivially, cf. Table~\ref{tab:bsc_variable_evo}. These values, i.e., $\alpha^{\varoast 2}_0$ resp. $\alpha^{\varoast 2}_3$, follow directly from~\eqref{eq:evo_F_alpha_0} resp.~\eqref{eq:evo_F_alpha_n}, and are determined prior to conversion. One further verifies that $\sum_{i=0}^{3}\alpha^{\varoast 2}_i=1$.
\end{exmp}
\begin{remark}
Example~\ref{exmp:bsc_110} converts $W_{\varoast 2}$ from the variable- to check-domain. Conditioned on the subsequent bit being equal to $0$, the bit-channel $W_{\varoast 2}$ evolves into $W_{\varoast 2 \boxast}$, and the CE then continues via the $\boxast$-convolution.
\end{remark}

\begin{exmp}[$W_{\operatorname{BSC}(p)}$: from $\mathscr{G}_{|\mathfrak{a}^{\varoast 2 \boxast}|}(\nu)$ to $\mathscr{F}_{\mathfrak{a}^{\varoast 2 \boxast},1}(s)$]~\label{exmp:bsc_1101}
By applying~\eqref{eq:evo_G} into~\eqref{eq:bsc_var_2} and~\eqref{eq:evo_G_alpha_0} into $\alpha^{\varoast 2}_0$, we obtain
\begin{gather}
    \mathscr{G}_{|\mathfrak{a}^{\varoast 2 \boxast}|}(\nu) = 16 C^2 D_2^{2 \nu} S_2^{2(1-\nu)}+8 C D_2^\nu D_4^\nu S_2^{1-\nu} S_4^{1-\nu}+D_4^{2 \nu} S_4^{2(1-\nu)}, ~\label{eq:bsc_110} \\
    \alpha^{\varoast 2 \boxast}_0 = 12C^2-36C^4,
\end{gather}
from which identifying with~\eqref{eq:evo_G2F1_reform1} and~\eqref{eq:evo_G2F1_reform} it follows that $n^{\varoast 2 \boxast}=4$, $\alpha_1^{\varoast 2 \boxast}=16 C^2 S_2^2$, $\alpha_2^{\varoast 2 \boxast}=8 C S_2 S_4$, $\alpha_3^{\varoast 2 \boxast}=S_4^2$ and $z_1^{\varoast 2 \boxast}=D_2^2/S_2^2 < z_2^{\varoast 2 \boxast} = D_2D_4/(S_2S_4) < z_3^{\varoast 2 \boxast} =D_4^2/S_4^2$, which also implies that $w_1^{\varoast 2 \boxast} = 4C^2/(S_2^2+D_2^2) > w_2^{\varoast 2 \boxast}=(S_2 S_4-D_2 D_4)/(S_2 S_4+D_2 D_4) > w_3^{\varoast 2 \boxast}=4C^2/(S_4^2+D_4^2)$, while $\alpha_4^{\varoast 2 \boxast}=0$ trivially. One further verifies that $\sum_{i=0}^{4}\alpha^{\varoast 2 \boxast}_i=1$. Therefore, we have
\begin{equation}
    \mathscr{F}_{\mathfrak{a}^{\varoast 2 \boxast},1}(s) = 2^{s+4} C^{2 (s+1)} S_{4}^{1-s}+8 C^{2 s+1} S_{2}^{s} S_{6}^{1-s}+2^s C^{4 s} S_{8}^{1-s}.
\end{equation}
\end{exmp}
\begin{remark}
Example~\ref{exmp:bsc_1101} converts $W_{\varoast 2 \boxast}$ from the check- to variable-domain after the $\boxast$-convolution. Conditioned on the subsequent bit being equal to $1$, the bit-channel $W_{\varoast 2 \boxast}$ evolves into $W_{\varoast 2 \boxast \varoast}$, and the CE then continues via the $\varoast$-convolution.
\end{remark}

\begin{theorem}~\label{thm:CE}
CE derives analytic expressions for the Bhattacharyya parameters of bit-channels at arbitrary levels for a given discrete BMS channel via Procedure~\ref{alg:chan_polar}.
\begin{algorithm}
\floatname{algorithm}{Procedure}
\caption{CE-Based Channel Polarization of a discrete BMS channel}\label{alg:chan_polar}
\begin{algorithmic}[1]            
  \small
  \Require $W\equiv \mathfrak{a}$, $b_1\ldots b_k$ with $b_m\in\{0,1\}\equiv \{\boxast, \varoast\}$
  \Ensure $Z(W_{b_1\ldots b_k})$
  \State $\alpha_0, \alpha_n, \mathscr{G}_{|\mathfrak{a}|}(\nu), \mathscr{F}_{\mathfrak{a},1}(s) \gets \mathfrak{a}$; $\mathscr{F}_{\varphi}\gets 0$; $\mathscr{F}_{\psi,1}(s) \gets 0$
  \State $\mathrm{state_G} \gets (\mathscr{G}_{|\mathfrak{a}|}(\nu), \alpha_0, \alpha_n)$; $\mathrm{state_F} \gets (\mathscr{F}_{\mathfrak{a},1}(s), \alpha_0, \alpha_n)$
  \State $\mathrm{state_{Cr}} \gets (\mathscr{F}_{\varphi}, \mathscr{F}_{\psi,1}(s))$
  \For{$m = 1,\dots,k$}
      \Switch{$b_m$}
        \Case{$\boxast$}
          \State $\mathrm{state_G^\boxast} \gets \mathsf{CheckUpdate}(\mathrm{state_G})$ \Comment{\eqref{eq:evo_G}\eqref{eq:evo_G_alpha_0}\eqref{eq:evo_G_alpha_n}}
          \State $\mathrm{state_G} \gets \mathrm{state_G^\boxast}$
          \If{$m=k \ \textbf{or}\ b_{m+1}=\varoast$}
            \State $\mathrm{state_F} \gets \mathsf{CheckToVar}(\mathrm{state_G})$ \Comment{\eqref{eq:evo_G2F1}}
          \EndIf
        \EndCase
        \Case{$\varoast$}
          \State $\mathrm{state_{Cr}} \gets \mathsf{CrossUpdate}(\mathrm{state_F}[0])$ \Comment{\eqref{eq:phi_zero}\eqref{eq:F_psi_1_F}}
          \State $\mathrm{state_F^\varoast} \gets \mathsf{VarUpdate}(\mathrm{state_{Cr}},\mathrm{state_F})$ \Comment{\eqref{eq:evo_F1}\eqref{eq:evo_F_alpha_0}\eqref{eq:evo_F_alpha_n}}
          \State $\mathrm{state_F} \gets \mathrm{state_F^\varoast}$
          \If{$m<k \ \textbf{and}\ b_{m+1}=\boxast$}
            \State $\mathrm{state_G} \gets \mathsf{VarToCheck}(\mathrm{state_F})$ \Comment{\eqref{eq:evo_F122G}}
          \EndIf
        \EndCase
      \EndSwitch  
  \If{$m = k$}
    \State $\mathfrak{B}(|\mathfrak{a}|) \gets \mathsf{Bhatt}(\mathrm{state_F}[0], \mathrm{state_F}[1])$ \Comment{Corollary~\ref{corol:F_trans_parts}}
  \EndIf
  \EndFor
\end{algorithmic}
\end{algorithm}
\end{theorem}

\section{CE Analysis of the BSC and Bhattacharyya Parameters}~\label{sec:discuss}

In this section, we analyze $ \mathfrak{a}\equiv W_{\operatorname{BSC}(p)}$ via Theorem~\ref{thm:CE}, where the BSC serves as a representative discrete BMS channel. We explicitly derive analytic expressions for the Bhattacharyya parameters of all bit-channels of $W_{\operatorname{BSC}(p)}$ for levels $k=3,4$. CE reveals prefix/suffix-based recursions for the BSC, reducing the analysis of all bit-channels to a specific subclass, whose simplest case is further analyzed combinatorially and probabilistically, enabling a series-expansion approach that bypasses exponential complexity.

\subsection{Prefix-\texorpdfstring{$\boxast$}{boxast}- and suffix-\texorpdfstring{$\varoast$}{varoast}-based recursions}
\begin{exmp}
At level $k=3$, the Bhattacharyya parameters are given by
\begin{align}
    \mathfrak{B}(|\mathfrak{a}^{\boxast 3}|) &= (1- M_{16})^{1/2}, \\
    \Ba[\boxast 2 \varoast] &= 1-M_8, \\
    \mathfrak{B}(|\mathfrak{a}^{\boxast \varoast \boxast}|) &= \frac{1}{4}\left(1-M_4\right)\left(\left(M_{8}+6 M_4+1\right)^{1 / 2}+M_4+3\right), \\
    \Ba[\boxast \varoast 2] &= (1-M_4)^2, \\
    \Ba[\varoast \boxast 2]  &= \left(S_2^8-M_8\right)^{1/2}+1-S_2^4, \\
    \Ba[\varoast \boxast \varoast] &= \left(\left(S_2^4-M_4\right)^{1/2}+4C\bar{C}\right)^2, \\
    \Ba[\varoast 2 \boxast] &= 32 \sqrt{2} C^3 S_4^{1/2} + 16 C^{2} S_2^{1/2} S_6^{1/2}  +2 \sqrt{2} C^2 S_8^{1/2}+12 C^2-36 C^4, ~\label{eq:Ba_110} \\
    \Ba[\varoast 3] &= 256C^4.
\end{align}
The Bhattacharyya parameters at level $k=4$ are provided in Section~\ref{sec:BSC_34}. Their evaluations across different values of $p\in(0,1/2)$ are illustrated in Fig.~\ref{fig:bhatt_k_3_20},~\ref{fig:bhatt_k_4_20},~\ref{fig:bhatt_k_3_2k} and~\ref{fig:bhatt_k_4_2k}.
\end{exmp}
\begin{proof}
Please refer to Section~\ref{sec:BSC_34}.
\end{proof}
\begin{remark}
Note that the $\boxast$-convolution upon $W_{\operatorname{BSC}}(p)$ results again in another $W_{\operatorname{BSC}}(p^{\boxast})$ with crossover $p^{\boxast}=2C$, cf. Example~\ref{emp:bsc_check}. We denote this identification under the $\boxast$-convolution via the parameter mapping $2C\gets p$. One further verifies that for every $i\in\mathbb{Z}_{>0}$, $M_{2i}\gets M_i$, $S_2^i+(2C)^i\gets S_i$, $S_2^i-(2C)^i\gets D_i$ and $(1-M_4)/4\gets C$. Let $\Ba[\ldots]\triangleq\mathfrak{B}_{|\mathfrak{a^{\ldots}}|}(p)$. Therefore, the Bhattacharyya parameters of $\mathfrak{a}^{\boxast\ldots}$ with \textit{prefix}-$\boxast$ can be obtained from those of $\mathfrak{a}^{\ldots}$ via the corresponding parameter mapping, which suggests that the Bhattacharyya parameters of the \textit{first} $2^{k-1}$ bit-channels at level $k$ can be recursively derived from those of \textit{all} bit-channels at level $k-1$, i.e., $\Ba[\boxast \ldots]\triangleq\mathfrak{B}_{|\mathfrak{a^{\boxast \ldots}}|}(p)=\mathfrak{B}_{|\mathfrak{a^{\ldots}}|}(2C)$. On the other hand, the Bhattacharyya parameters of $\mathfrak{a}^{\ldots\varoast}$ with \textit{suffix}-$\varoast$ can be obtained by \textit{squaring} those of $\mathfrak{a}^{\ldots}$, cf.~\cite[Lemma 4.64]{richardson2008modern}, which suggests that the Bhattacharyya parameters of the $2^{k-1}$ bit-channels with \textit{even} indices at level $k$ can be directly derived from those of \textit{all} bit-channels at level $k-1$, i.e., $\Ba[\ldots\varoast]\triangleq \mathfrak{B}_{|\mathfrak{a^{\ldots\varoast}}|}(p) = \mathfrak{B}^2_{|\mathfrak{a^{\ldots}}|}(p)$. Such recursive structure of Bhattacharyya parameters is illustrated in Fig.~\ref{fig:Bhatt_recur}. For example,
\begin{align}
     Z(W_{16}^{(4)})\equiv \Ba[\boxast 2 \varoast 2] \triangleq \mathfrak{B}_{|\mathfrak{a}^{\boxast 2 \varoast 2}|}(p) &= \mathfrak{B}_{|\mathfrak{a}^{\boxast \varoast 2}|}(2C) = \mathfrak{B}^2_{|\mathfrak{a}^{\boxast 2 \varoast}|}(p) = \left(1-M_8\right)^2, \\
     Z(W_{16}^{(5)})\equiv \Ba[\boxast \varoast \boxast 2] \triangleq \mathfrak{B}_{|\mathfrak{a}^{\boxast \varoast \boxast 2}|}(p) &= \mathfrak{B}_{|\mathfrak{a}^{\varoast \boxast 2}|}(2C) = 1-\frac{1}{16}\left(1+M_4\right)^4+\frac{1}{16}\left(\left(1+M_4\right)^8-256 M_{16}\right)^{1/2}, \\
     Z(W_{16}^{(10)})\equiv \Ba[\varoast \boxast 2 \varoast] \triangleq \mathfrak{B}_{|\mathfrak{a}^{\varoast \boxast 2 \varoast}|}(p) &= \mathfrak{B}^2_{|\mathfrak{a}^{\varoast \boxast 2}|}(p) = \left(\left(S_2^8-M_8\right)^{1/2}+1-S_2^4\right)^2.
\end{align}

\begin{figure}
    \centering
    \begin{subfigure}[t]{0.48\linewidth}
        \centering
        \includegraphics[width=\linewidth]{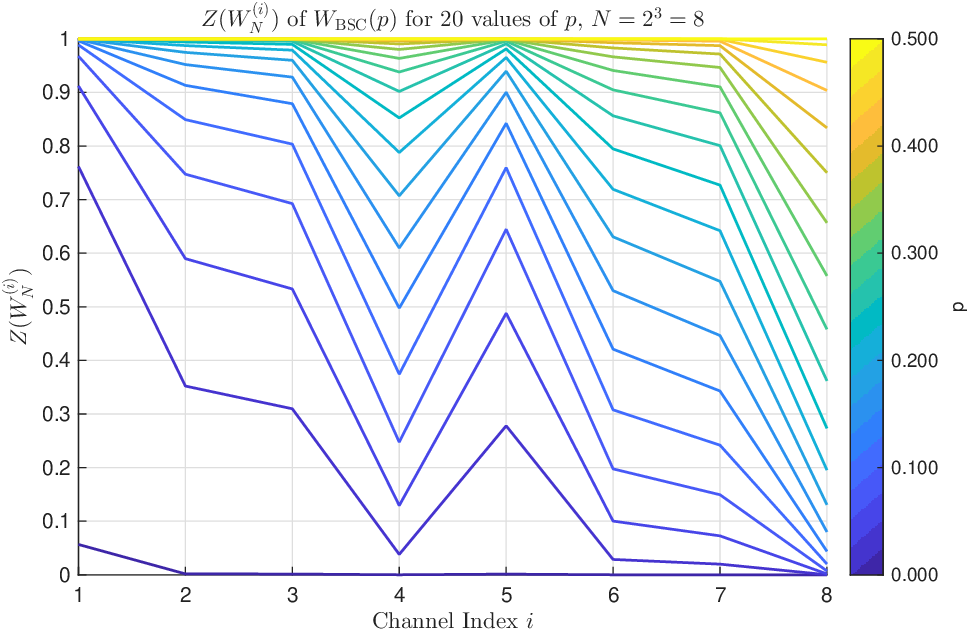}
        \caption{$k=3$}
        \label{fig:bhatt_k_3_20}
    \end{subfigure}
    \hfill
    \begin{subfigure}[t]{0.48\linewidth}
        \centering
        \includegraphics[width=\linewidth]{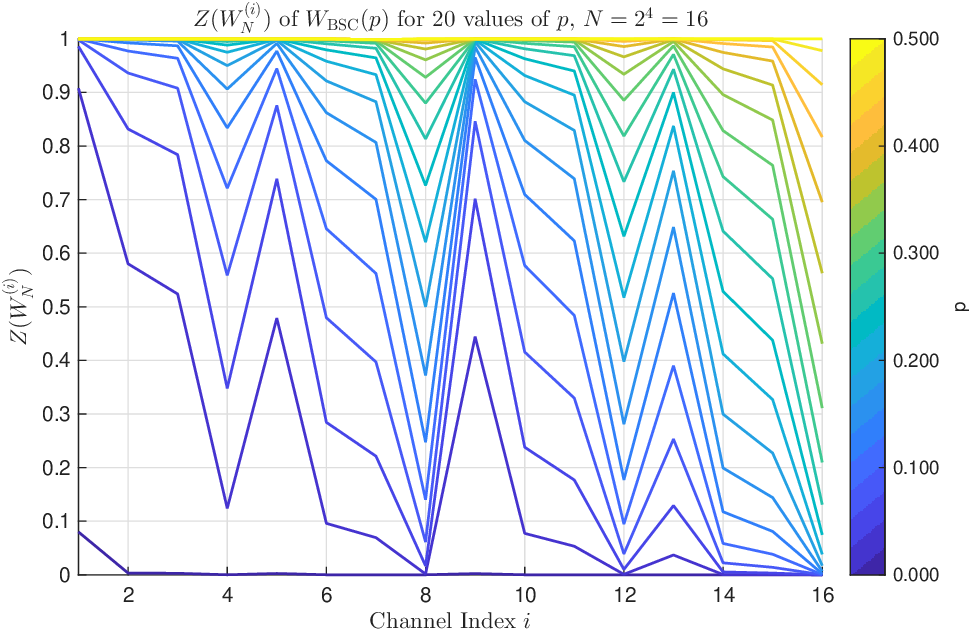}
        \caption{$k=4$}
        \label{fig:bhatt_k_4_20}
    \end{subfigure}
    \caption{Bhattacharyya parameters of bit-channels of $W_{\operatorname{BSC}(p)}$ for $20$ values of $p\in(0,1/2)$.}
    \label{fig:bhatt_k_3_4_20}
\end{figure}

\begin{figure}
    \centering
    \begin{subfigure}[t]{0.48\linewidth}
        \centering
        \includegraphics[width=\linewidth]{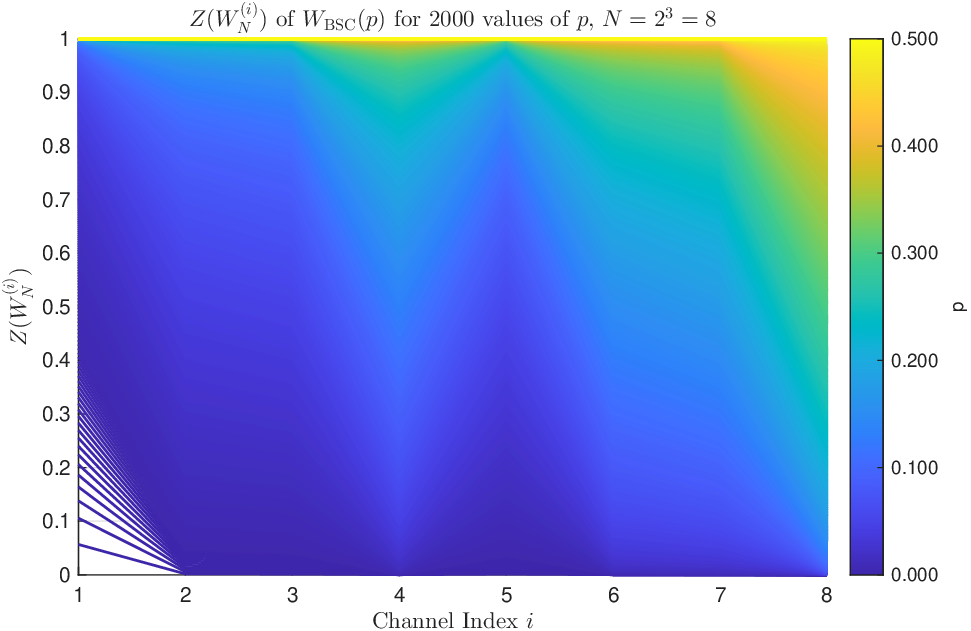}
        \caption{$k=3$}
        \label{fig:bhatt_k_3_2k}
    \end{subfigure}
    \hfill
    \begin{subfigure}[t]{0.48\linewidth}
        \centering
        \includegraphics[width=\linewidth]{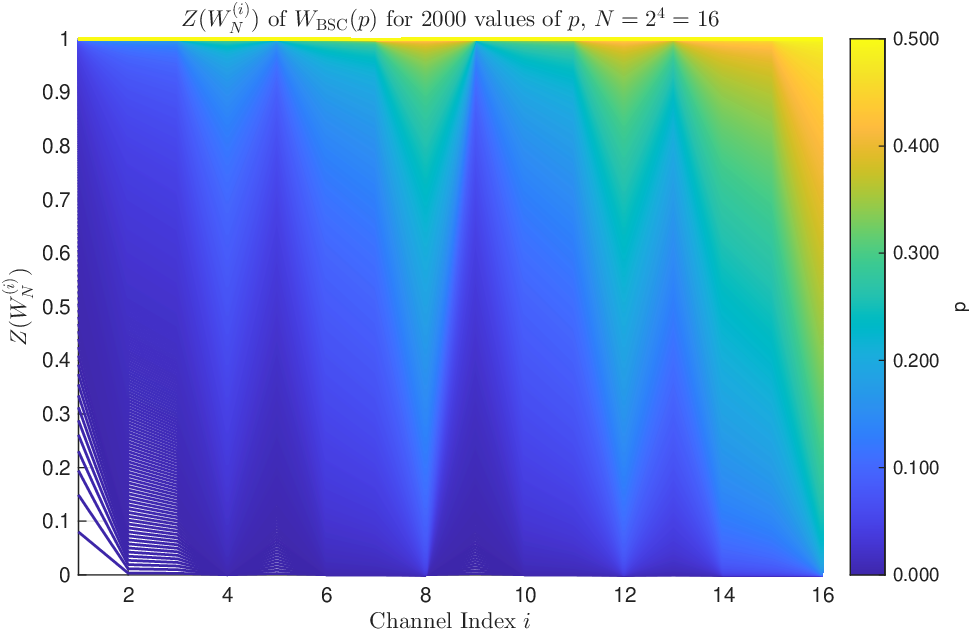}
        \caption{$k=4$}
        \label{fig:bhatt_k_4_2k}
    \end{subfigure}
    \caption{Bhattacharyya parameters of bit-channels of $W_{\operatorname{BSC}(p)}$ for $2000$ values of $p\in(0,1/2)$.}
    \label{fig:bhatt_k_3_4_2k}
\end{figure}

\begin{figure}[!ht]
\centering
\begin{tikzpicture}[x=2.2cm,y=0.3cm]
\tikzset{
  dot/.style={circle,fill=black,inner sep=1.2pt},
  hdot/.style={circle,draw=black,fill=none,inner sep=1.2pt},
  solid arrow/.style={arrows = {-Latex[width'=0pt .5, length=5pt]}, shorten >=.25pt, shorten <=.25pt},
  rsolid arrow/.style={arrows = {-Latex[width'=0pt .5, length=5pt]}, draw=red, line width=0.8pt, opacity=0.5, shorten >=.25pt, shorten <=.25pt},
  dash arrow/.style={arrows = {-Latex[width'=0pt .5, length=5pt]}, dash pattern=on 2pt off 3pt, shorten >=.25pt, shorten <=.25pt},
  lab/.style={right=1pt, font=\scriptsize}
}
\def\L{15}

\draw (0,0) -- (0,\L);
\draw (1,0) -- (1,\L);
\draw (2,0) -- (2,\L);
\draw (3,0) -- (3,\L);

\foreach \i in {1,2} {
  \node[dot] (L1-\i) at (0,{ \L - (\i-0.5)*\L/2 }) {};
}

\foreach \i in {1,2,4} {
  \node[dot] (L2-\i) at (1,{ \L - (\i-0.5)*\L/4 }) {};
}
\node[hdot] (L2-3) at (1,{ \L - (3-0.5)*\L/4 }) {};

\foreach \i in {1,2,3,4,6,8} {
  \node[dot] (L3-\i) at (2,{ \L - (\i-0.5)*\L/8 }) {};
}
\node[hdot] (L3-5) at (2,{ \L - (5-0.5)*\L/8 }) {};
\node[hdot] (L3-7) at (2,{ \L - (7-0.5)*\L/8 }) {};

\foreach \i in {1,2,...,8,10,12,14,16} {
  \node[dot] (L4-\i) at (3,{ \L - (\i-0.5)*\L/16 }) {};
}
\node[hdot] (L4-9) at (3,{ \L - (9-0.5)*\L/16 }) {};
\node[hdot] (L4-11) at (3,{ \L - (11-0.5)*\L/16 }) {};
\node[hdot] (L4-13) at (3,{ \L - (13-0.5)*\L/16 }) {};
\node[hdot] (L4-15) at (3,{ \L - (15-0.5)*\L/16 }) {};

\node[lab] at (L1-1) {$0$};
\node[lab] at (L1-2) {$1$};

\node[lab] at (L2-1) {$00$};
\node[lab] at (L2-2) {$01$};
\node[lab] at (L2-3) {$10$};
\node[lab] at (L2-4) {$11$};

\node[lab] at (L3-1) {$000$};
\node[lab] at (L3-2) {$001$};
\node[lab] at (L3-3) {$010$};
\node[lab] at (L3-4) {$011$};
\node[lab] at (L3-5) {$100$};
\node[lab] at (L3-6) {$101$};
\node[lab] at (L3-7) {$110$};
\node[lab] at (L3-8) {$111$};

\node[lab] at (L4-1)  {$0000 \ (1)$};
\node[lab] at (L4-2)  {$0001 \ (2)$};
\node[lab] at (L4-3)  {$0010 \ (3)$};
\node[lab] at (L4-4)  {$0011 \ (4)$};
\node[lab] at (L4-5)  {$0100 \ (5)$};
\node[lab] at (L4-6)  {$0101 \ (6)$};
\node[lab] at (L4-7)  {$0110 \ (7)$};
\node[lab] at (L4-8)  {$0111 \ (8)$};
\node[lab] at (L4-9)  {$1000 \ (9)$};
\node[lab] at (L4-10) {$1001 \ (10)$};
\node[lab] at (L4-11) {$1010 \ (11)$};
\node[lab] at (L4-12) {$1011 \ (12)$};
\node[lab] at (L4-13) {$1100 \ (13)$};
\node[lab] at (L4-14) {$1101 \ (14)$};
\node[lab] at (L4-15) {$1110 \ (15)$};
\node[lab] at (L4-16) {$1111 \ (16)$};

\draw[solid arrow] (L1-1) -- (L2-1);
\draw[solid arrow] (L1-2) -- (L2-2);

\draw[solid arrow] (L2-1) -- (L3-1);
\draw[solid arrow] (L2-2) -- (L3-2);
\draw[solid arrow] (L2-3) -- (L3-3);
\draw[solid arrow] (L2-4) -- (L3-4);

\draw[solid arrow] (L3-1) -- (L4-1);
\draw[solid arrow] (L3-2) -- (L4-2);
\draw[solid arrow] (L3-3) -- (L4-3);
\draw[solid arrow] (L3-4) -- (L4-4);
\draw[solid arrow] (L3-5) -- (L4-5);
\draw[solid arrow] (L3-6) -- (L4-6);
\draw[solid arrow] (L3-7) -- (L4-7);
\draw[solid arrow] (L3-8) -- (L4-8);

\draw[dash arrow] (L1-1) -- (L2-2);
\draw[dash arrow] (L1-2) -- (L2-4);

\draw[dash arrow] (L2-1) -- (L3-2);
\draw[dash arrow] (L2-2) -- (L3-4);
\draw[dash arrow] (L2-3) -- (L3-6);
\draw[dash arrow] (L2-4) -- (L3-8);

\draw[dash arrow] (L3-1) -- (L4-2);
\draw[dash arrow] (L3-2) -- (L4-4);
\draw[dash arrow] (L3-3) -- (L4-6);
\draw[dash arrow] (L3-4) -- (L4-8);
\draw[dash arrow] (L3-5) -- (L4-10);
\draw[dash arrow] (L3-6) -- (L4-12);
\draw[dash arrow] (L3-7) -- (L4-14);
\draw[dash arrow] (L3-8) -- (L4-16);




\end{tikzpicture}
\caption{Recursions of Bhattacharyya parameters: solid arrows (prefix-$\boxast$) denote parameter mapping, while dashed lines (suffix-$\varoast$) denote squaring.}
\label{fig:Bhatt_recur}
\end{figure}
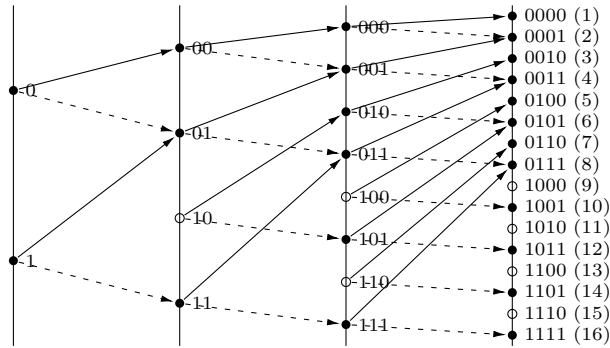
\end{remark}
However, this recursive structure relies on the availability of the Bhattacharyya parameters of some \textit{base} channels (of size $2^{m-2}$ at level $m$), which have neither prefix-$\boxast$ nor suffix-$\varoast$ (or equivalently, they have both prefix-$\varoast$ and suffix-$\boxast$), and are indicated by the hollow dots in Fig.~\ref{fig:Bhatt_recur}. In general, the exact Bhattacharyya parameters of these \textit{base} channels should first be obtained via the CE prescribed in Theorem~\ref{thm:CE}, and typically exhibit a complicated structure, e.g.,~\eqref{eq:Ba_110},~\eqref{eq:Ba_1010},~\eqref{eq:Ba_1100}, and~\eqref{eq:Ba_1110}. Consequently, to derive the Bhattacharyya parameters of \textit{all} bit-channels at level $k$, the CE should be applied to the \textit{last} $2^{m-1}$ bit-channels with \textit{prefix}-$\varoast$ at each level $m$ up to $k$, so as to preserve the characterization in terms of $\Fai(s)$, $\Ga(\nu)$ and boundary atoms, from which $\Ba$ is derived.

\subsection{Prefix-\texorpdfstring{$\varoast$}{varoast}- and suffix-\texorpdfstring{$\boxast$}{boxast}-based analysis}
In this subsection, we analyze the simplest form among base channels $\mathfrak{a}^{\varoast \ldots \boxast}$, i.e., $\mathfrak{a}^{\varoast q \boxast \ell}$ with $q\geq 2$ and $\ell\geq0$. After applying the $q$-fold $\varoast$-convolution to $\mathfrak{a}\equiv W_{\operatorname{BSC}(p)}$, we obtain characterizations of the resultant bit-channels, as stated next.

\begin{lemma}~\label{lemma:q_base}
Let $q\geq 2$, then
\begin{align}
    \Ga[\varoast q](\nu)&=\sum_{i=0}^{Q-1} \binom{2Q}{i} C^{i}S_{r_i}^{1-\nu}D_{r_i}^{\nu}, ~\label{eq:G_q_base} \\
    \Fai[\varoast q](s)&=C^Q\sum_{i=0}^{Q-1}\binom{2Q}{i} \left(S_{r_i}-D_{r_i}\right)^{s-1/2}\left(S_{r_i}+D_{r_i}\right)^{1/2-s}, ~\label{eq:F_q_base} \\
    \Ba[\varoast q]&=2^{2Q}C^Q, ~\label{eq:Ba_q_base}
\end{align}
where $r_i\triangleq2Q-2i$ with $Q\triangleq 2^{q-1}$ and $i=0,\ldots, Q-1$.
\end{lemma}
\begin{proof}
Please refer to Section~\ref{sec:q_base}.
\end{proof}

Next, applying $\ell$ successive $\boxast$-convolutions yields the following characterizations.

\begin{lemma}~\label{lemma:GF_bin_base}
Let $q\geq 2$, and $\ell\geq0$, then
\begin{gather}
    \Ga[\varoast q \boxast \ell](\nu)= \sum_{\mathbf{j}\in\mathcal{J}_{QL}}R_{\mathbf{j}} \cdot C^{\alpha_\mathbf{j}}\prod_{i=0}^{Q-1}\left(S_{r_i}^{j_{r_i}(1-\nu)}D_{r_i}^{j_{r_i}\nu
    }\right)=\left(\sum_{i=0}^{Q-1} \binom{2Q}{i} C^{i}S_{r_i}^{1-\nu}D_{r_i}^{\nu} \right)^L, ~\label{eq:G_10_base} \\
    \Fai[\varoast q \boxast \ell](s)= \frac{1}{2}\sum_{\mathbf{j}\in\mathcal{J}_{QL}}R_{\mathbf{j}} \cdot C^{\alpha_\mathbf{j}}\left(\prod_{i=0}^{Q-1}S_{r_i}^{j_{r_i}}-\prod_{i=0}^{Q-1}D_{r_i}^{j_{r_i}}\right)^{s}\left(\prod_{i=0}^{Q-1}S_{r_i}^{j_{r_i}}+\prod_{i=0}^{Q-1}D_{r_i}^{j_{r_i}}\right)^{1-s},~\label{eq:F_10_base}
\end{gather}
where $\mathcal{J}_{QL}\triangleq \left\{ (j_{r_{Q-1}},\ldots, j_{r_{0}})\in\mathbb{Z}_{\geq 0}^Q: \sum_{i=0}^{Q-1}j_{r_i}=L\triangleq 2^{\ell}, Q\triangleq 2^{q-1}\right\}$ with $\left|\mathcal{J}_{QL}\right|=\binom{L+Q-1}{Q-1}$, $\alpha_{\mathbf{j}} \triangleq \sum_{i=0}^{Q-1}i\cdot j_{r_i}$, $r_i\triangleq 2Q-2i$, and $R_{\mathbf{j}} \triangleq \frac{L!}{\prod_{i=0}^{Q-1}j_{r_i}!}\prod_{i=0}^{Q-1}\binom{2Q}{i}^{j_{r_i}}$.
\end{lemma}

\begin{proof}
Please refer to Section~\ref{sec:GF_bin_base}.
\end{proof}

\begin{remark}
Since $\Ga[\varoast q \boxast \ell](\nu)$ has the structure of a \textit{generating function}~\cite[Section 24.1.2]{abramowitz1965handbook}, the \textit{multinomial coefficient} $(L; j_{r_0},\ldots,j_{r_{Q-1}})$ appears in $R_{\mathbf{j}}$ and, together with $\alpha_{\mathbf{j}}$, forms the coefficients of both $\Fai[\varoast q \boxast \ell](s)$ and $\Ga[\varoast q \boxast \ell](\nu)$. This brings difficulty in deriving $\Ba[\varoast q \boxast \ell]$ at level $q+\ell$ via Procedure~\ref{alg:chan_polar}, as the cardinality $\left|\mathcal{J}_{QL}\right|$ grows as $\Theta\left(2^{\ell(Q-1)}\right)$ and becomes intractable for large $\ell$.
\end{remark}

In what follows, we show that by exploiting this multinomial structure, $\Ba[\varoast q \boxast \ell]$ at level $q+\ell$ can be linked, via a series expansion, to infinitely many even moments generated by $\Ga[\varoast q](\nu)$ at level $q$.

\begin{prop}~\label{prop:Ba_G_L}
Let $Q \triangleq 2^{q-1}$, $L\triangleq 2^\ell$ with $q\geq 2$, $\ell\geq0$. Let $X_Q \triangleq \tanh(\lvert Q-K \rvert \cdot \lvert\Delta \rvert)\in[0,1)$ with $K\sim\mathrm{Bin}(2Q,p)$. Let $X_{QL}\triangleq \prod_{i=1}^L X_{Q,(i)}$ where $X_{Q,(1)}, \ldots, X_{Q,(L)}$ are i.i.d. copies of $X_Q$, then
\begin{equation}
    \Ba[\varoast q \boxast \ell] = 1-\sum_{\kappa=1}^{\infty} \beta_\kappa \Ga[\varoast q]^L(2\kappa) = \mathbb{E}\left[ \sqrt{1- X_{QL}^2} \right], ~\label{eq:B2_bin_base}
\end{equation}
where $\beta_\kappa=(2\kappa-3)!!/(2^\kappa \kappa!)$.
\end{prop}

\begin{proof}
Please refer to Section~\ref{sec:Ba_G_L}.
\end{proof}

\begin{remark}
The second proof of the second equality in~\eqref{eq:B2_bin_base} provides an alternative interpretation of the power series expansion: it arises not only from the Taylor expansion of $\sqrt{1-z}$, but also from the \textit{moment generating function} (MGF) of a product of i.i.d. random variables $X_{Q,(i)}$ induced by repeated $\boxast$-convolutions.
\end{remark}

Proposition~\ref{prop:Ba_G_L} shows that applying $\ell$-fold $\boxast$-convolution only corresponds to raising each even moment to the power $L$. This brings a significant benefit in evaluating $\Ba[\varoast q \boxast \ell]$ at level $q+\ell$, provided that the expression of $\Ga[\varoast q](\nu)$ at level $q$ is available, thereby bypassing the exponential grow of $\left|\mathcal{J}_{QL}\right|$. Moreover, tighter upper bounds of $\Ba[\varoast q \boxast \ell]$ can be obtained systematically by truncating the series expansion at higher order. Since $\Ga[\varoast q](2\kappa) = \mathbb{E}[X_Q^{2\kappa}]$, it also reveals that $\Ba[\varoast q \boxast \ell]$ depends not on the scalar $\Ba[\varoast q]=2^{2Q}C^Q$, but on all even moments of $X_Q$.


\section{Conclusion}
In this paper, we developed CE as a complex-analytic framework for finite-blocklength channel polarization on discrete BMS channels. By viewing the Bhattacharyya parameter as a real-valued instance of a broader class of complex-valued channel functionals, CE yields a systematic procedure for deriving analytic expressions for the Bhattacharyya parameters of the bit-channels of a given discrete BMS channel at arbitrary polarization levels. Beyond its importance in exact channel-dependent code construction, the framework provides structural insights, including further evidence of the extremality of the BEC and BSC, as well as new channel-dependent recursions for a class of BSC bit-channels.

\section{Appendices}

\subsection{Proof of Lemma~\ref{lemma:var_node_holo}}~\label{sec:var_node_holo}

Observe that~\eqref{eq:var_node_double_LT} is a form of the \textit{bilateral Laplace transform}, of which the region of convergence (ROC) consists of those values of $s$ for which $\mathrm{a}(x) \mathrm{e}^{-s x}$ is \textit{absolutely integrable}, i.e., $\int_{-\infty}^{+\infty} |\mathrm{a}(x) \mathrm{e}^{-s x}| \mathrm{~d} x<\infty$. Recall that if $\mathrm{a}$ is a symmetric $L$-density, then $\mathrm{a}(-x)=\mathrm{e}^{-x}\mathrm{a}(x)$ for all $x\in\mathbb{R}$, cf.~\cite[Definition 4.11]{richardson2008modern}. Therefore, $\mathcal{F}_{\mathrm{a}}(s)=\int_{0}^{+\infty}\mathrm{a}(x)(\mathrm{e}^{-sx}+\mathrm{e}^{(s-1)x})\mathrm{d}x$. Let $s\triangleq \sigma+\mathrm{i} t$. The integrand is absolutely integrable if and only if $\int_{0}^{+\infty}\mathrm{a}(x)|\mathrm{e}^{-sx}+\mathrm{e}^{(s-1)x}|\mathrm{d}x<\infty$. If $\sigma\in[0,1]$, then for all $x\geq 0$, $|\mathrm{e}^{-sx}|=\mathrm{e}^{-\sigma x}\leq 1$ and $|\mathrm{e}^{(s-1)x}|=\mathrm{e}^{(\sigma-1) x}\leq 1$, thus, $\int_{0}^{+\infty}\mathrm{a}(x)|\mathrm{e}^{-sx}+\mathrm{e}^{(s-1)x}|\mathrm{d}x\leq 2\int_0^{\infty} \mathrm{a}(x) d x<\infty$. Therefore, $\{s \in \mathbb{C}: 0<\Re s<1\}$ is the \textit{strip of absolute convergence} of $\mathcal{F}_{\mathrm{a}}(s)$, over which $\mathcal{F}_{\mathrm{a}}(s)$ is holomorphic.

\subsection{Proof of Lemma~\ref{lemma:check_node_holo}}~\label{sec:check_node_holo}

Let $\mu=0$ and recall that if $a$ is a symmetric $G$-density, then $a(0,y)=a(1,y)\coth(y/2)$ for all $y\in [0,+\infty]$, cf.~\cite[(4.21)]{richardson2008modern}. Therefore, $\mathcal{G}_a(0, \nu)=\int_0^{+\infty}g(y)\cdot \mathrm{e}^{-\nu y}\mathrm{d}y$, which is a form of the \textit{unilateral Laplace transform}, with $g(y)\triangleq a(1, y)(\operatorname{coth}(y / 2)+1)=2\cdot a(1,y)/(1-\mathrm{e}^{-y})$. A necessary condition for existence of such integral is that $g(y)$ must be \textit{locally integrable} on $y\in [0,\infty)$. Indeed,
\begin{equation}~\label{eq:check_node_exist}
    \begin{aligned}
        \int_{0}^{+\infty}|g(y)|\mathrm{~d}y = 2 \int_0^{+\infty} \frac{a(1,y)}{1-\mathrm{e}^{-y}} \mathrm{~d} y \leq \frac{2}{1-\mathrm{e}^{-1}}\left(\int_0^1 a(1,y) /y \mathrm{~d} y+\int_1^{\infty} a(1, y) \mathrm{~d} y\right),
    \end{aligned}
\end{equation}
which follows from that $1-\mathrm{e}^{-x} \geq\left(1-\mathrm{e}^{-1}\right) \cdot \min (1, x)$ for all $x\geq 0$. Note that for the $G$-distribution~\cite[p.183]{richardson2008modern}, i.e., $A(s,y)=A(0,y)\cdot \mathbbm{1}_{\{s=0\}}+A(1,y)\cdot \mathbbm{1}_{\{s=1\}}$, we have $A(0,0)\geq 0$ and $A(1,0)=0$. These conditions correspond, respectively, to $\lim_{x\to+\infty}\mathsf{A}(x)\leq 1$ and $\lim_{x\to-\infty}\mathsf{A}(x)=0$ for $\mathsf{A}\in\mathcal{A}_L$, cf.~\cite[p.178]{richardson2008modern}. In another words, $a(0,y)$ has a Dirac mass at $y=0$ of magnitude $A(0,0)$, while $a(1,y)$ has no Dirac mass at $y=0$. Moreover, both $a(0,y)$ and $a(1,y)$ have Dirac masses at $y=+\infty$ of magnitude $\frac{1}{2}\left(1-\lim _{x \rightarrow+\infty} A(0, x)-\lim _{x \rightarrow+\infty} A(1, x)\right)$. For the second integral in the RHS of~\eqref{eq:check_node_exist},
\begin{equation}~\label{eq:check_int2}
    \int_1^{+\infty} a(1, y) \mathrm{~d} y < +\infty,
\end{equation}
follows from $a$ being a $G$-density. For the first integral in the RHS of~\eqref{eq:check_node_exist}, since $a(1,y)$ has no Dirac mass at $y=0$,
\begin{equation}~\label{eq:check_int1}
    \int_0^1 a(1,y) /y \mathrm{~d} y \stackrel{(a)}{=} \int_0^1 a(0,y) \tanh(y/2) /y \mathrm{~d} y \stackrel{(b)}{\leq} \frac{1}{2}\int_0^1 a(0,y) \mathrm{~d} y < +\infty,
\end{equation}
where $(a)$ follows from the symmetry $a(1,y)=a(0,y)\tanh(y/2)$ for every $y>0$, while $(b)$ uses $\tanh (x) \leq x$ for every $x\geq 0$. Combining~\eqref{eq:check_int2} and~\eqref{eq:check_int1} yields $\int_0^{+\infty}|g(y)| \mathrm{~d} y<+\infty$. Consequently, for every $\nu=\sigma+\mathrm{i} t$ with $\sigma \geq 0$,
\begin{equation}
    \int_0^{+\infty}|g(y)\cdot \mathrm{e}^{-\nu y}|\mathrm{~d}y = \int_0^{+\infty}|g(y)|\cdot \mathrm{e}^{-\sigma y} \mathrm{~d} y \leq \int_0^{+\infty}|g(y)| \mathrm{~d} y<+\infty,
\end{equation}
thus, $\mathcal{G}_a(0, \nu)$ converges absolutely for $\Re\nu\geq 0$ and is holomorphic for $\Re\nu> 0$. For $\mu=\mathrm{i}\pi$, $\mathcal{G}_{a}(\mathrm{i} \pi, \nu)=\mathcal{G}_{a}(0, \nu+1)$ by symmetry~\cite[p. 200]{richardson2008modern}. Since $\Re(\nu+1)>0$ whenever $\Re \nu \geq 0$, it follows that $\mathcal{G}_{\mathfrak{a}}(\mathrm{i} \pi, \nu)$ also converges absolutely for $\Re \nu \geq 0$.



\subsection{Proof of Lemma~\ref{lemma:f_a_analy_cont}}~\label{sec:f_a_analy_cont}
Since $f_{|\mathfrak{a}|}(z)=\sum_{i=1}^{n-1} \alpha_i\cdot z\cdot \delta\left(z-z_i\right)$ by Lemma~\ref{lemma:absD_dens}, atoms of $|\mathfrak{a}|(z)$ at $z_0\equiv 0$ and $z_n\equiv 1$ (if present) are not reflected in $f_{|\mathfrak{a}|}(z)$. The expression $\mathscr{G}_{|\mathfrak{a}|}(\nu)=\sum_{i=1}^{n-1}\alpha_i z_i^\nu$ follows from the fact that the Mellin transform of $z\cdot \delta(z-a)$ with $a>0$ equals $a^{\nu}$ for every $\nu\in\mathbb{C}$, as a consequence of the translational property of the Mellin transform~\cite[(3.1.7)]{paris2001asymptotics}. Since each $z_i^\nu=\mathrm{e}^{\nu \log z_i}$ for every $z_i\in(0,1)$, with $i=1,\ldots,n-1$, is entire on $\mathbb{C}$, $\mathscr{G}_{|\mathfrak{a}|}(\nu)$ is entire too.

\subsection{Proof of Lemma~\ref{lemma:G_comp_evo}}~\label{sec:G_comp_evo}
Let $\mathrm{a}$ and $\mathrm{b}$ be independent symmetric $L$-densities with $X\sim \mathrm{a}$ and $Y\sim \mathrm{b}$, then $\tanh (|X| / 2) \sim|\mathfrak{a}|$ and $\tanh (|Y| / 2) \sim|\mathfrak{b}|$. By property of the check-domain $\boxast$-convolution~\cite[p.195]{richardson2008modern}, $Z=2 \tanh ^{-1}(\tanh (X / 2) \tanh (Y / 2)) \sim \mathrm{a} \boxast \mathrm{b}$, hence for $\mathrm{a}=\mathrm{b}$, thus, $\mathfrak{a}=\mathfrak{b}$, i.e., $\mathfrak{a} \boxast \mathfrak{b} = \mathfrak{a}^\boxast$, we have $\tanh (|Z| / 2) = \tanh (|X| / 2) \tanh (|Y| / 2)\sim |\mathfrak{a}^{\boxast}|$. Recall Corollary~\ref{corol:check_node_double_LT}, i.e., $\mathcal{G}_{|\mathfrak{a}|}(0, \nu)=\int_0^1|\mathfrak{a}|(z) z^\nu \mathrm{d} z$, then
\begin{equation}
    \begin{aligned}
        &\mathcal{G}_{|\mathfrak{a}^{\boxast}|}(0, \nu) = \int_0^1|\mathfrak{a^{\boxast}}|(z) z^\nu \mathrm{d} z \\
        =& \int_0^1\int_0^1|\mathfrak{a}|(x) |\mathfrak{a}|(y)\cdot (x y)^\nu \mathrm{~d} x\mathrm{~d} y  = \left(\int_0^1|\mathfrak{a}|(x) x^\nu \mathrm{d} x\right)^2  \\
        =& \sum_{i, j} \alpha_i \alpha_j\left(z_i z_j\right)^\nu=\left(\sum_{i=0}^{n} \alpha_i z_i^\nu\right)^2 = \mathcal{G}^2_{|\mathfrak{a}|}(0, \nu).
    \end{aligned}
\end{equation}
By Corollary~\ref{corol:G2G_connect}, we write $\mathcal{G}_{|\mathfrak{a}^{\boxast}|}(0, \nu)= \mathscr{G}_{|\mathfrak{a}^{\boxast}|}(\nu)+ \alpha^\boxast_0 \cdot 0^\nu + \alpha^\boxast_{n}$, which is holomorphic for $\Re \nu>0$ by Lemma~\ref{lemma:check_node_holo}. In the sequel, we discuss $\alpha^\boxast_0$, $\alpha^\boxast_{n}$ and $\mathscr{G}_{|\mathfrak{a}^{\boxast}|}(\nu)$ individually. The mass of the atom at $z=0$ after the $\boxast$-convolution equals the total weight of all pairs $(i,j)$ with $z_iz_j=0$, thus $\alpha^\boxast_0 = \alpha_0^2 + 2\alpha_0\cdot \sum_{i\geq 1}\alpha_i=\alpha_0^2 + 2\alpha_0(1-\alpha_0)=2\alpha_0-\alpha_0^2$, while the mass of the atom at $z=1$ after the $\boxast$-convolution equals the total weight of all pairs $(i,j)$ with $z_iz_j=1$, which is only possible with $z_i=z_j=1$, i.e., $\alpha^\boxast_{n} = \alpha^2_{n}$. If there is no atom at $z=1$, the interior part after the $\boxast$-convolution is $\mathscr{G}^2_{|\mathfrak{a}|}(\nu)$. If an atom at $z=1$ is present, its interaction with the interior atoms contributes $\alpha_n \alpha_i z_i^\nu$ for $i=1,\ldots,n-1$, each occurring twice. Consequently, $\mathscr{G}_{|\mathfrak{a}^{\boxast}|}(\nu)=\mathscr{G}^2_{|\mathfrak{a}|}(\nu)+2 \alpha_n \mathscr{G}_{|\mathfrak{a}|}(\nu) \cdot \mathbbm{1}_{\{\alpha_n > 0\}}$, which admits an \textit{analytic continuation} to an \textit{entire} function on $\mathbb{C}$.

\subsection{Proof of Lemma~\ref{lemma:f_a_w_analy_cont}}~\label{sec:f_a_w_analy_cont}

To prove Lemma~\ref{lemma:f_a_w_analy_cont}, we need the following Lemma~\ref{lemma:delta_comp}.

\begin{lemma}[Sec.II.2.5. in~\cite{gelfand1969generalized}]~\label{lemma:delta_comp}
Let $h(x)\in C^1$ be an \textit{continuously differentiable} function, having the simple roots $x_k$ with $h'(x_k)\neq 0$ for each $k$, then for every \textit{compactly supported and infinitely differentiable} test functions $\varphi(x)\in C^\infty_c$,
\begin{equation}
    \int_{-\infty}^{\infty} \delta(h(x)) \varphi(x) \mathrm{~d} x=\sum_k \frac{\varphi\left(x_k\right)}{\left|h^{\prime}\left(x_k\right)\right|},
\end{equation}
where the sum is taken over all simple roots of $h(x)$. It can also be written in the form
\begin{equation}
    \delta(h(x))=\sum_k \frac{\delta\left(x-x_k\right)}{\left|h^{\prime}\left(x_k\right)\right|}.
\end{equation}
\end{lemma}

Now, we prove Lemma~\ref{lemma:f_a_w_analy_cont}. By substituting $\beta_{\pm i}=(1\pm z_i)/2\cdot \alpha_i$ for all $i=1,\ldots,n$ and $\beta_0=\alpha_0$, we have
\begin{equation}
    \mathfrak{a}(z)=\alpha_0\delta(z)+\alpha_n\delta(z-1) + \sum_{i=1}^{n-1}\alpha_i \cdot \left( \frac{1+z_i}{2}\delta(z-z_i)+ \frac{1+z_{-i}}{2}\delta(z-z_{-i}) \right), \quad \forall z\in(-1,1].
\end{equation}
Let us first discuss the summand $(1+z_i)/2\cdot \delta(z-z_i)$. Let $w=(1-z)/(1+z)\in[0,\infty)$.  In order to apply Lemma~\ref{lemma:delta_comp}, define $h(w)\triangleq (1-w)/(1+w)-z_i$. Solving $h(w)=0$ for the unique root, i.e., $w=w_i\triangleq (1-z_i)/(1+z_i)$, and computing the derivative, i.e., $h'(w)=-2/(1+w)^2$, we have $\delta(z-z_i)=\delta(h(w))=\delta(w-w_i)/|h'(w_i)|=(1+w_i)^2/2\cdot\delta(w-w_i)$. In the same vein, $\delta(z-z_{-i})=(1+w_{-i})^2/2\cdot \delta(w-w_{-i})$ with $w_{-i}=w_i^{-1}$ and $\delta(z-z_0)=2\delta(w-w_0)$, i.e., $\delta(z)=2\delta(w-1)$. Putting it together, we have
\begin{equation}
    \mathfrak{a}\left(\frac{1-w}{1+w}\right)=2\alpha_0 \delta(w-1) + \frac{\alpha_n}{2}\delta(w) + \sum_{i=1}^{n-1}\alpha_i \cdot \left( \frac{1+w_i}{2} \delta(w-w_i) + \frac{1+w_{-i}}{2} \delta(w-w_{-i}) \right), \ \forall w\in[0,\infty),
\end{equation}
where $w=w_n\equiv 0$ resp. $w=w_0\equiv 1$ is attained by the atom at $z_n\equiv 1$ resp. $z_0\equiv 0$. Therefore, by definition of $f_{\mathfrak{a}}(w)$ for every $w\in(0,\infty)$, they do not contribute to $\mathcal{M}\{f_{\mathfrak{a}}\}(s)$. In particular,
\begin{equation}
    \begin{aligned}
        &\mathscr{F}_{\mathfrak{a}}(s)\triangleq \mathcal{M}\{f_{\mathfrak{a}}\}(s) =\int_0^\infty f_{\mathfrak{a}}(w) \cdot w^{s-1}\mathrm{~d}w = \sum_{i=1}^{n-1} \alpha_i\cdot \left( \frac{w_i^s}{1+w_i} + \frac{w_{-i}^s}{1+w_{-i}} \right) \\
        =& \sum_{i=1}^{n-1} \alpha_i\left( \frac{1+z_i}{2} \left( \frac{1-z_i}{1+z_i} \right)^s + \frac{1-z_i}{2} \left( \frac{1+z_i}{1-z_i} \right)^s \right) = \sum_{i=1}^{n-1}\left(\beta_i w_i^s+\beta_{-i} w_{-i}^{s}\right),
    \end{aligned}
\end{equation}
where $\beta_{\pm i}=(1\pm z_i)/2\cdot \alpha_i$ and $w_{\pm i}=(1\mp z_i)/(1\pm z_i)$ for every $i=1,\ldots,n-1$. Since each $w_i^s=\mathrm{e}^{s\log w_i}$ for every $w_i\in(0,1) \cup(1, \infty)$, with $i=1,\ldots,n-1$, is entire on $\mathbb{C}$, $\mathcal{M}\{f_{\mathfrak{a}}\}(s)$ is entire too.

\subsection{Proof of Corollary~\ref{corol:F12_sym}}~\label{sec:F12_sym}
\begin{proof}
Recall that by symmetry we have~\eqref{eq:beta_alpha_1_n-1}, i.e., $\beta_{ \pm i}=\left(1 \pm z_i\right) / 2 \cdot \alpha_i$ for all $i=1, \ldots, n-1$, then $\beta_{-i}/\beta_i=(1-z_i)/(1+z_i)$. By definition, $w_i=(1-z_i)/(1+z_i)$, then $\beta_{-i}=\beta_i w_i$. Substituting it into the definition of $\mathscr{F}_{\mathfrak{a},2}(s)$, we have $\forall s \in \mathbb{C}$,
\begin{equation}
    \mathscr{F}_{\mathfrak{a}, 2}(s)=\sum_{i=1}^{n-1} \beta_{-i} w_{-i}^s=\sum_{i=1}^{n-1}\left(\beta_i w_i\right)\left(w_i^{-1}\right)^s=\sum_{i=1}^{n-1} \beta_i w_i^{1-s} = \mathscr{F}_{\mathfrak{a},1}(1-s).
\end{equation}
\end{proof}

\subsection{Proof of Lemma~\ref{lemma:F_comp_evo}}~\label{sec:F_comp_evo}

Let $\mathrm{a}$ and $\mathrm{b}$ be independent symmetric $L$-densities with $X\sim \mathrm{a}$ and $Y\sim \mathrm{b}$, then $\tilde{X}\triangleq \tanh (X / 2) \sim\mathfrak{a}$ and $\tilde{Y}\triangleq\tanh (Y/ 2) \sim\mathfrak{b}$. By property of the check-domain $\varoast$-convolution~\cite[p.195]{richardson2008modern}, $Z=X+Y \sim \mathrm{a} \varoast \mathrm{b}$. Let $\mathfrak{W}_1\triangleq (1-\tilde{X}) /(1+\tilde{X})=\mathrm{e}^{-X}$ and $\mathfrak{W}_2\triangleq (1-\tilde{Y}) /(1+\tilde{Y})=\mathrm{e}^{-Y}$, then $\mathfrak{W}_3\triangleq (1-\tilde{Z}) /(1+\tilde{Z})=\mathrm{e}^{-Z}=\mathrm{e}^{-(X+Y)}=\mathfrak{W}_1\mathfrak{W}_2$. Hence for $\mathrm{a}=\mathrm{b}$, thus, $\mathfrak{a}=\mathfrak{b}$, i.e., $\mathfrak{a} \varoast \mathfrak{b} = \mathfrak{a}^\varoast$, we have $\mathcal F_{\mathfrak a^{\varoast}}(s) = \mathcal F^2_{\mathfrak{a}}(s)$, cf.~\cite[p.199]{richardson2008modern}, which is holomorphic for $\Re s\in(0,1)$ by Lemma~\ref{lemma:var_node_holo}. 

In the sequel, we discuss $\alpha_0^{\varoast}$, $\alpha_n^{\varoast}$, $\mathscr{F}_{\mathfrak{a}^\varoast,1}(s)$ and $\mathscr{F}_{\mathfrak{a}^\varoast,2}(s)$ individually. If there is already an atom at $w=w_0\equiv 1$ with weight $\beta_0=\alpha_0$, then $\beta_0^2=\alpha_0^2$ contributes to $\alpha_0^{\varoast}$. Additionally, the mass of this atom at $w=w_0\equiv 1$ after the $\varoast$-convolution equals the total weight of all pairs $(i,j)$ with $w_i w_j=1$. Since for every $i=1,\ldots,n-1$, $w_i \cdot w_{-i}=w_i \cdot w_i^{-1}=1$, each occurring twice with the same weight $\beta_i \beta_{-i}$. Putting it together, $\alpha_0^{\varoast}=\alpha_0^2+2 \sum_{i=1}^{n-1} \beta_i \beta_{-i}$. The mass of atom at $w=w_n\equiv 0$ after the $\varoast$-convolution equals the total weight of all pairs $(i,j)$ with $w_i w_j=0$, which occurs each time when at least one operand is zero, thus $\alpha_n^{\varoast} = \alpha_n+\alpha_n-\alpha_n^2=2\alpha_n-\alpha_n^2$. 

When expanding terms and grouping from $\mathcal F^2_{\mathfrak{a}}(s)$, we have $\mathscr{F}^2_{\mathfrak{a}, 1}(s)=(\sum_i \beta_i w_i^s)^2=\sum_{i, j} \beta_i \beta_j\left(w_i w_j\right)^s$, where all products $w_iw_j\in(0,1)$ since $w_i,w_j\in(0,1)$, hence it entirely contributes to $\mathscr{F}_{\mathfrak{a}^\varoast,1}(s)$. In the same vein, $\mathscr{F}^2_{\mathfrak{a}, 2}(s)$ contributes to $\mathscr{F}_{\mathfrak{a}^\varoast,2}(s)$. Additionally, $2 \alpha_0 \mathscr{F}_{\mathfrak{a}, 1}(s)$ and $2 \alpha_0 \mathscr{F}_{\mathfrak{a}, 2}(s)$ arise from $2 \cdot\left(\alpha_0 \cdot 1^s\right) \cdot\left(\mathscr{F}_{\mathfrak{a}, 1}(s)+\mathscr{F}_{\mathfrak{a}, 2}(s)\right)$. The convolved atoms $w_0\cdot w$ remain in the same $w$-range as before the $\varoast$-convolution, hence, e.g., $2 \alpha_0 \mathscr{F}_{\mathfrak{a}, 1}(s)$ contributes to $\mathscr{F}_{\mathfrak{a}^\varoast,1}(s)$. On the other hand, the cross term $2\mathscr{F}_{\mathfrak{a}, 1}(s)\mathscr{F}_{\mathfrak{a}, 2}(s)$ is split as
\begin{equation}
    2\mathscr{F}_{\mathfrak{a}, 1}(s)\mathscr{F}_{\mathfrak{a}, 2}(s)=2\sum_{i, j} \beta_i \beta_{-j}\left(w_i / w_j\right)^s,
\end{equation}
where each \textit{ordered} pair $(i,j)$ contributes to $\mathscr{F}_{\mathfrak{a}^\varoast,1}(s)$ or $\mathscr{F}_{\mathfrak{a}^\varoast,2}(s)$ depending on whether $w_i/w_j \gtrless 1$. With monotonicity of the interior atoms for every $i,j=1, \ldots, n-1$, we have
\begin{gather}
    w_i / w_j \begin{cases}
        \in(0,1), \quad &i>j, \\
        =1, \quad &i=j, \\
        \in(1,\infty), \quad &i<j,
    \end{cases} \\
    \varphi(s)\triangleq\mathscr{F}_{\mathfrak{a}, 1}(s)\mathscr{F}_{\mathfrak{a}, 2}(s) = \left(\sum_{j<i}+\sum_{i<j}+\sum_{i=j}\right) \beta_i \beta_{-j}\left(w_i / w_j\right)^s \triangleq \mathscr{F}_{\psi, 1}(s) +\mathscr{F}_{\psi, 2}(s) + \mathscr{F}_{\varphi}, \label{eq:F_decomp}
\end{gather}
where $2\mathscr{F}_{\psi, 1}(s)\triangleq 2 \sum_{1 \leq j<i \leq n-1} \beta_i \beta_{-j}\left(w_i / w_j\right)^s$ resp. $2\mathscr{F}_{\psi, 2}(s)\triangleq 2 \sum_{1 \leq i<j \leq n-1} \beta_i \beta_{-j}\left(w_i / w_j\right)^s$ contributes to $\mathscr{F}_{\mathfrak{a}^\varoast,1}(s)$ resp. $\mathscr{F}_{\mathfrak{a}^\varoast,2}(s)$. By symmetry, we further have $\mathscr{F}_{\psi, 2}(s) = \mathscr{F}_{\psi, 1}(1-s)$, hence~\eqref{eq:F_decomp} can be written as
\begin{equation}
    \varphi(s)=\mathscr{F}_{\mathfrak{a}, 1}(s)\mathscr{F}_{\mathfrak{a}, 1}(1-s)=\mathscr{F}_{\psi, 1}(s) +\mathscr{F}_{\psi, 1}(1-s) + \mathscr{F}_{\varphi}. \label{eq:F_decomp_sym}
\end{equation}
Consequently, $\varphi(s)=\varphi(1-s)$ as well as $\mathscr{F}_{\mathfrak{a}^\varoast,2}(s)=\mathscr{F}_{\mathfrak{a}^\varoast,1}(1-s)$, too. On the other hand, the case $i=j$, i.e., $2\mathscr{F}_{\varphi}\triangleq 2\sum_{i=1}^{n-1} \beta_i \beta_{-i}$, contributes to $\alpha_0^{\varoast}$, coinciding with earlier analysis. Moreover, since for every $i,j=1, \ldots, n-1$, $w_i / w_j \in (0,\infty)$ both $\mathscr{F}_{\psi, 1}(s)$ and $\mathscr{F}_{\psi, 2}(s)$, thus $\varphi(s)$, admit an \textit{analytic continuation} to \textit{entire} functions on $\mathbb{C}$.

\subsection{Proof of Corollary~\ref{corol:F_dc_comp}}~\label{sec:F_dc_comp}

Let $\tau\triangleq c+\mathrm{i}u\in\mathbb{C}$ with any $c\in \mathbb{R}$ fixed and rewrite~\eqref{eq:F_decomp} as $\tilde{\varphi}(u) \triangleq \sum_{i,j} A_{ij}^{(c)} \cdot \mathrm{e}^{\mathrm{i}\omega_{ij}u}$, where $\omega_{ij}\triangleq \log (w_i / w_j)\in\mathbb{R}$ resp. $A_{ij}^{(c)}\triangleq \beta_i \beta_{-j}\left(w_i / w_j\right)^c\in\mathbb{R}$ is regarded as the \textit{frequency} resp. \textit{amplitude} attached to the ordered pair $(i,j)$, both are fixed once $c$ is chosen. Let $\mathcal{P}\triangleq \{ (i,j): 1 \leq i, j \leq n-1 \}$ and fix any bijection $\kappa: \mathcal{P}\to \{1,\ldots, M\}$ with $|\mathcal{P}|=M\triangleq (n-1)^2$. For any $k$-index with $k\triangleq \kappa(i,j)\in\{1,\ldots, M\}$, denote $\omega_k \triangleq \omega_{ij}$ and $A_k^{(c)} \equiv A_{i j}^{(c)}$. Define the swap involution $\bar{k}\triangleq \kappa(j,i)$ then $\omega_{\bar{k}}=-\omega_k$ and $A_{\bar{k}}^{(c)} \equiv A_{ji}^{(c)}$. Note that $\bar{k}=k$ \textit{if and only if} $i=j$ indicating diagonal indices with $\omega_{ii}=0$ and $A_{i i}^{(c)}=\beta_i \beta_{-i}$. In the sequel, we interchangeably write $A_{ij}\triangleq A_{ij}^{(c)}$ without superscript for brevity. With this convention,
\begin{equation}
    \tilde{\varphi}(u) \triangleq \sum_{k=1}^M A_k\cdot \mathrm{e}^{\mathrm{i}\omega_k u}= \sum_{k=1}^M A_k\cdot \mathrm{e}^{\omega_k (\tau-c)}\triangleq \varphi(\tau), \quad \text{and} \quad 
    \omega_k \begin{cases}
        <0, \quad &i>j, \\
        =0, \quad &i=j, \\
        >0, \quad &i<j,
    \end{cases} ~\label{eq:phi_tau}
\end{equation}
In the sequel, we show that $\mathscr{F}_{\varphi}$ corresponds to the DC-component of the \textit{Fourier transform} of $\tilde{\varphi}(u)$. In particular,
\begin{equation}
    \widehat{\varphi}(\omega)\triangleq \int_{-\infty}^{\infty} \tilde{\varphi}(u)\cdot \mathrm{e}^{-\mathrm{i}\omega u}\mathrm{~d}u = \sum_{k=1}^M A_k \cdot \int_{-\infty}^{\infty}  \mathrm{e}^{\mathrm{i}(\omega_k-\omega)u}\mathrm{~d}u = 2\pi \sum_{k=1}^M A_k \cdot \delta(\omega-\omega_k),
\end{equation}
which when evaluated at $\omega=0$ yields $\widehat{\varphi}(0)=2 \pi \sum_{k: \omega_k=0} A_k=2 \pi\sum_{i=1}^{n-1} A_{ii}= 2 \pi \mathscr{F}_{\varphi}$, resulting in~\eqref{eq:phi_zero}.

\subsection{Proof of Corollary~\ref{corol:F_dc_comp_half}}~\label{sec:F_dc_comp_half}

Every pair $(i,j)$ produces a term at frequency $\omega_{ij}$, and the swapped pair $(j,i)$ produces a terms at frequency $\omega_{ji}=-\omega_{ij}$ with coefficient $A_{ji}=\beta_j \beta_{-i}(w_j / w_i)^c \in \mathbb{R}$, By symmetry $\beta_{-i}=\beta_iw_i$, we have
\begin{equation}~\label{eq:F_dc_comp_half}
    \frac{A_{j i}}{A_{i j}}=\frac{\beta_j \beta_{-i}\left(w_j / w_i\right)^c}{\beta_i \beta_{-j}\left(w_i / w_j\right)^c}=\frac{\beta_j\left(\beta_i w_i\right)\left(w_j / w_i\right)^c}{\beta_i\left(\beta_j w_j\right)\left(w_i / w_j\right)^c}=\left(\frac{w_i}{w_j}\right)^{1-2 c},
\end{equation}
where the two coefficients are equal, i.e., $A_{j i}=A_{i j}$ \textit{if and only if} $c=1/2$. With such choice, the pair $(i,j)$ together with its swapped counterpart contributes to $A_{i j} \cdot \mathrm{e}^{\mathrm{i} \omega_{ij} u}+A_{j i} \cdot \mathrm{e}^{\mathrm{i} \omega_{ji} u}=2 A_{i j} \cos (\omega_{ij} u)\in \mathbb{R}$. By enumerating two off-diagonal index sets with $1 \leq j<i \leq n-1$ or $1 \leq i<j \leq n-1$, each of size $(n-1)(n-2)/2$, and using the $k$-index convention, we obtain~\eqref{eq:phi_omega_half}. Additionally, it makes no difference to write the second terms as $2\sum_{k:\omega_k>0} A_{k}^{(c)} \cdot \cos (\omega_{k} u)$ since the cosine function is \textit{even}.

\subsection{Proof of Corollary~\ref{corol:F_mellin_comp}}~\label{sec:F_mellin_comp}

To retrieve $\mathscr{F}_{\psi,i}(s)$ with $i\in\{1,2\}$ from the spectrum of $\tilde{\varphi}(u)$ in Corollary~\ref{corol:F_mellin_comp}, we first establish the following lemma.

\begin{lemma}~\label{lemma:plemelj_F}
Define $\psi(s)\triangleq \varphi(s)-\mathscr{F}_{\varphi}$, $\forall s\in\mathbb{C}$, and the \textit{Cauchy-type integral}~\cite[(7.2.1)]{ablowitz2003complex} $\Phi(s)\triangleq \frac{1}{2\pi \mathrm{i}}\int_L \frac{\psi(\tau)}{\tau-s}\mathrm{~d}\tau$ for $s\notin L$ and $s\in\mathbb{C}$, where $L\triangleq \{\tau = c+\mathrm{i}u, u\in\mathbb{R}\}$ with $c\in(0,1)$ and the positive orientation upward, then
\begin{enumerate}
    \item $\forall s\notin L$: $\mathscr{F}_{\psi, 1}(s)=-\Phi^{-}(s)$, $\forall \Re s>c$; $\mathscr{F}_{\psi, 2}(s) =\Phi^+(s)$, $\forall \Re s<c$.
    \item $\forall t\in L$: the \textit{Plemelj formulae}~\cite[Lemma 7.2.1]{ablowitz2003complex} give boundary values of $\Phi$, i.e., $\Phi^{\pm}(t) = \pm\frac{1}{2}\psi(t)+\frac{1}{2\pi \mathrm{i}} \dashint_L \frac{\psi(\tau)}{\tau-t}\mathrm{~d}\tau$.
    \item $\mathscr{F}_{\psi, 1}(1-s)= \mathscr{F}_{\psi, 2}(s)$ required by the symmetric $\mathfrak{a}$ holds true if $c=1/2$.
\end{enumerate}
Here, $\dashint$ denotes the \textit{principal value integral}.
\end{lemma}

\begin{proof}
By definition, $\psi(s)=\mathscr{F}_{\psi, 1}(s) +\mathscr{F}_{\psi, 2}(s)$, $\forall s\in\mathbb{C}$, which is \textit{entire} on $\mathbb{C}$. Let $\tau\triangleq c+\mathrm{i}u$, then we have $h(w)=\frac{1}{2\pi \mathrm{i}}\int_{c-\mathrm{i}\infty}^{c+\mathrm{i}\infty} \varphi(\tau)w^{-\tau}\mathrm{~d}\tau=\frac{1}{2\pi \mathrm{i}}\int_{c-\mathrm{i}\infty}^{c+\mathrm{i}\infty} \psi(\tau)w^{-\tau}\mathrm{~d}\tau + \mathscr{F}_{\varphi}\cdot \delta(w-1)$, $\forall c\in(0,1)$, which can be further decomposed as
\begin{equation}
    h(w)\triangleq \begin{cases}
        h_-(w), \quad & w\in(0,1), \\
        \mathscr{F}_{\varphi}, & w=1, \\
        h_+(w), & w\in(1,\infty).
    \end{cases}
\end{equation}
Therefore, $\mathscr{F}_{\psi, 1}(s)=\int_0^1 h_-(w)w^{s-1}\mathrm{~d}w=\frac{1}{2\pi \mathrm{i}}\int_{c-\mathrm{i}\infty}^{c+\mathrm{i}\infty} \psi(\tau) \left( \int_0^1 w^{s-\tau-1} \mathrm{~d}w\right)\mathrm{d}\tau=-\frac{1}{2\pi \mathrm{i}}\int_{c-\mathrm{i}\infty}^{c+\mathrm{i}\infty} \frac{\psi(\tau)}{\tau-s}\mathrm{~d}\tau$, where the inner $w$-integral converges if and only if $\Re s>\Re \tau=c$; while $\mathscr{F}_{\psi, 2}(s)=\int_1^\infty h_+(w)w^{s-1}\mathrm{~d}w=\frac{1}{2\pi \mathrm{i}}\int_{c-\mathrm{i}\infty}^{c+\mathrm{i}\infty} \psi(\tau) \left( \int_1^\infty w^{s-\tau-1} \mathrm{~d}w\right)\mathrm{d}\tau=\frac{1}{2\pi \mathrm{i}}\int_{c-\mathrm{i}\infty}^{c+\mathrm{i}\infty} \frac{\psi(\tau)}{\tau-s}\mathrm{~d}\tau$, where the convergence of inner $w$-integral holds true if and only if $\Re s<\Re \tau=c$. Let $L$ be the vertical line, i.e., $L\triangleq \{\tau\triangleq c+\mathrm{i}u, u\in\mathbb{R}\}$ with $c\in(0,1)$ fixed and the \textit{positive orientation upward}. Define the \textit{Cauchy-type integral}. i.e., 
\begin{equation}~\label{eq:phi_s}
    \Phi(s)\triangleq \frac{1}{2\pi \mathrm{i}}\int_L \frac{\psi(\tau)}{\tau-s}\mathrm{~d}\tau,
\end{equation}
which is \textit{sectionally analytic} off $L$, i.e., $s\notin L$. Therefore, we have
\begin{gather}
   \mathscr{F}_{\psi, 1}(s) =-\Phi^-(s), \quad \forall \Re s>c, \label{eq:F_psi_1} \\
   \mathscr{F}_{\psi, 2}(s) =\Phi^+(s),  \quad \ \ \ \forall \Re s<c, \label{eq:F_psi_2}
\end{gather}
By applying the \textit{Plemelj formulae}, we obtain \textit{boundary values} of the \textit{Cauchy-type integral} for every $t\in L$, i.e.,
\begin{equation}~\label{eq:plemelj}
    \Phi^{\pm}(t) = \pm\frac{1}{2}\psi(t)+\frac{1}{2\pi \mathrm{i}} \dashint_L \frac{\psi(\tau)}{\tau-t}\mathrm{~d}\tau, \qquad \Phi^+(t)-\Phi^-(t) = \psi(t),
\end{equation}
where $\dashint$ denotes the \textit{principal value integral} defined by $\dashint_L \frac{\psi(\tau) d \tau}{\tau-t}\triangleq \lim _{\varepsilon \rightarrow 0} \int_{L-L_{\varepsilon}} \frac{\psi(\tau) d \tau}{\tau-t}$ where $L_{\varepsilon}$ is the part of $L$ that has length $2 \varepsilon$ and is centered around $t$. In the sequel, the specific choice of $c\in(0,1)$ is determined in order to fulfill the symmetry imposed by $\mathscr{F}_{\psi, 2}(s) = \mathscr{F}_{\psi, 1}(1-s)$ in~\eqref{eq:F_decomp_sym}. Recall that $\psi(s) \triangleq \varphi(s)-\mathscr{F}_{\varphi}$ and $\varphi(s) = \varphi(1-s)$, then $\psi(s) = \psi(1-s)$. Let $\sigma \triangleq 1-\tau$, then the image of $L$ under $\tau \mapsto 1-\tau$ yields $L'\triangleq \{\sigma =(1-c)-\mathrm{i}u, u\in\mathbb{R}\}$ with $(1-c)\in(0,1)$ fixed and the \textit{positive orientation downward}, thus
\begin{equation}
    \Phi(1-s)=\frac{1}{2\pi \mathrm{i}} \int_L \frac{\psi(\tau)}{\tau-(1-s)}\mathrm{~d} \tau = -\frac{1}{2\pi \mathrm{i}} \int_{L'}\frac{\psi(\sigma)}{s-\sigma} \mathrm{~d} \sigma = \frac{1}{2\pi \mathrm{i}} \int_{L'}\frac{\psi(\sigma)}{\sigma-s} \mathrm{~d} \sigma,
\end{equation}
for which $L'=-L$ if $c=1/2$, which implies that $\Phi(1-s)=-\Phi(s)$ for every $s\notin L$. Therefore, $\mathscr{F}_{\psi, 1}(1-s)=-\Phi^{-}(1-s)=-\left(-\Phi^{+}(s)\right)=\mathscr{F}_{\psi, 2}(s)$.
\end{proof}

Now, we prove Corollary~\ref{corol:F_mellin_comp}. Recall in~\eqref{eq:phi_tau} that $\varphi(\tau)=\sum_k A_k\cdot \mathrm{e}^{\omega_k (\tau-c)}$. Substituting it into~\eqref{eq:phi_s}, we obtain
\begin{equation}
    \Phi(s) =  \frac{1}{2\pi \mathrm{i}} \sum_k A_k  \mathrm{e}^{-\omega_k c} \cdot \int_L \frac{\mathrm{e}^{\omega_k \tau}}{\tau-s} \mathrm{~d}\tau,
\end{equation}
where $L\triangleq \{\tau\triangleq c+\mathrm{i}u, u\in\mathbb{R}\}$ with the \textit{positive orientation upward}. If $\omega_k<0$ resp. $\omega_k>0$, then $\mathrm{e}^{\omega_k \tau}$ decays exponentially as $\Re \tau \to +\infty$ resp. $\Re \tau \to -\infty$. Consequently, if $\omega_k<0$ resp. $\omega_k>0$, closing the contour \textit{at right} resp. \textit{at left}, the semicircle contribution vanishes. Let us first discuss the first case, i.e., $\omega_k <0$, under contour deformation at right. If $\Re s> c$, since $\mathrm{e}^{\omega_k \tau}$ is \textit{entire} in $\tau\in\mathbb{C}$, there exists only one simple pole at $\tau = s$ to the right of $L$, by \textit{Cauchy's residue theorem}, we have $\int_L \frac{\mathrm{e}^{\omega_k \tau}}{\tau-s} d \tau=-2\pi \mathrm{i}\cdot \operatorname{Res}_{\tau=s} \frac{\mathrm{e}^{\omega_k \tau}}{\tau-s}=-2\pi \mathrm{i}\cdot \mathrm{e}^{\omega_k s}$; while if $\Re s< c$, no pole is enclosed, thus, $\int_L \frac{\mathrm{e}^{\omega_k \tau}}{\tau-s} d \tau=0$. In the same vein, for the case $\omega_k>0$ under contour deformation at left, if $\Re s < c$, $\int_L \frac{\mathrm{e}^{\omega_k \tau}}{\tau-s} d \tau=2\pi \mathrm{i}\cdot \mathrm{e}^{\omega_k s}$; while if $\Re s>c$, $\int_L \frac{\mathrm{e}^{\omega_k \tau}}{\tau-s} d \tau=0$. Putting it together,
\begin{equation}~\label{eq:phi_s_F}
    \Phi(s) = \begin{cases}
        - \sum_{k:\omega_k <0} A_k \cdot \mathrm{e}^{\omega_k(s-c)}, & \forall\Re s>c, \\
        \sum_{k:\omega_k >0} A_k \cdot \mathrm{e}^{\omega_k(s-c)}, & \forall\Re s<c. 
    \end{cases}
\end{equation}
By substituting~\eqref{eq:phi_s_F} into~\eqref{eq:F_psi_1} and~\eqref{eq:F_psi_2}, respectively, we obtain
\begin{gather}
   \mathscr{F}_{\psi, 1}(s) =\sum_{k:\omega_k <0} A_k \cdot \mathrm{e}^{\omega_k(s-c)}, \quad \forall \Re s>c, \label{eq:F_psi_1_A_omega} \\
   \mathscr{F}_{\psi, 2}(s) =\sum_{k:\omega_k >0} A_k \cdot \mathrm{e}^{\omega_k(s-c)}, \quad \forall \Re s<c. \label{eq:F_psi_2_A_omega}
\end{gather}
Moreover, as shown in Lemma~\ref{lemma:F_comp_evo}, both $\mathscr{F}_{\psi, 1}(s)$ and $\mathscr{F}_{\psi, 2}(s)$ admit an \textit{analytic continuation} to \textit{entire} functions on $\mathbb{C}$, hence~\eqref{eq:F_psi_1_A_omega} resp.~\eqref{eq:F_psi_2_A_omega} is analytically continued to $\Re s\leq c$ resp. $\Re s\geq c$, resulting in~\eqref{eq:F_psi_1_F} resp.~\eqref{eq:F_psi_2_F}. Here, $c=1/2$ is chosen to guarantee $\mathscr{F}_{\psi, 1}(1-s)=\mathscr{F}_{\psi, 2}(s)$ required by the symmetry of $\mathfrak{a}$, cf. Lemma~\ref{lemma:plemelj_F}.3. To see this, we rewrite~\eqref{eq:F_psi_1_F} as $\mathscr{F}_{\psi, 1}(s) =\sum_{k:\omega_{\bar{k}} >0} A_k^{(c)} \cdot \mathrm{e}^{-\omega_{\bar{k}}(s-c)}$. Hence,
\begin{equation}
    \mathscr{F}_{\psi, 1}(1-s) =\sum_{k:\omega_{\bar{k}} >0} A_k^{(c)} \cdot \mathrm{e}^{-\omega_{\bar{k}}(1-s-c)}=\sum_{k:\omega_{\bar{k}} >0} A_k^{(c)} \cdot \mathrm{e}^{\omega_{\bar{k}}(s-c)} \mathrm{e}^{\omega_{\bar{k}}(2 c-1)}.
\end{equation}
With $c=1/2$, $\mathrm{e}^{\omega_{\bar{k}}(2 c-1)}=1$ and $A_{k}^{(c)}=A_{\bar{k}}^{(c)}$, cf.~\eqref{eq:F_dc_comp_half}, then $\mathscr{F}_{\psi, 1}(1-s) = \sum_{k:\omega_{\bar{k}} >0} A_{\bar{k}}^{(c)} \cdot \mathrm{e}^{\omega_{\bar{k}}(s-c)}=\mathscr{F}_{\psi, 2}(s)$.

\subsection{Proof of Proposition~\ref{prop:G2F_convert}}~\label{sec:G2F_convert}

To prove Proposition~\ref{prop:G2F_convert}, we need the following Lemma~\ref{lemma:k_s_analy_cont} and its Corollary~\ref{corol:K_s_abs_int}.

\begin{lemma}~\label{lemma:k_s_analy_cont}
Fix $\Re s \in(0,1)$. Define
\begin{equation}
    g_s(z)\triangleq \begin{cases}
        \left(\frac{1-z}{2}\left(\frac{1+z}{1-z}\right)^s+\frac{1+z}{2}\left(\frac{1-z}{1+z}\right)^s\right) \cdot \frac{1}{z}, & 0< z < 1, \\
        0, & z\geq 1,
    \end{cases}
\end{equation}
of which the Mellin transform
\begin{equation}
    \mathcal{K}_s(\nu)\triangleq \mathcal{M}\{g_s\}(\nu)=\Gamma(\nu-1)(A(1-s,1-\nu)+A(s,1-\nu)),
\end{equation}
is initially defined and holomorphic for $\Re \nu>1$, where $A(\alpha,\beta)\triangleq 2^{-\alpha} \cdot \Gamma(\alpha+1)_2 \tilde{F}_1(\alpha-1, \alpha+1 ; 1-\beta+\alpha ; 1 / 2)$. $\mathcal{K}_s(\nu)$ admits an \textit{analytic continuation} to a meromorphic function on $\mathbb{C}$ whose only singularities are simple poles at $\nu=1,0,-1,-2,\ldots$. 
\end{lemma}

\begin{proof}
We first justify the strip of analyticity (SOA) of the Mellin transform $\mathcal{M}\{g_s\}(\nu)$, i.e., $\mathcal{K}_s(\nu)=\mathcal{K}(s,\nu)\triangleq \int_0^{+\infty} g_s(z)\cdot z^{\nu-1} \mathrm{d} z$. Fix $\Re s\in(0,1)$. As $z\to 0^+$, $\left(\frac{1+z}{1-z}\right)^s=1+2 s z+O\left(z^2\right)$ and $\left(\frac{1-z}{1+z}\right)^s=1-2 s z+O\left(z^2\right)$, hence $g_s(z)=(1+O(z^2))/z\sim z^{-1}$ for small $z$. Consequently, the integrand behaves like $g_s(z) z^{\nu-1} \sim z^{\nu-2}$. The integral $\int_0^\delta z^{\nu-2} d z$ for small $\delta>0$ converges absolutely~\cite[A.2]{fikioris2022mellin} if and only if $\Re\nu-2>-1$, i.e., $\Re\nu>1$. On the other hand, $g_s(z) \sim 2^{s-1}(1-z)^{1-s}+2^{-s}(1-z)^s \quad$ as $z \rightarrow 1^{-}$. Since $\Re s \in(0,1)$, i.e., both exponents $\Re(1-s)$ and $\Re s$ are positive, by definition of the \textit{Hausdorff dimension function}~\cite[Definition 1.2]{sudland2004mellin}, we have $\lim_{z\to 1^-}g_s(z)=2^{s-1}\cdot 0^{1-s}+2^{-s}\cdot 0^s=0$. By definition, $g_s(z)=0$ for $z \geq 1$, so there is no contribution from infinity. Therefore, $\mathcal{K}_s(\nu)$ converges absolutely and possesses its SOA~\cite[Sec. 3.1.1]{paris2001asymptotics} as $\Re\nu>1$. Moreover,
\begin{equation}
    \begin{aligned}
        \mathcal{K}(s, \nu) &\triangleq \int_0^1 g_s(z) z^{\nu-1} \mathrm{~d} z = \int_0^1 \left(\frac{1-z}{2}\left(\frac{1+z}{1-z}\right)^s+\frac{1+z}{2}\left(\frac{1-z}{1+z}\right)^s\right) z^{\nu-2} \mathrm{~d} z \\
        &= \frac{1}{2} \int_0^1 z^{\nu-2}(1-z)^{1-s}(1+z)^{s} \mathrm{~d} z + \frac{1}{2} \int_0^1 z^{\nu-2}(1-z)^s(1+z)^{1-s} \mathrm{~d} z \\
        &\triangleq \frac{1}{2}(J(1-s,1-\nu)+J(s,1-\nu)),
    \end{aligned}
\end{equation}

where $J(\alpha,\beta)\triangleq \int_0^1 z^{-(\beta+1)}(1-z)^\alpha(1+z)^{1-\alpha} \mathrm{~d} z$ with $\alpha\in\{s,1-s\}$ and $\beta\in\{\nu,1-\nu\}$. Recall that the \textit{Gauss hypergeometric function}~\cite[(15.3.1)]{abramowitz1965handbook} possesses an integral representation, i.e., $\int_0^1x^{b-1}(1-x)^{c-b-1}(1-zx)^{-a}\mathrm{d}x = \mathrm{B}(b,c-b){}_2F_1(a,b;c;z)$ which converges absolutely for $\Re c > \Re b >0$. By identifying that $a=\alpha-1$, $b=-\beta$, $c = 1-\beta+\alpha$, and $z=-1$, hence $J(\alpha,\beta)=\mathrm{B}(-\beta,\alpha+1){}_2F_1(\alpha-1, -\beta; 1-\beta+\alpha;-1)$. By using the \textit{Pfaff transformation}~\cite[(15.3.4)]{abramowitz1965handbook}, ${}_2F_1(\alpha-1, -\beta; 1-\beta+\alpha;-1)=2^{1-\alpha}{}_2F_1(\alpha-1,\alpha+1;1-\beta+\alpha;1/2)$, thus,
\begin{equation}
    J(\alpha, \beta)=2^{1-\alpha} \mathrm{B}(-\beta, \alpha+1){ }_2 F_1\left(\alpha-1, \alpha+1 ; 1-\beta+\alpha ; 1/2\right),
\end{equation}
which is valid initially where the integral representation converges, i.e., $\Re \beta<0$ and $\Re \alpha>-1$ (which holds true for $\alpha\in\{s,1-s\}$ with $\Re s \in(0,1)$). In the sequel, we show the analytic continuation of $J(\alpha, \beta)$ in $\beta$ (or $\nu$). Putting $\beta = 1-\nu$, then
\begin{equation}
    J(\alpha, 1-\nu)=2^{1-\alpha}\Gamma(\nu-1) \Gamma(\alpha+1){ }_2 \tilde{F}_1\left(\alpha-1, \alpha+1 ; \alpha+\nu ; 1/2\right),
\end{equation}
where ${ }_2 \tilde{F}_1(a, b ; c ; z)\triangleq { }_2 F_1(a, b ; c ; z) / \Gamma(c)$ denotes the \textit{regularized Gauss hypergeometric function}~\cite[(9.02)]{olver1997asymptotics}, for which the limits exists as $c \to-m$, where $m$ is a non-negative integer~\cite[(15.1.2)]{abramowitz1965handbook}. Note that $\Gamma(\nu-1)$ is meromorphic with simple poles at $\nu=1,0,-1,-2, \ldots$, while $\Gamma(\alpha+1)$ with $\Re \alpha>-1$ contributes no simple poles. Therefore, $J(\alpha, 1-\nu)$ is a meromorphic function on $\mathbb{C}$ of $\nu$ with simple poles from $\Gamma(\nu-1)$ at $\nu=1,0,-1,-2,\ldots$. Therefore, $\mathcal{K}_s(\nu)$ extends analytically from initial $\Re\nu>1$ to $\nu\in \mathbb{C}$ with these simple poles. In particular,
\begin{equation}
    \mathcal{K}(s,\nu)=\Gamma(\nu-1)(A(1-s,1-\nu)+A(s,1-\nu)),
\end{equation}
where $A(\alpha,\beta)\triangleq 2^{-\alpha} \cdot \Gamma(\alpha+1)_2 \tilde{F}_1(\alpha-1, \alpha+1 ; 1-\beta+\alpha ; 1 / 2)$ with $\alpha\in\{s,1-s\}$ and $\beta\in\{\nu,1-\nu\}$. Putting $\beta=\nu$ gives $\mathcal{K}_s(1-\nu)= \Gamma(-\nu) (A(1-s,\nu)+A(s,\nu))$ as well an analytic continuation with simple poles at the non-negative integers $\nu=0,1,2,\ldots$.
\end{proof}

\begin{corollary}~\label{corol:K_s_abs_int}
Fix $s \in \mathbb{C}$ with $\Re s \in(0,1)$ and $k\in\mathbb{R}\backslash\{\mathbb{Z}_{\geq 0}\}$, then $t\mapsto\mathcal{K}_s(1-k-\mathrm{i}t)\in L(-\infty,\infty)$.
\end{corollary}
\begin{proof}
By Lemma~\ref{lemma:k_s_analy_cont} and writing $\nu=k+\mathrm{i}t$ with fixed $k\in\mathbb{R}\backslash\{\mathbb{Z}_{\geq 0}\}$, we have
\begin{align}
    &\mathcal{K}_s(1-k-\mathrm{i} t)\triangleq \mathcal{K}(s, 1-\nu) \nonumber \\
    =& \Gamma(-\nu)\left(2^{s-1}\cdot \Gamma(2-s){ }_2 \tilde{F}_1(-s, 2-s ; 2-\nu-s ; 1 / 2) + 2^{-s}\cdot \Gamma(s+1){ }_2 \tilde{F}_1(s-1, s+1 ; 1-\nu+s ; 1 / 2) \right)  \nonumber \\
    =& 2^{s-1} \cdot  \frac{\Gamma(2-s)\Gamma(-\nu)}{\Gamma(2-\nu-s)} { }_2 F_1(-s, 2-s ; 2-\nu-s ; 1 / 2)+ 2^{-s}\cdot  \frac{\Gamma(s+1)\Gamma(-\nu)}{\Gamma(1-\nu+s)}{ }_2 F_1(s-1, s+1 ; 1-\nu+s ; 1 / 2)  \nonumber \\
    \triangleq & 2^{s-1} \cdot \Gamma(2-s) \cdot G(1-s,\nu) + 2^{-s}\cdot \Gamma(s+1) \cdot G(s,\nu),
\end{align}

where $G(\alpha,\nu)\triangleq \Gamma(-\nu)/\Gamma(1-\nu+\alpha)\cdot { }_2 F_1(\alpha-1, \alpha+1 ; 1-\nu+\alpha ; 1 / 2)$ with $\alpha\in\{s,1-s\}$ fixed. Thus, the $\nu$-dependence arises from the third parameter of ${}_2 F_1(\alpha-1, \alpha+1 ; 1-\nu+\alpha ; 1 / 2)$ and the Gamma-ratio. In the sequel, we show their asymptotic behaviors separately. For the Gauss hypergeometric part, note that $c=1-\nu+\alpha=(1-k+\alpha)-\mathrm{i}t$ with other parameters $a=\alpha-1$, $b=\alpha+1$ and $z=1/2$ fixed. By~\cite[Sec. 15.12 (ii)(c)]{NIST:DLMF}, let $a,b,z\in\mathbb{R}$ or $\mathbb{C}$ and fixed, if $\Re z=1/2$ and $|\arg c|\leq \pi-\delta$ with $\delta$ an arbitrary small positive constant, then for every fixed $m\in\mathbb{Z}_{\geq 0}$, 
\begin{equation}
    {}_2F_1(a, b ; c ; z)=\sum_{s=0}^{m-1} \frac{(a)_s(b)_s}{(c)_s s!} z^s+O\left(c^{-m}\right),
\end{equation}
as $|c|\to\infty$. In our case, let $u\triangleq 1-k+\alpha$ as a constant. As $|t|\to \infty$, then $\arg c =\arctan (-t/u)$ tends to $\pm \pi/2$ depending on the sign of $t$ and $u$. Therefore, for a sufficiently large $|t|$, we write $|\arg c| \leq \pi / 2 - \varepsilon$ with $\varepsilon$ an arbitrary small positive constant, which clearly fulfill the condition $|\arg c|\leq \pi-\delta$. By taking $m=1$, we have ${}_2F_1(a, b ; c ; z)=1+O(c^{-1})$ as $|c|\to\infty$. In our case, we have ${}_2 F_1(\alpha-1, \alpha+1 ; 1-\nu+\alpha ; 1 / 2)=1+O((u-\mathrm{i}t)^{-1})$ as $|t|\to\infty$, which implies $|{}_2 F_1(\alpha-1, \alpha+1 ; 1-\nu+\alpha ; 1 / 2)|=1+O(|t|^{-1})$.

On the other hand, $\Gamma(z+a)/\Gamma(z+b)\sim z^{a-b}$ for $a,b\in\mathbb{R}$ or $\mathbb{C}$ as $z\to\infty$ in the sector $|\arg z|\leq \pi-\delta$, cf.~\cite[(5.11.12)]{NIST:DLMF}~\cite[(6.1.47)]{abramowitz1965handbook}. In our case, let $z\triangleq -\nu=-k-\mathrm{i}t$ and identify $a=0$ and $b=1+\alpha$, then $\Gamma(z)/\Gamma(z+(1+\alpha))=z^{-(1+\alpha)}(1+O(z^{-1}))$. Since $|t|\to\infty$, we further have $|\Gamma(z)/\Gamma(z+(1+\alpha))|=O(|t|^{-(1+\Re \alpha)})$, which implies a polynomial decay of the Gamma-ratio on the vertical line $\Re \nu=k$. Putting it together, we have $|G(1-s, \nu)|=O\left(|t|^{-(1+\Re(1-s))}\right)$ and $|G(s, \nu)|=O\left(|t|^{-(1+\Re s)}\right)$. Therefore, as $|t|\to \infty$,
\begin{equation}
    |\mathcal{K}_s(1-k-\mathrm{i} t)|=O\left(|t|^{-(1+\Re s)}\right)+O\left(|t|^{-(2-\Re s)}\right)=O\left(|t|^{-(1+\min (\Re s, 1-\Re s))}\right).
\end{equation}
Since $\Re s \in(0,1)$ implying $1+\min (\Re s, 1-\Re s)>1$, hence $t\mapsto\mathcal{K}_s(1-k-\mathrm{i} t) \in L(-\infty, \infty)$.
\end{proof}

Now, we prove Proposition~\ref{prop:G2F_convert}. Recall that $\mathcal{F}_{|\mathfrak{a}|}(s)=\int_0^1|\mathfrak{a}|(z)\left(\frac{1-z}{2}\left(\frac{1+z}{1-z}\right)^s+\frac{1+z}{2}\left(\frac{1-z}{1+z}\right)^s\right) \mathrm{d} z$, which is holomorphic for $\Re s\in(0,1)$ by Lemma~\ref{lemma:var_node_holo}. By Lemma~\ref{lemma:f_a_analy_cont} and~\ref{lemma:k_s_analy_cont}, we have $\mathcal{F}_{|\mathfrak{a}|}(s)=\alpha_0\cdot 1 \cdot \mathbbm{1}_{\{\alpha_0 > 0\}} + \alpha_{n}\cdot 0 \cdot \mathbbm{1}_{\{\alpha_n > 0\}} +\int_0^{+\infty}f_{|\mathfrak{a}|}(z)\cdot g_s(z)\mathrm{~d} z$. By the \textit{Parseval formula for the Mellin transform}~\cite[(3.1.11)]{paris2001asymptotics}, we obtain
\begin{equation}~\label{eq:Mellin_Barnes_G2F}
    \int_0^{+\infty}f_{|\mathfrak{a}|}(z)\cdot g_s(z)\mathrm{~d} z=\frac{1}{2 \pi \mathrm{i}} \int_{c-\mathrm{i} \infty}^{c+\mathrm{i} \infty} \mathscr{G}_{|\mathfrak{a}|}(\nu) \mathcal{K}_s(1-\nu)\mathrm{d} \nu, \quad \forall c\in(-1,0),
\end{equation}
where by Corollary~\ref{corol:K_s_abs_int} $\mathcal{K}_s(1-\nu)\in L(-\infty, \infty)$ is meromorphic on $\mathbb{C}$ with simple poles $\nu=t \in \mathbb{Z}_{\geq 0}$. More precisely,
\begin{equation}
    \mathcal{K}(s, 1-\nu) = \Gamma(-\nu) (A(1-s,\nu)+A(s,\nu)),
\end{equation}
where $A(\alpha,\beta)$ is defined in Lemma~\ref{lemma:k_s_analy_cont}. Additionally, $\mathscr{G}_{|\mathfrak{a}|}(\nu)$ is entire by Lemma~\ref{lemma:f_a_analy_cont}. The integrand of~\eqref{eq:Mellin_Barnes_G2F} is $I(s,\nu)\triangleq \mathscr{G}_{|\mathfrak{a}|}(\nu) \mathcal{K}(s, 1-\nu)$. For fixed $s \in \mathbb{C}$ with $\Re s \in(0,1)$, $I(s, \nu)$ is holomorphic except at \textit{simple poles} $\nu=t \in \mathbb{Z}_{\geq 0}$ solely from $\Gamma(-\nu)$ with residue $\operatorname{Res}_{\nu=t} \Gamma(-\nu)=(-1)^{t+1} / t!$. In what follows, we show that the integral of $I(s,\nu)$ over the \textit{right} semicircular part tends to $0$ as $R\to\infty$. Let the large semicircle, denoted as $C_N$, be $\nu\triangleq R \cdot\mathrm{e}^{\mathrm{i}\theta}\triangleq \sigma+\mathrm{i}t$ of radius $R=N+1/2$ and $\theta\in[\pi/2,-\pi/2]$. We will show that $\lim_{N\to \infty} \int_{C_N} I(s,\nu)\mathrm{~d}\nu\to 0$ via asymptotics of $\mathscr{G}_{|\mathfrak{a}|}(\nu)$ and $\mathcal{K}_s(1-\nu)$, respectively. By Lemma~\ref{lemma:f_a_analy_cont}, $\mathscr{G}_{|\mathfrak{a}|}(\nu)=\sum_{i=1}^{n-1} \alpha_i z_i^\nu$ with $|z_i|<1$ and $\alpha_i\in[0,1]$, then $|\mathscr{G}_{|\mathfrak{a}|}(\nu)|\leq \sum_{i=1}^{n-1} \alpha_i\cdot \mathrm{e}^{\log z_i\cdot  \Re \nu}\leq C_1 \mathrm{e}^{-c_1 R}$ for some positive constants $C_1$ and $c_1$. By Corollary~\ref{corol:K_s_abs_int}, $|\mathcal{K}_s(1-\nu)|=C_2\cdot R^{-c_2}$ with some positive constants $C_2$ and $c_2>1$. Therefore, $|\int_C I(s,\nu)\mathrm{~d}\nu| \leq \pi R\cdot  \sup_{\nu\in C} |I(s,\nu)|\leq C_3 R^{1-c_2}\mathrm{e}^{-c_1 R}$ with some positive constant $C_3$, where since $c_2>1$ the first term $R^{1-c_2}$ decays polynomially, while $\mathrm{e}^{-c_1 R}$ decays exponentially. Therefore, $\lim_{N\to \infty} \int_{C_N} I(s,\nu)\mathrm{~d}\nu\to 0$ and by the residue theorem~\cite[Theorem 3.2.1]{stein2010complex}, we have $\int_{c-i \infty}^{c+i \infty} I(s, \nu) d \nu=-2\pi \mathrm{i}\cdot\sum_{t=0}^{\infty} \operatorname{Res}_{\nu=t} I(s, \nu)$. Consequently, by closing the contour \textit{at right}, the \textit{Mellin–Barnes integral} equals the sum of these residues, i.e. $\int_0^{+\infty}f_{|\mathfrak{a}|}(z)\cdot g_s(z)\mathrm{~d} z=-\sum_{t=0}^{\infty} \operatorname{Res}_{\nu=t} I(s,\nu)=-\sum_{t=0}^{\infty} \frac{(-1)^{t+1}}{t!}\left(A(1-s, t)+A(s, t)\right) \mathscr{G}_{|\mathfrak{a}|}(t)$.

\subsection{Proof of Proposition~\ref{prop:F2G_convert}}~\label{sec:F2G_convert}

To prove Proposition~\ref{prop:F2G_convert}, we need the following Lemma~\ref{lemma:h_v_analy_cont}.

\begin{lemma}~\label{lemma:h_v_analy_cont}
Fix $\Re \nu >0$. Define
\begin{equation}
    \phi_\nu(w) \triangleq \frac{1}{w}\left|\frac{1-w}{1+w}\right|^\nu, \ \quad \forall w\in(0,\infty),
\end{equation}
of which the Mellin transform
\begin{equation}
    \mathcal{H}_{\nu}(s) \triangleq \mathcal{M}\{\phi_\nu\}(s)=2^{-\nu}\Gamma(\nu+1) (B(2-s,\nu)+B(s,\nu)),
\end{equation}
has no non-empty SOA, where $B(\alpha,\nu)\triangleq \Gamma(1-\alpha) {}_2\tilde{F}_1(\nu,\nu-1;2+\nu-\alpha;1/2)$. Nevertheless, $\mathcal{H}_{\nu}(s)$ admits an \textit{analytic continuation} to a meromorphic function on $\mathbb{C}$ whose singularities are \textit{simple poles} at $s=\mathbb{Z}\backslash \{1\}$.
\end{lemma}

\begin{proof}
\begin{align}
    \mathcal{H}(s, \nu) &\triangleq \int_0^\infty \phi_\nu(w) w^{s-1}\mathrm{~d}w= \int_0^{+\infty} \left|\frac{1-w}{1+w}\right|^\nu  w^{s-2}\mathrm{d} w \\
    &= \int_0^{1} \left(\frac{1-w}{1+w}\right)^\nu w^{s-2} \mathrm{d} w + \int_1^{+\infty} \left(\frac{w-1}{1+w}\right)^\nu w^{s-2} \mathrm{d} w \\
    &= \int_0^{1} \left(\frac{1-w}{1+w}\right)^\nu w^{s-2} \mathrm{d} w + \int_0^1 \left(\frac{1-w}{1+w}\right)^\nu w^{-s} \mathrm{d} w \\
    &\triangleq M(2-s,\nu)+M(s,\nu),
\end{align}
with $M(\alpha,\nu)\triangleq \int_0^1 w^{-\alpha}(1-w)^\nu (1+w)^{-\nu}\mathrm{d} w$ with $\alpha\in\{s,2-s,1-s,1+s\}$. By the integral representation of the Gauss hypergeometric function ($a=\nu$, $b=1-\alpha$, $c=2+\nu-\alpha$ and $z=-1$), we have $M(\alpha,\nu)=\mathrm{B}(1-\alpha, \nu+1){}_2 F_1(\nu, 1-\alpha ; 2+\nu-\alpha;-1)$. By using the \textit{Pfaff transformation}~\cite[(15.3.4)]{abramowitz1965handbook}, ${}_2F_1(\nu, 1-\alpha ; 2+\nu-\alpha;-1)=2^{-\nu}{}_2F_1(\nu,\nu+1;2+\nu-\alpha;1/2)$, thus,
\begin{equation}
    M(\alpha,\nu) = 2^{-\nu}  \mathrm{B}(1-\alpha, \nu+1){}_2F_1(\nu,\nu+1;2+\nu-\alpha;1/2),
\end{equation}
which is valid initially where the integral representation converges if $\Re \alpha < 1$ and $\Re \nu>-1$. The latter is fulfilled by assumption, while the former requires $\Re s>1$ when $\alpha=2-s$ resp. $\Re s < 1$ when $\alpha = s$. Since the absolute convergence requires simultaneously $\Re s \gtrless 1$, which is impossible, there exists no non-empty SOA for $\mathcal{H}(s, \nu)$. Nevertheless, $\mathcal{H}(s, \nu)$ admits an analytic continuation in $\Re \alpha < 1$ to $\mathbb{C}$. To further justify that, let
\begin{equation}
    \phi_{\nu,1}(w) \triangleq \begin{cases}
        \frac{1}{w}\left(\frac{1-w}{1+w}\right)^\nu, \ w\in(0,1), \\
        0, \ w \in[1,\infty);
    \end{cases} \ \text{ and } \quad \phi_{\nu,2}(w) \triangleq \begin{cases}
        0, \ w\in(0,1], \\
        \frac{1}{w}\left(\frac{w-1}{1+w}\right)^\nu, \ w \in(1,\infty),
    \end{cases}
\end{equation}
then clearly $\phi_\nu(w)=\phi_{\nu,1}(w)+\phi_{\nu,2}(w)$ and $\mathcal{M}\{\phi_{\nu,1}\}(s)=M(2-s,\nu)$ is holomorphic in $\Re s>1$, while $\mathcal{M}\{\phi_{\nu,2}\}(s)=M(s,\nu)$ is holomorphic in $\Re s<1$. As $w\to 0^+$ with $|w|<1$, we have $\phi_{\nu,1}(w)=w^{-1}\mathrm{e}^{-2\nu \sum_{k=0}^{\infty} w^{2k+1}/(2k+1)}=w^{-1}-2\nu+O(w)$ which is \textit{algebraic}. By~\cite[Lemma 4.3.6]{bleistein1975asymptotic}, $\mathcal{M}\{\phi_{\nu,1}\}(s)$ can be analytically continued as a meromorphic function into $\Re s\leq 1$ with poles at $s=-(m-1)$ with $m\in\mathbb{Z}_{\geq 0}$, i.e., $s=1,0,-1,\ldots$. On the other hand, as $w\to \infty$ with $|w|\geq 1$, we have $\phi_{\nu,2}(w)=w^{-1}-2\nu \cdot w^{-2}+O(w^{-3})$ which is \textit{algebraic} too. By~\cite[Lemma 4.3.3]{bleistein1975asymptotic}, $\mathcal{M}\{\phi_{\nu,2}\}(s)$ can also be analytically continued into $\Re s\geq 1$ as a meromorphic function with poles at $s=m+1$ with $m\in\mathbb{Z}_{\geq 0}$, i.e., $s=1,2,\ldots$. Therefore, we can conclude that these continued functions are at worst meromorphic functions on $\mathbb{C}$ (\cite[p.115]{bleistein1975asymptotic} provides a similar discussion). Putting $\alpha=2-s$ resp. $\alpha=s$, then
\begin{align}
    M(2-s,\nu) &= 2^{-\nu} \Gamma(s-1) \Gamma(\nu+1) {}_2\tilde{F}_1(\nu,\nu+1;\nu+s;1/2), \\
    M(s,\nu) &= 2^{-\nu} \Gamma(1-s) \Gamma(\nu+1) {}_2\tilde{F}_1(\nu,\nu+1;2+\nu-s;1/2),
\end{align}
where $\Gamma(s-1)$ is meromorphic with simple poles at $s=1,0,-1,\ldots$, and $\Gamma(1-s)$ is meromorphic with simple poles at $s=1,2,\ldots$, while $\Gamma(\nu+1)$ with $\Re \nu>0$ and ${}_2\tilde{F}_1(\cdot)$ in both contribute no simple poles. These observations are consistent with our earlier analysis based on the decay behaviors as $w\to 0^+$ and $\infty$. Moreover, we now show that the \textit{singularity} at $s=1$ is \textit{removable}. Let $s\triangleq 1+\varepsilon$ with $\varepsilon\to 0^+$, $c_0\triangleq \nu+1$, and $F(c)\triangleq { }_2 \tilde{F}_1(\nu, \nu+1 ; c ; 1 / 2)$. Therefore,
\begin{equation}
    \lim_{\varepsilon\to 0^+} \mathcal{H}_\nu(1+\varepsilon) = 2^{-\nu} \Gamma(\nu+1)\left( \Gamma(\varepsilon) F(c_0+\varepsilon) + \Gamma(-\varepsilon) F(c_0-\varepsilon) \right).
\end{equation}
Note that $F(c)$ is entire on $c\in\mathbb{C}$, thus, $F(c_0\pm \varepsilon)=F(c_0)\pm F'(c_0)\varepsilon + O(\varepsilon^2)$, while $\Gamma(\varepsilon)$ possesses the \textit{Laurent expansion} about $\varepsilon=0$, i.e., $\Gamma(\varepsilon)=\varepsilon^{-1}-\gamma + (\gamma^2+\pi^2/6)/2\cdot \varepsilon+O(\varepsilon^2)$ with $\gamma$ as \textit{Euler's constant}, thus, $\Gamma(\pm \varepsilon)=\pm \varepsilon^{-1}-\gamma + O(\varepsilon)$. Putting it together,
\begin{equation}
    \lim_{\varepsilon\to 0^+} \mathcal{H}_\nu(1+\varepsilon) = 2^{-\nu}\Gamma(\nu+1) \left( 2(F'(c_0)-\gamma F(c_0))+O(\varepsilon) \right),
\end{equation}
which is finite, since the $\varepsilon^{-1}$ singularities from $\Gamma(\pm \varepsilon)$ are canceled exactly. Therefore, $\mathcal{H}(s, \nu)$ as a whole extends analytically from $s\in \varnothing$ to $s\in\mathbb{C}$ with \textit{simple poles} at $s=\mathbb{Z}\backslash \{1\}$. In particular,
\begin{equation}
    \mathcal{H}(s, \nu) = 2^{-\nu}\Gamma(\nu+1) (B(2-s,\nu)+B(s,\nu)),
\end{equation}
where $B(\alpha,\nu)\triangleq \Gamma(1-\alpha) {}_2\tilde{F}_1(\nu,\nu+1;2+\nu-\alpha;1/2)$ with with $\alpha\in\{s,2-s,1-s,1+s\}$. Putting $\alpha=1-s$ resp. $\alpha=1+s$ gives $\mathcal{H}(1-s, \nu)=2^{-\nu}\Gamma(\nu+1) (B(1+s,\nu)+B(1-s,\nu))$ as well an analytic continuation with \textit{simple poles} at $s\in\mathbb{Z}\backslash \{0\}$.
\end{proof}

Now, we prove Proposition~\ref{prop:F2G_convert}. Recall that $\mathcal{G}_{\mathfrak{a}}(0, \nu)=\int_{-1}^1 \mathfrak{a}(z)|z|^\nu \mathrm{d}z$, which is holomorphic for $\Re \nu>0$ by Lemma~\ref{lemma:check_node_holo}. By Lemma~\ref{lemma:f_a_w_analy_cont} and~\ref{lemma:h_v_analy_cont}, we have $\mathcal{G}_{\mathfrak{a}}(0, \nu)= \alpha_{n}\cdot 1^\nu\cdot \mathbbm{1}_{\left\{\alpha_n>0\right\}}+\alpha_{0}\cdot 0^\nu\cdot \mathbbm{1}_{\left\{\alpha_0>0\right\}} + \int_0^{+\infty} f_{\mathfrak{a}}(w) \phi_\nu(w) \mathrm{d} w$. By Corollary~\ref{corol:f_a_1_2_cont} and Lemma~\ref{lemma:h_v_analy_cont}, then $f_{\mathfrak{a},i}(w) \phi_{\nu,j}(w)=0$, $\forall i,j\in\{1,2\}$ and $i\neq j$, thus, $\int_0^{+\infty} f_{\mathfrak{a}}(w) \phi_\nu(w) \mathrm{d} w=\sum_{i\in\{1,2\}}\int_0^{+\infty} f_{\mathfrak{a},i}(w) \phi_{\nu,i}(w) \mathrm{d} w$. By the \textit{Parseval formula for the Mellin transform}, we obtain
\begin{equation}~\label{eq:Mellin_Barnes_F2G}
    \begin{aligned}
        \int_0^{+\infty} f_{\mathfrak{a},1}(w) \phi_{\nu,1}(w) \mathrm{d} w &\triangleq \frac{1}{2 \pi \mathrm{i}} \int_{c_1-\mathrm{i} \infty}^{c_1+\mathrm{i} \infty} \mathscr{F}_{\mathfrak{a}, 1}(s) M(1+s,\nu) \mathrm{d}s, \ \forall c_1 \in(-1,0), \\
        \int_0^{+\infty} f_{\mathfrak{a},2}(w) \phi_{\nu,2}(w) \mathrm{d} w &\triangleq \frac{1}{2 \pi \mathrm{i}} \int_{c_2-\mathrm{i} \infty}^{c_2+\mathrm{i} \infty} \mathscr{F}_{\mathfrak{a}, 2}(s) M(1-s,\nu) \mathrm{d}s, \ \forall c_2 \in(0,1),
    \end{aligned}
\end{equation}
where $M(\alpha,\nu)$ is defined in Lemma~\ref{lemma:h_v_analy_cont}. In particular,
\begin{equation}
    \begin{aligned}
        M(1+s, \nu)&\triangleq 2^{-\nu} \Gamma(-s) \Gamma(\nu+1)_2 \tilde{F}_1\left(\nu, \nu+1 ; 1+\nu-s ; 1/2\right), \\
        M(1-s, \nu)&\triangleq 2^{-\nu} \Gamma(s) \Gamma(\nu+1){ }_2 \tilde{F}_1\left(\nu, \nu+1 ; 1+\nu+s ; 1/2\right).
    \end{aligned}
\end{equation}
Additionally, $M(1+s, \nu)$ and $M(1-s, \nu)$ are meromorphic on $\mathbb{C}$, with simple poles at $s=0,1,2, \ldots$ and $s=0,-1,-2, \ldots$, respectively. Moreover, $\mathscr{F}_{\mathfrak{a}, 1}(s)$ and $\mathscr{F}_{\mathfrak{a}, 2}(s)$ are entire by Corollary~\ref{corol:f_a_1_2_cont}. Closing the contour of the first resp. the second integral in~\eqref{eq:Mellin_Barnes_F2G} at right resp. at left, we have

\begin{equation}
    \begin{aligned}
        \frac{1}{2 \pi \mathrm{i}} \int_{c_1-\mathrm{i} \infty}^{c_1+\mathrm{i} \infty} \mathscr{F}_{\mathfrak{a}, 1}(s) M(1+s,\nu) \mathrm{d}s &= - \sum_{t=0}^\infty \mathrm{Res}_{s=t} M(1+s,\nu) \mathscr{F}_{\mathfrak{a}, 1}(s), \\
        \frac{1}{2 \pi \mathrm{i}} \int_{c_2-\mathrm{i} \infty}^{c_2+\mathrm{i} \infty} \mathscr{F}_{\mathfrak{a}, 2}(s) M(1-s,\nu) \mathrm{d}s &= \sum_{t=0}^\infty \mathrm{Res}_{s=-t} M(1-s,\nu) \mathscr{F}_{\mathfrak{a}, 2}(s).
    \end{aligned}
\end{equation}

Since $\operatorname{Res}_{s=t} \Gamma(-s)=(-1)^{t+1} / t!$ and $\operatorname{Res}_{s=-t} \Gamma(s)=(-1)^t / t!$, then
\begin{equation}
    \begin{aligned}
        \frac{1}{2 \pi \mathrm{i}} \int_{c_1-\mathrm{i} \infty}^{c_1+\mathrm{i} \infty} \mathscr{F}_{\mathfrak{a}, 1}(s) M(1+s,\nu) \mathrm{d}s &= -\sum_{t=0}^{\infty} \frac{(-1)^{t+1}}{t!} G(\nu,t) \mathscr{F}_{\mathfrak{a}, 1}(t), \\
        \frac{1}{2 \pi \mathrm{i}} \int_{c_2-\mathrm{i} \infty}^{c_2+\mathrm{i} \infty} \mathscr{F}_{\mathfrak{a}, 2}(s) M(1-s,\nu) \mathrm{d}s &= \sum_{t=0}^{\infty}\frac{(-1)^t}{t!} G(\nu,t) \mathscr{F}_{\mathfrak{a}, 2}(-t),
    \end{aligned}
\end{equation}
where $G(\nu,t) \triangleq 2^{-\nu}\cdot \Gamma(\nu+1){}_2\tilde{F}_1(\nu, \nu+1 ; 1+\nu-t ; 1 / 2)$.

\subsection{Proof of Theorem~\ref{thm:G2F_convert2}}~\label{sec:G2F_convert2}

By definition of $\mathscr{F}_{\mathfrak{a}, 1}(s)$ in~\eqref{eq:atom2F12} and definition of $\mathscr{G}_{|\mathfrak{a}|}(\nu)$ in~\eqref{eq:atom2G} where $\beta_i = (1+z_i)/2\cdot \alpha_i$, $w_i=(1-z_i)/(1+z_i)$ and $z_i\in(0,1)$ for every $i=1,\ldots,n-1$, we write~\eqref{eq:evo_G2F1} explicitly in terms of $\alpha_i$ and $z_i$, i.e.,
\begin{equation}~\label{eq:evo_G2F1_reform}
    \sum_{i=1}^{n-1} \alpha_i \left( 1/2 \cdot (1-z_i)^s(1+z_i)^{1-s} \right) = \sum_{i=1}^{n-1} \alpha_i \left( \sum_{t=0}^{\infty} \frac{(-1)^t}{t!} A(s,t) z_i^t \right).
\end{equation}
Therefore, to prove~\eqref{eq:evo_G2F1}, it suffices to prove
\begin{equation}~\tag{\ref{eq:evo_G2F1_reform1}}
    1/2 \cdot(1-z)^s(1+z)^{1-s} = \sum_{t=0}^{\infty} \frac{(-1)^t}{t!} A(s,t) z^t, \quad \forall z\in(0,1), \ \forall s\in\mathbb{C},
\end{equation}
where $A(s,t)=2^{-s} \Gamma(s+1){}_2\tilde{F}_1(s-1,s+1;1-t+s;1/2)$. For $|z|<1$, the \textit{binomial series} \textit{converges absolutely} for every $s\in \mathbb{C}$, i.e., $(1+z)^s=\sum_{k=0}^{\infty}\binom{s}{k} z^k$, cf.~\cite[Theorem 245]{bromwich2005introduction}, thus, $(1-z)^s(1+z)^{1-s}=\sum_{t=0}^{\infty} a_t(s) z^t$ holds true for $|z|<1$ with $a_t(s)\triangleq \sum_{k=0}^t(-1)^k\binom{s}{k}\binom{1-s}{t-k}=(-1)^t \sum_{k=0}^t(-1)^k \frac{(-s)_k(s-1)_{t-k}}{k!(t-k)!}$, where the binomial coefficients are represented using \textit{(rising) Pochhammer’s Symbol}~\cite[(1.2.6)]{NIST:DLMF}, i.e., $\binom{s}{k}=\frac{(-1)^k(-s)_k}{k!}$ for $s\in\mathbb{C}$ and $k\in\mathbb{Z}_{\geq 0}$. Therefore, to prove~\eqref{eq:evo_G2F1_reform1}, it suffices to prove
\begin{equation}~\label{eq:evo_G2F1_reform2}
    a_t(s)=\frac{(-1)^t}{t!} 2A(s,t).
\end{equation}
Using the identity $(a)_{n-k}=\frac{(-1)^k(a)_n}{(1-a-n)_k}$ for $0 \leq k \leq n$, cf.~\cite[p.32]{rainville1960special}, we obtain $a_t(s)=\frac{(-1)^t(s-1)_n}{t!}{ }_2 F_1(-t,-s ; 2-s-t ;-1)$ by recognizing that $\sum_{k=0}^t \frac{(-t)_k(-s)_k}{(2-s-t)_k} \frac{(-1)^k}{k!}={ }_2 F_1(-t,-s ; 2-s-t ;-1)$ where the series terminates at $k=t$. Using the \textit{Pfaff transformation}~\cite[(15.3.4)]{abramowitz1965handbook}, then $a_t(s)=\frac{(-1)^t(s-1)_n}{t!} 2^t{ }_2 F_1(-t, 2-t ; 2-s-t ; 1 / 2)=\frac{(-1)^t}{t!} 2^t \frac{\Gamma(s+1)}{\Gamma(s-t+1)}{ }_2 F_1(-t, 2-t ; s-t+1 ; 1 / 2)$. Next by applying the \textit{Euler transformation}~\cite[(15.3.3)]{abramowitz1965handbook}, then $a_t(s)=\frac{(-1)^t}{t!} 2^t \frac{\Gamma(s+1)}{\Gamma(s-t+1)} \cdot 2^{1-s-t} \cdot{ }_2 F_1(s+1, s-1 ; s-t+1 ; 1 / 2)=\frac{(-1)^t}{t!} 2^{1-s} \Gamma(s+1)_2 \tilde{F}_1(s-1, s+1 ; 1-t+s ; 1 / 2)$, resulting in the RHS of~\eqref{eq:evo_G2F1_reform2}. This completes the proof of~\eqref{eq:evo_G2F1}. In the following, we prove~\eqref{eq:evo_F122G} in the same vein. By definition of $\mathscr{F}_{\mathfrak{a}, 1}(s)$ in~\eqref{eq:atom2F12} and definition of $\mathscr{G}_{|\mathfrak{a}|}(\nu)$ in~\eqref{eq:atom2G} as well as symmetry in Corollary~\ref{corol:F12_sym}, we write~\eqref{eq:evo_F122G} explicitly in terms of $\alpha_i$ and $w_i\triangleq (1-z_i)/(1+z_i)\in(0,1)$ for every $i=1,\ldots,n-1$, cf. Lemma~\ref{lemma:f_a_w_analy_cont}, i.e.,
\begin{equation}~\label{eq:evo_F2G1_reform1}
    \sum_{i=1}^{n-1} \alpha_i (1-w_i)^\nu(1+w_i)^{-\nu} = \sum_{i=1}^{n-1} \alpha_i \left( \sum_{t=0}^{\infty} \frac{(-1)^t}{t!} G(\nu, t) w_i^{t} \right).
\end{equation}
Therefore, it suffices to prove 
\begin{equation}~\tag{\ref{eq:evo_F2G1_reform2}}
    (1-w)^\nu(1+w)^{-\nu}=\sum_{t=0}^{\infty} \frac{(-1)^t}{t!} G(\nu, t) w^t, \quad \forall w\in(0,1), \ \forall \nu\in\mathbb{C},
\end{equation}
where $G(\nu,t)=2^{-\nu}\cdot \Gamma(\nu+1){}_2\tilde{F}_1(\nu, \nu+1 ; 1+\nu-t ; 1 / 2)$. For $|w|<1$, $(1-w)^\nu(1+w)^{-\nu}=\sum_{t=0}^{\infty} b_t(\nu) w^t$ holds true for every $\nu\in\mathbb{C}$ with $b_t(\nu)\triangleq \sum_{k=0}^t(-1)^{t-k}\binom{\nu}{t-k}\binom{-\nu}{k}=\frac{(-\nu)_n}{t!} \sum_{k=0}^t \frac{(-t)_k(\nu)_k}{(1+\nu-t)_k} \frac{(-1)^k}{k!}$, where the identity $\binom{s}{k}=\frac{(-1)^k(-s)_k}{k!}$ for $s\in\mathbb{C}$ and $k\in\mathbb{Z}_{\geq 0}$, $(a)_{n-k}=\frac{(-1)^k(a)_n}{(1-a-n)_k}$ for $0 \leq k \leq n$ are applied. This expression can be written as $\frac{(-\nu)_n}{t!}{ }_2 F_1(-t, \nu ; 1+\nu-t ;-1)$ by recognizing that $\sum_{k=0}^t \frac{(-t)_k(\nu)_k}{(1+\nu-t)_k} \frac{(-1)^k}{k!}={ }_2 F_1(-t, \nu ; 1+\nu-t ;-1)$. Additionally, since $\frac{(-\nu)_n}{t!}=\frac{1}{t!}\frac{\Gamma(-\nu+t)}{\Gamma(-\nu)}=\frac{(-1)^t}{t!}\frac{\Gamma(\nu+1)}{\Gamma(\nu-t+1)}$, we obtain $b_t(\nu)=\frac{(-1)^t}{t!} \frac{\Gamma(\nu+1)}{\Gamma(\nu-t+1)}{ }_2 F_1(-t, \nu ; \nu-t+1 ;-1)$. By using the \textit{Pfaff transformation} and the \textit{Euler transformation}, we obtain ${ }_2 F_1(-t, \nu ; \nu-t+1 ;-1)=2^t{ }_2 F_1(-t, 1-t ; \nu-t+1 ; 1 / 2)=2^t \cdot 2^{-(\nu+t)}{ }_2 F_1(\nu+1, \nu ; \nu-t+1 ; 1 / 2)$. Putting it together, we obtain $b_t(\nu)=\frac{(-1)^t}{t!}G(\nu,t)$, completing the proof of~\eqref{eq:evo_F2G1_reform2}.

\subsection{Proof of Lemma~\ref{lemma:q_base}}~\label{sec:q_base}

By induction from Table~\ref{tab:bsc_variable_evo}, we obtain $\Fai[\varoast q](s)=\sum_{i=0}^{Q-1}\binom{2Q}{i}C^{i}A^{r_i}(s)$. Recall that by definition $A(s)=C^{1/2}\cdot \mathrm{e}^{\Delta (s-1/2)}$, then $A^{r}(s)=C^{r/2}\cdot (S_r-D_r)^{s-1/2}(S_r+D_r)^{1/2-s}$. Substituting $A^{r}(s)$ with $r=2Q-2i$ into $\Fai[\varoast q](s)$, we obtain~\eqref{eq:F_q_base}. By~\eqref{eq:evo_F122G} in Theorem~\ref{thm:G2F_convert2}, we obtain~\eqref{eq:G_q_base}. It is evident from Corollary~\ref{corol:G2G_connect} and~\ref{corol:F_trans_parts} that for any bit-channel of $\mathfrak{a} \equiv W_{\operatorname{BSC}(p)}$, we have $\Ba[\ldots]=1-\Ga[\ldots](0)+2\Fai[\ldots](1/2)$. Starting from $\Ga[\ldots](0)$, we have
\begin{equation}
    \Ga[\varoast q](0) = \sum_{i=0}^{Q-1} \binom{2Q}{i} C^{i}S_{r_i}=\sum_{i=0}^{Q-1} \binom{2Q}{i}(p^{i}\barp^{2Q-i}+p^{2Q-i}\barp^{i}) = \sum_{i=0,i\neq Q}^{2Q}\binom{2Q}{i}p^{i}\barp^{2Q-i}=1-\binom{2Q}{Q}C^Q,
\end{equation}
where the \textit{binomial theorem} is applied. Therefore,
\begin{equation}
    \Ba[\varoast q] = 1-\Ga[\varoast q](0)+2\Fai[\varoast q](1/2)=C^Q\left(\binom{2Q}{Q}+2\sum_{i=0}^{Q-1} \binom{2Q}{i} \right) = C^Q \sum_{i=0}^{2Q} \binom{2Q}{i}=2^{2Q}C^Q,
\end{equation}
completing the proof of~\eqref{eq:Ba_q_base}.

\subsection{Proof of Lemma~\ref{lemma:GF_bin_base}}~\label{sec:GF_bin_base}

By Lemma~\ref{lemma:G_comp_evo} and~\ref{lemma:q_base}, after applying the $\ell$-fold $\boxast$-convolution to~\eqref{eq:G_q_base}, we have $\Ga[\varoast q \boxast \ell](\nu)=\Ga[\varoast q]^{L}(\nu)$ with $L\triangleq 2^\ell$, thus the RHS of~\eqref{eq:G_10_base}. Next, we prove that the LHS of~\eqref{eq:G_10_base} equals the RHS. By definition of $R_{\mathbf{j}}$, $\alpha_{\mathbf{j}}$ and $r_i$, we have $\sum_{\mathbf{j}\in\mathcal{J}_{QL}}R_{\mathbf{j}} \cdot C^{\alpha_\mathbf{j}}\prod_{i=0}^{Q-1}\left(S_{r_i}^{j_{r_i}(1-\nu)}D_{r_i}^{j_{r_i}\nu}\right) = \sum_{\mathbf{j}\in\mathcal{J}_{QL}}\frac{L!}{\prod_{i=0}^{Q-1}j_{r_i}!}\prod_{i=0}^{Q-1}\left(\binom{2Q}{i} C^i S_{r_i}^{(1-\nu)}D_{r_i}^{\nu}\right)^{j_{r_i}}=\left(\sum_{i=0}^{Q-1} \binom{2Q}{i} C^iS_{r_i}^{1-\nu}D_{r_i}^{\nu} \right)^L$ where the last equality is obtained by the \textit{multinomial theorem}. By~\eqref{eq:evo_G2F1} in Theorem~\ref{thm:G2F_convert2}, we obtain~\eqref{eq:F_10_base}.

\subsection{Proof of Proposition~\ref{prop:Ba_G_L}}~\label{sec:Ba_G_L}

Let $r_{\mathbf{j}}\triangleq(Y_{\mathbf{j}}/X_{\mathbf{j}})^2\in(0,1)$ with $X_{\mathbf{j}} \triangleq \prod_{i=0}^{Q-1} S_{r_i}^{j_{r_i}}$ and $Y_{\mathbf{j}} \triangleq \prod_{i=0}^{Q-1} D_{r_i}^{j_{r_i}}$, by Lemma~\ref{lemma:q_base}, we have
\begin{align}
    \Ba[\varoast q \boxast \ell] &= 1-\Ga[\varoast q \boxast \ell](0)+2\Fai[\varoast q \boxast \ell](1/2)  \nonumber \\
    &= 1+\sum_{\mathbf{j} \in \mathcal{J}_{Q L}} R_{\mathbf{j}} \cdot C^{\alpha_{\mathbf{j}}}\left[\left(\left(\prod_{i=0}^{Q-1} S_{r_i}^{j_{r_i}}\right)^2-\left(\prod_{i=0}^{Q-1} D_{r_i}^{j_{r_i}}\right)^2\right)^{1 / 2}-\prod_{i=0}^{Q-1} S_{r_i}^{j_{r_i}}\right] \nonumber \\
    &=1+\sum_{\mathbf{j} \in \mathcal{J}_{Q L}} R_{\mathbf{j}} \cdot C^{\alpha_{\mathbf{j}}}\cdot X_{\mathbf{j}}\left( (1-r_{\mathbf{j}})^{1/2}-1 \right) \nonumber \\
    &= 1-\sum_{\mathbf{j} \in \mathcal{J}_{Q L}} R_{\mathbf{j}} \cdot C^{\alpha_{\mathbf{j}}} \sum_{\kappa=1}^{\infty} \beta_\kappa Y_{\mathbf{j}}^{2\kappa}X_{\mathbf{j}}^{1-2\kappa} \nonumber \nonumber \\
    &=1-\sum_{\kappa=1}^{\infty} \beta_\kappa \sum_{\mathbf{j}\in\mathcal{J}_{QL}}\frac{L!}{\prod_{i=0}^{Q-1}j_{r_i}!}\prod_{i=0}^{Q-1}\left(\binom{2Q}{i} C^{r_i } S_{r_i}^{1-2\kappa}D_{r_i}^{2\kappa}\right)^{j_{r_i}} \nonumber \\
    &=1-\sum_{\kappa=1}^{\infty} \beta_\kappa\left(\sum_{i=0}^{Q-1}\binom{2 Q}{i} C^i S_{r_i}^{1-2 \kappa} D_{r_i}^{2 \kappa} \right)^L \nonumber \\
    &= 1-\sum_{\kappa=1}^{\infty} \beta_\kappa \Ga[\varoast q]^L(2\kappa),
\end{align}
where the \textit{Taylor series} of $1-\sqrt{1-z}=\sum_{\kappa=1}^{\infty}\beta_\kappa z^\kappa$ with $\beta_\kappa\triangleq (2\kappa-3)!!/(2^\kappa \kappa!)>0$ for every $|z|<1$ is applied. Note that $(-1)!!\equiv 1$ by convention. This completes the proof of the first equality of~\eqref{eq:B2_bin_base}. By Lemma~\ref{lemma:GF_bin_base}, the even moments generated by $\Ga[\varoast q](\nu)$ can be expressed as
\begin{align}
    &\Ga[\varoast q](2\kappa) = \sum_{i=0}^{Q-1} \binom{2Q}{i} C^{i}S_{r_i}^{1-2\kappa}D_{r_i}^{2\kappa} \nonumber \\
    =& \sum_{i=0}^{Q-1} \binom{2Q}{i}(p^{i}\barp^{2Q-i}+p^{2Q-i}\barp^{i}) \tanh^{2\kappa}((Q-i)|\Delta|) \nonumber \\
    =& \sum_{i=0,i\neq Q}^{2Q}\binom{2Q}{i}p^{i}\barp^{2Q-i} \tanh^{2\kappa}((Q-i)|\Delta|) \nonumber \\
    =& \sum_{i=0}^{2Q}\binom{2Q}{i}p^{i}\barp^{2Q-i} \tanh^{2\kappa}((Q-i)|\Delta|) \nonumber \\
    =& \mathbb{E}_{K\sim\mathrm{Bin}(2Q,p)}\left[ \tanh^{2\kappa}(\lvert Q-K \rvert \cdot \lvert\Delta \rvert) \right].
\end{align}
Let $X_Q \triangleq \tanh(\lvert Q-K \rvert \cdot \lvert\Delta \rvert)$ with $K\sim\mathrm{Bin}(2Q,p)$. Define $L\triangleq 2^\ell$ and let $X_{Q,(1)}, \ldots, X_{Q,(L)}$ be i.i.d. copies of $X_Q$, then
\begin{equation}
    \Ga[\varoast q]^L(2\kappa) = \mathbb{E}\left[ X_Q^{2\kappa} \right]^L \stackrel{(a)}{=} \mathbb{E}\left[\prod_{i=1}^L X_{Q,(i)}^{2 \kappa}\right] = \mathbb{E}\left[ \left(\prod_{i=1}^L X_{Q,(i)}\right)^{2\kappa} \right] = \mathbb{E}\left[ X_{QL}^{2\kappa} \right],
\end{equation}
where $X_{QL}\triangleq \prod_{i=1}^L X_{Q,(i)}\in[0,1)$ and $(a)$ uses independence of the $L$ copies. Therefore,
\begin{equation}
    \sum_{\kappa=1}^{\infty} \beta_\kappa \Ga[\varoast q]^L(2\kappa) = \mathbb{E}\left[\sum_{\kappa=1}^{\infty} \beta_\kappa X_{QL}^{2 \kappa}\right] = \mathbb{E}\left[1-\sqrt{1-X_{QL}^2}\right].
\end{equation}
Substituting this result into the first equality of~\eqref{eq:B2_bin_base} yields the second equality. In what follow, we provide a second proof of the second equality via the MGF. First, we write $X^2_{QL}=\mathrm{e}^{2\log X_{QL}}=\mathrm{e}^{2\sum_{i=1}^{L} \log X_{Q,(i)}}=\mathrm{e}^{2\sum_{i=1}^{L}Y_i}$ with $Y_i\triangleq \log\tanh (|Q-Z_i| \cdot|\Delta|)$ and $Z_i \stackrel{\mathrm{i.i.d.}}{\sim} \mathrm{Bin}(2Q,p)$. Since all r.v.s $Y_i$ with $i=1,\ldots,L$ are i.i.d, they possess the unique MGF, i.e., $M_{Y_i}(t)\equiv M_{Y}(t)=\mathbb{E}[\mathrm{e}^{tY}]$ with $Y\triangleq \log X_Q= \log\tanh (|Q-K| \cdot|\Delta|)$. By properties of the MGF~\cite[(8.39)]{kobayashi2011probability}, we have $M_{\log X_{QL}}(t)=\prod_{i=1}^{L}M_{Y_i}(t)=M_Y^L(t)$, thus, $M_{2\log X_{QL}}(t)=M_Y^L(2t)$. Using the power series expansion of the exponential function and the definition of MGF, we have $M_{\mathrm{e}^X}(t)=\sum_{\kappa=0}^{\infty}\frac{t^\kappa}{\kappa!}\mathbb{E}[\mathrm{e}^{nX}]=\sum_{\kappa=0}^{\infty}\frac{t^\kappa}{\kappa!}M_X(\kappa)$ for a r.v. $X$. By applying this property to $X^2_{QL}$, we have
\begin{equation}
    M_{X^2_{QL}}(t)=\sum_{\kappa=0}^{\infty}\frac{t^\kappa}{\kappa!}M_{2\log X_{QL}}(2\kappa)=\sum_{\kappa=0}^{\infty}\frac{t^\kappa}{\kappa!}M_Y^L(2\kappa).
\end{equation}
On the other hand, recall the \textit{fractional power} representation~\cite[p.vii]{schilling2012bernstein}, i.e., $\lambda^\alpha=\frac{\alpha}{\Gamma(1-\alpha)} \int_0^{\infty}\left(1-\mathrm{e}^{-\lambda t}\right) t^{-\alpha-1} \mathrm{d} t$ for $\lambda>0$ and $\alpha\in(0,1)$. Let $\alpha=1/2$ and $\lambda=1-X$ with a r.v. $X<1$, then $\sqrt{1-X}=\frac{1}{2 \sqrt{\pi}} \int_0^{\infty}\left(1-\mathrm{e}^{-t} \mathrm{e}^{t X}\right) t^{-3 / 2} \mathrm{d} t$. Therefore, $\mathbb{E}[\sqrt{1-X}]=\frac{1}{2 \sqrt{\pi}} \int_0^{\infty}\left(1-\mathrm{e}^{-t} M_X(t)\right) t^{-3 / 2} \mathrm{d} t$. Let $X=X^2_{QL}\in[0,1)$, we obtain
\begin{align}
    \mathbb{E}[\sqrt{1-X^2_{QL}}] &= \frac{1}{2 \sqrt{\pi}} \int_0^{\infty}\left(1-\mathrm{e}^{-t} \sum_{\kappa=0}^{\infty}\frac{t^\kappa}{\kappa!}M_Y^L(2\kappa) \right) t^{-3 / 2} \mathrm{d} t \nonumber \\
    &= \frac{1}{2 \sqrt{\pi}} \int_0^{\infty}\left(\left(1-\mathrm{e}^{-t}\right)-\sum_{\kappa=1}^{\infty} \mathrm{e}^{-t} \frac{t^\kappa}{\kappa!} M_Y^L(2\kappa) \right) t^{-3 / 2} \mathrm{d} t \nonumber \\
    &=\frac{1}{2 \sqrt{\pi}} \int_0^{\infty} \frac{1-\mathrm{e}^{-t}}{t^{3 / 2}} \mathrm{d} t-\sum_{\kappa=1}^{\infty}  \frac{M_Y^L(2\kappa)}{2 \sqrt{\pi} \kappa!} \int_0^{\infty} \mathrm{e}^{-t} t^{\kappa-\frac{3}{2}} \mathrm{d} t \nonumber \\
    &=1-\sum_{\kappa=1}^{\infty} \frac{\Gamma\left(\kappa-\frac{1}{2}\right)}{2 \sqrt{\pi} \kappa!}M_Y^L(2\kappa).
\end{align}
Since $M_Y(t)=\sum_{i=0}^{2Q}\binom{2Q}{i}p^{i}\barp^{2Q-i} \mathrm{e}^{t\cdot \log \tanh(|Q-i|\cdot |\Delta|)}=\sum_{i=0}^{2Q}\binom{2Q}{i}p^{i}\barp^{2Q-i} \tanh^t(|Q-i|\cdot |\Delta|)$, then $M_Y(2\kappa)=\Ga[\varoast q](2\kappa)$. Moreover, one verifies $\beta_\kappa=\frac{\Gamma\left(\kappa-\frac{1}{2}\right)}{2 \sqrt{\pi} \kappa!}$. Putting it together, this completes the second proof.

\subsection{Bhattacharyya parameters of the BSC at level \texorpdfstring{$k=3$ and $k=4$}{k=3 and k=4}}~\label{sec:BSC_34}
\begin{exmp}
At level $k=4$, the Bhattacharyya parameters are given by
\begin{align}
    \Ba[\boxast 4] &= (1- M_{32})^{1/2}, \\
    \Ba[\boxast 3 \varoast] &= 1-M_{16}, \\
    \Ba[\boxast 2 \varoast \boxast] &= \frac{1}{4}\left(1-M_8\right)\left(\left(M_{16}+6 M_8+1\right)^{1/2}+M_8+3\right), \\
    \Ba[\boxast 2 \varoast 2] &= \left(1-M_8\right)^2, \\
    \Ba[\boxast \varoast \boxast 2] &= 1-\frac{1}{16}\left(1+M_4\right)^4+\frac{1}{16}\left(\left(1+M_4\right)^8-256 M_{16}\right)^{1/2}, \\ 
    \Ba[\boxast \varoast \boxast \varoast] &= \frac{1}{16}\left(1-M_4\right)^2\left(\left(M_{8}+6 M_4+1\right)^{1 / 2}+M_4+3\right)^2, \\
    \Ba[\boxast \varoast 2 \boxast] &=\frac{\left(1-M_4\right)^2}{64}\left(16\left(1-M_4\right) \left(1+6 M_4+M_4^2\right)^{1/2}+\left(1+28 M_4+70 M_4^2+28 M_4^3+M_4^4\right)^{1/2}\right.  \nonumber \\
        & \left.\quad+8\left(1+M_4\right) \left(1+14 M_4+M_4^2\right)^{1/2}+39+18 M_4-9 M_4^2\right), \\
    \Ba[\boxast \varoast 3] &= \left(1-M_4\right)^4, \\
    \Ba[\varoast \boxast 3]  &= \left(S_2^{16}-M_{16} \right)^{1/2}+1-S_2^{8}, \\
    \Ba[\varoast \boxast 2 \varoast]  &= \left(\left(S_2^8-M_8\right)^{1/2}+1-S_2^4\right)^2, \\
    \Ba[\varoast \boxast \varoast \boxast] &= 64 C^2 \bar{C}^2 \left( S_2^8-M_8 \right)^{1/2} + 8 C \bar{C} S_2^2 \left(S_2^4-M_4\right)^{1/2}\left(S_2^4+3 M_4\right)^{1/2}+\frac{1}{4}\left(S_2^4-M_4\right)\left( S_2^8+6 S_2^4 M_4+M_8 \right)^{1/2} \nonumber \\
            &\quad+32C^2\bar{C}^2-M_4+S_2^4-\frac{1}{4}\left(32C^2\bar{C}^2-M_4+S_2^4\right)^2, ~\label{eq:Ba_1010} \\
    \Ba[\varoast \boxast \varoast 2] &= \left(\left(S_2^4-M_4\right)^{1/2}+4C\bar{C}\right)^4, \\
    \Ba[\varoast 2 \boxast 2] &= 256 C^4 \left(S_2^8-D_2^{8}\right)^{1/2}+256 C^3 \left(S_2^6S_4^2-D_2^{6}D_4^2\right)^{1/2} +96 C^2 \left(S_2^4 S_4^4-D_2^4 D_4^4\right)^{1/2} \nonumber \\
                & + 16 C \left(S_2^2 S_4^6-D_2^2 D_4^6\right)^{1/2} +\left(S_4^8-D_4^8\right)^{1/2} + 24 C^2\left(1-9 C^2+36 C^4-54 C^6\right),
                ~\label{eq:Ba_1100} \\
    \Ba[\varoast 2 \boxast \varoast] &= \left(32 \sqrt{2} C^3 S_4^{1/2} + 16 C^{2 s+1} S_2^{1/2} S_6^{1/2}  +2 \sqrt{2} C^2 S_8^{1/2}+12 C^2-36 C^4\right)^2, \\
    \Ba[\varoast 3 \boxast] &= 6272 \sqrt{2}C^7 S_4^{1/2} + C^6\left(1568 \sqrt{2}S_8^{1/2}+6272 S_2^{1/2} S_6^{1/2}\right) \nonumber \\
                & \quad+C^5\left(128 \sqrt{2}S_{12}^{1/2}+1792 S_4^{1/2} S_8^{1/2}+896 S_2^{1/2} S_{10}^{1/2}\right) \nonumber \\
                & \quad+C^4\left(2 \sqrt{2}S_{16}^{1/2}+224 S_6^{1/2} S_{10}^{1/2}+112 S_4^{1/2} S_{12}^{1/2}+32 S_2^{1/2} S_{14}^{1/2}\right) + 140 C^4-4900 C^8, ~\label{eq:Ba_1110} \\
    \Ba[\varoast 4] &= 65536 C^8.
\end{align}
\end{exmp}

\begin{proof}
Note that we have the following recurrence relations: $S_i^2-D_i^2=4C^i$, $S_i^4-D_i^4=8C^iS_{2i}$, $S_i S_j-D_i D_j=2 C^i S_{j-i}$, and $S_i S_j+D_i D_j=2 S_{i+j}$ with $j>i\in \mathbb{Z}_{> 0}$. Additionally, $D_2\equiv M_1$, $2S_2\equiv 1+M_2$, $4C=1-M_2$ and $S_2\equiv 1-2C$. Let $\tau = 1/2+\mathrm{i}u$. Some intermediate results of importance of each bit-channel at level $k=3$ and $k=4$ are summarized as follows.
\begin{enumerate}
    \item[\textit{i})] for $\mathfrak{a}^{\boxast 3}$:
    \begin{align}
        \mathscr{G}_{|\mathfrak{a}^{\boxast 3}|}(\nu) &= M_8^\nu, \\
        \mathscr{F}_{\mathfrak{a}^{\boxast 3},1}(s) &= \frac{1}{2} \left( \frac{1-M_8}{1+M_8} \right)^s(1+M_8), \quad \alpha_0^{\boxast 3} \equiv 0,  \\
        \mathfrak{B}(|\mathfrak{a}^{\boxast 3}|) &= 2\Fai[\boxast 3](1/2) = (1- M_{16})^{1/2}. ~\label{eq:B_a000}
    \end{align}
    \begin{enumerate}
        \item for $\mathfrak{a}^{\boxast 4}$:
        \begin{equation}
            \mathfrak{B}(|\mathfrak{a}^{\boxast 4}|) = 2\Fai[\boxast 4](1/2) = (1- M_{32})^{1/2}.
        \end{equation}
        \item for $\mathfrak{a}^{\boxast 3 \varoast}$:
        \begin{equation}
            \Ba[\boxast 3 \varoast] = 1-M_{16}.
        \end{equation}
    \end{enumerate}
    \item[\textit{ii})] for $\mathfrak{a}^{\boxast 2 \varoast}$:
    \begin{align}
        \mathscr{F}_{\mathfrak{a}^{\boxast 2},1}(s) &= \frac{1}{2} \left( \frac{1-M_4}{1+M_4} \right)^s(1+M_4), \quad \alpha_0^{\boxast 2} \equiv 0,  \\
        \tilde{\varphi}^{\boxast 2}(u) &= \mathscr{F}_{\mathfrak{a}^{\boxast 2},1}(\tau) \mathscr{F}_{\mathfrak{a}^{\boxast 2},1}(\tau^*)=  \frac{1}{4} (1-M_4^2) \equiv \mathscr{F}_{\varphi}^{\boxast 2}, \\
        \mathscr{F}_{\psi, 1}^{\boxast 2}(s) &= \mathscr{F}_{\psi, 2}^{\boxast 2}(s) \equiv 0, \\
        \mathscr{F}_{\mathfrak{a}^{\boxast 2 \varoast},1}(s) &= \mathscr{F}^2_{\mathfrak{a}^{\boxast 2},1}(s) = \frac{1}{4} \left( \frac{1-M_4}{1+M_4} \right)^{2s}(1+M_4)^2, \quad \alpha_0^{\boxast 2 \varoast} = 2 \mathscr{F}_{\varphi}^{\boxast 2} = \frac{1}{2} (1-M_4^2), \\
        \Ba[\boxast 2 \varoast] &= 2\Fai[\boxast 2 \varoast](1/2) + \alpha_0^{\boxast 2 \varoast} = 1-M_8. ~\label{eq:B_a001}
    \end{align}
    \begin{enumerate}
        \item for $\mathfrak{a}^{\boxast 2 \varoast \boxast}$:
        \begin{align}
            \Fai[\boxast 2 \varoast](t) + \Fai[\boxast 2 \varoast](t+1) &= \frac{1}{2}\left(1+M_8\right)\left(\frac{1-M_4}{1+M_4}\right)^{2 t}, \\
            \mathscr{G}_{|\mathfrak{a}^{\boxast 2 \varoast \boxast}|}(t) &= \frac{1}{4}\left(1+M_8\right)^2\left(\frac{2 M_4}{1+M_8}\right)^{2 t}, \\
            \Fai[\boxast 2 \varoast \boxast](s) &= \frac{1}{8}\left(1-M_4^2\right)^{2 s}\left(M_{16}+6 M_8+1\right)^{1-s}, \\
            \Ba[\boxast 2 \varoast \boxast] &= \frac{1}{4}\left(1-M_8\right)\left(\left(M_{16}+6 M_8+1\right)^{1/2}+M_8+3\right). 
        \end{align}
        \item for $\mathfrak{a}^{\boxast 2 \varoast 2}$:
        \begin{align}
            \tilphi[\boxast 2 \varoast](u) &= \Fai[\boxast 2 \varoast](\tau) \Fai[\boxast 2 \varoast](\tau^*) = \frac{1}{16}\left(1-M_8\right)^2 \equiv \Fphi[\boxast 2 \varoast],, \\
            \Fpsii[\boxast 2 \varoast](s) &= \Fpsiii[\boxast 2 \varoast](s) \equiv 0, \\
            \Fai[\boxast 2 \varoast 2](s) &= \Fai[\boxast 2 \varoast]^2(s) + 2 \alpha_0^{\boxast 2 \varoast} \Fai[\boxast 2 \varoast](s) \\
                &= \frac{1}{16}\left(1-M_4\right)^{4 s}\left(1+M_4\right)^{4(1-s)} + \frac{1}{4}\left(1-M_8\right)\left(1-M_4\right)^{2 s}\left(1+M_4\right)^{2(1-s)}, \\
            \alpha_0^{\boxast 2 \varoast 2} &= (\alpha_0^{\boxast 2 \varoast})^2 + 2 \Fphi[\boxast 2 \varoast] = \frac{3}{8}\left(1-M_8\right)^2, \\
            \Ba[\boxast 2 \varoast 2] &= 2\Fai[\boxast 2 \varoast 2](1/2) + \alpha_0^{\boxast 2 \varoast 2} =\left(1-M_8\right)^2.
        \end{align}
    \end{enumerate}
    \item[\textit{iii})] for $\mathfrak{a}^{\boxast \varoast \boxast}$:
    \begin{align}
        \mathscr{F}_{\mathfrak{a}^{\boxast },1}(s) &= \frac{1}{2} \left( \frac{1-M_2}{1+M_2} \right)^s (1+M_2), \quad \alpha_0^{\boxast } \equiv 0, \\
        \tilde{\varphi}^{\boxast }(u) &= \mathscr{F}_{\mathfrak{a}^{\boxast },1}(\tau) \mathscr{F}_{\mathfrak{a}^{\boxast },1}(\tau^*)=  \frac{1}{4} (1-M_4) \equiv \mathscr{F}_{\varphi}^{\boxast}, \\
        \mathscr{F}_{\psi, 1}^{\boxast}(s) &= \mathscr{F}_{\psi, 2}^{\boxast}(s) \equiv 0, \\
        \mathscr{F}_{\mathfrak{a}^{\boxast \varoast},1}(s) &= \mathscr{F}^2_{\mathfrak{a}^{\boxast },1}(s) = \frac{1}{4} \left( \frac{1-M_2}{1+M_2} \right)^{2s}(1+M_2)^2, \quad \alpha_0^{\boxast \varoast} = 2 \mathscr{F}_{\varphi}^{\boxast } = \frac{1}{2} (1-M_4), \\
        \mathscr{F}_{\mathfrak{a}^{\boxast \varoast},1}(t) + \mathscr{F}_{\mathfrak{a}^{\boxast \varoast},1}(t+1) &= \frac{1}{2}\left(1+M_4\right) \left( \frac{1-M_2}{1+M_2} \right)^{2t}, ~\label{eq:Fsum_a010} \\
        w_1^{\boxast \varoast} &= \left( \frac{1-M_2}{1+M_2} \right)^{2}, \quad z_1^{\boxast \varoast} = \frac{2 M_2}{1+M_4}, \quad \alpha_1^{\boxast \varoast} = \frac{1}{2} \left(1+M_4\right), \\
        \mathscr{G}_{|\mathfrak{a}^{\boxast \varoast}|}(\nu) &= 2^{\nu-1} M_2^\nu\left(1+M_4\right)^{1-\nu}, \quad \alpha_n^{\boxast \varoast} \equiv 0, \\
        \mathscr{G}_{|\mathfrak{a}^{\boxast \varoast \boxast}|}(\nu) &= \mathscr{G}_{|\mathfrak{a}^{\boxast \varoast}|}^2(\nu) = 2^{2 \nu-2} M_2^{2 \nu}\left(1+M_4\right)^{2(1-\nu)}, ~\label{eq:Ga_a010} \\
        \alpha_0^{\boxast \varoast \boxast} &= \frac{1}{4}\left(1-M_4\right)\left(3+M_4\right) = 1-\frac{1}{4}\left(1+M_4\right)^2, \quad \alpha_n^{\boxast \varoast \boxast} \equiv 0, \\
        \mathscr{G}_{|\mathfrak{a}^{\boxast \varoast \boxast}|}(t) &= \frac{1}{4} \left( \frac{2 M_2}{1+M_4} \right)^{2t} \left( 1+M_4 \right)^2, \ z_1^{\boxast \varoast \boxast} = \left( \frac{2 M_2}{1+M_4} \right)^{2}, \ \alpha_1^{\boxast \varoast \boxast} = \frac{1}{4} \left( 1+M_4 \right)^{2}, \\
        \mathscr{F}_{\mathfrak{a}^{\boxast \varoast \boxast}, 1}(s) &= \frac{1}{2}\left(\frac{1-z_1^{\boxast \varoast \boxast}}{1+z_1^{\boxast \varoast \boxast}}\right)^s\left(1+z_1^{\boxast \varoast \boxast}\right) \alpha_1^{\boxast \varoast \boxast} = \frac{1}{8}\left(1-M_4\right)^{2 s}\left(M_8+6 M_4+1\right)^{1-s}, ~\label{eq:Fa1_a010} \\
        \mathfrak{B}(|\mathfrak{a}^{\boxast \varoast \boxast}|) &= 2\Fai[\boxast \varoast \boxast](1/2) + \alpha_0^{\boxast \varoast \boxast} \nonumber \\
            &= \frac{1}{4}\left(1-M_4\right)\left(\left(M_{8}+6 M_4+1\right)^{1 / 2}+M_4+3\right). ~\label{eq:B_a010}
    \end{align}
    \begin{enumerate}
        \item for $\mathfrak{a}^{\boxast \varoast \boxast 2}$:
        \begin{align}
            \mathscr{G}_{|\mathfrak{a}^{\boxast \varoast \boxast 2}|}(\nu) &= \mathscr{G}_{|\mathfrak{a}^{\boxast \varoast \boxast}|}^2(\nu) = 2^{4 \nu-4} M_2^{4 \nu}\left(1+M_4\right)^{4(1-\nu)}, \\
            \alpha_0^{\boxast \varoast \boxast 2} &= 1-\frac{1}{16}\left(1+M_4\right)^4, \quad \alpha_n^{\boxast \varoast \boxast 2} \equiv 0, \\
            \mathscr{G}_{|\mathfrak{a}^{\boxast \varoast \boxast 2}|}(t) &= \frac{1}{16}\left(\frac{2 M_2}{1+M_4}\right)^{4 t}\left(1+M_4\right)^4, \quad z_1^{\boxast \varoast \boxast 2} = \left( \frac{2 M_2}{1+M_4} \right)^{4}, \quad \alpha_1^{\boxast \varoast \boxast 2} = \frac{1}{16} \left( 1+M_4 \right)^{4}, \\
            \Fai[\boxast \varoast \boxast 2](s) &= \frac{1}{32}\left((1+M_4)^4-16 M_8\right)^{s}\left((1+M_4)^4+16 M_8\right)^{1-s}, \\
            \Ba[\boxast \varoast \boxast 2] &= 1-\frac{1}{16}\left(1+M_4\right)^4+\frac{1}{16}\left(\left(1+M_4\right)^8-256 M_{16}\right)^{1/2}.
        \end{align}
        \item for $\mathfrak{a}^{\boxast \varoast \boxast \varoast}$:
        \begin{align}
            \tilphi[\boxast \varoast \boxast](u) &= \Fai[\boxast \varoast \boxast](\tau) \Fai[\boxast \varoast \boxast](\tau^*) = \frac{1}{64}\left(1-M_4\right)^2\left(M_8+6 M_4+1\right) \equiv \Fphi[\boxast \varoast \boxast], \\
            \Fpsii[\boxast \varoast \boxast](s) &= \Fpsiii[\boxast \varoast \boxast](s) \equiv 0, \\
            \Fai[\boxast \varoast \boxast \varoast](s) &= \Fai[\boxast \varoast \boxast]^2(s) + 2 \alpha_0^{\boxast \varoast \boxast} \Fai[\boxast \varoast \boxast](s) \nonumber \\
            &= \frac{1}{64}\left(1-M_4\right)^{4 s}\left(M_8+6 M_4+1\right)^{2(1-s)} - \frac{1}{16}\left(3+M_4\right)\left(1-M_4\right)^{2 s+1}\left(M_8+6 M_4+1\right)^{1-s} \\
            \alpha_0^{\boxast \varoast \boxast \varoast} &= (\alpha_0^{\boxast \varoast \boxast})^2 + 2 \Fphi[\boxast \varoast \boxast] = \frac{1}{32}\left(1-M_4\right)^2\left(3 M_8+18 M_4+19\right), \\
            \Ba[\boxast \varoast \boxast \varoast] &= 2\Fai[\boxast \varoast \boxast \varoast](1/2) + \alpha_0^{\boxast \varoast \boxast \varoast} \nonumber \\
            &= \frac{1}{16}\left(1-M_4\right)^2\left(\left(M_{8}+6 M_4+1\right)^{1 / 2}+M_4+3\right)^2.
        \end{align}
    \end{enumerate}
    \item[\textit{iv})] for $\mathfrak{a}^{\boxast \varoast 2}$:
    \begin{align}
        \tilphi[\boxast \varoast](u) &= \Fai[\boxast \varoast](\tau) \Fai[\boxast \varoast](\tau^*) = \frac{1}{16}\left(1-M_4\right)^2 \equiv \Fphi[\boxast \varoast], ~\label{eq:tilphi_a011} \\
        \Fpsii[\boxast \varoast](s) &= \Fpsiii[\boxast \varoast](s) \equiv 0, \\
        \Fai[\boxast \varoast 2](s) &= \Fai[\boxast \varoast]^2(s) + 2 \alpha_0^{\boxast \varoast} \Fai[\boxast \varoast](s) \nonumber \\
            &= \frac{1}{16} (1-M_2)^{4s}(1+M_2)^{4(1-s)} + \frac{1}{4}\left(1-M_4\right)\left(1-M_2\right)^{2 s}\left(1+M_2\right)^{2(1-s)},  \\
        \alpha_0^{\boxast \varoast 2} &= (\alpha_0^{\boxast \varoast})^2 + 2 \Fphi[\boxast \varoast] = \frac{3}{8} (1-M_4)^2,  \\
        \Ba[\boxast \varoast 2] &= 2\Fai[\boxast \varoast 2](1/2) + \alpha_0^{\boxast \varoast 2} = (1-M_4)^2. ~\label{eq:B_a011}
    \end{align}
    \begin{enumerate}
        \item for $\mathfrak{a}^{\boxast \varoast 2 \boxast}$:
        \begin{gather}
            \Fai[\boxast \varoast 2](t)+ \Fai[\boxast \varoast 2](t+1) = \frac{1}{8}\left(1+6M_4+M_8\right)\left(\frac{1-M_2}{1+M_2}\right)^{4 t}+\frac{1}{2}\left(1-M_8\right)\left(\frac{1-M_2}{1+M_2}\right)^{2 t}, \\
            w_1^{\boxast \varoast 2} = \left(\frac{1-M_2}{1+M_2}\right)^2, \quad z_1^{\boxast \varoast 2} = \frac{2 M_2}{1+M_4}, \quad \alpha_1^{\boxast \varoast 2} = \frac{1}{2}\left(1-M_8\right), \\
            w_2^{\boxast \varoast 2} = \left(\frac{1-M_2}{1+M_2}\right)^4, \quad z_2^{\boxast \varoast 2} = \frac{4 M_2\left(1+M_4\right)}{1+6 M_4+M_8}, \quad \alpha_2^{\boxast \varoast 2} = \frac{1}{8}\left(1+6M_4+M_8\right), \\
            \Ga[\boxast \varoast 2](\nu) = \frac{1}{8}\left(1+6 M_4+M_8\right) \left(\frac{4 M_2\left(1+M_4\right)}{1+6 M_4+M_8}\right)^{\nu} + \frac{1}{2}\left(1-M_8\right) \left(\frac{2 M_2}{1+M_4}\right)^{\nu}, \\
            \Ga[\boxast \varoast 2 \boxast](t) = \Ga[\boxast \varoast 2]^2(t)=\sum_{i=1}^3\alpha_i^{\boxast \varoast 2 \boxast} \left(z_i^{\boxast \varoast 2 \boxast}\right)^t, \\
            z_1^{\boxast \varoast 2 \boxast} = \frac{4 M_4}{\left(1+M_4\right)^2}, \quad \alpha_1^{\boxast \varoast 2 \boxast} = \frac{1}{4}\left(1-M_8\right)^2, \\
            z_2^{\boxast \varoast 2 \boxast} = \frac{16 M_4\left(1+M_4\right)^2}{\left(1+M_4+6 M_4+M_8\right)^2}, \quad \alpha_2^{\boxast \varoast 2 \boxast} = \frac{1}{64}\left(1+6 M_4+M_8\right)^2, \\
            z_3^{\boxast \varoast 2 \boxast} = \frac{8 M_4}{1+6 M_4+M_8}, \quad \alpha_3^{\boxast \varoast 2 \boxast} = \frac{1}{8}\left(1+6 M_4+M_8\right)\left(1-M_8\right), \\
            \alpha_0^{\boxast \varoast 2 \boxast} = \frac{1}{64}\left(1-M_4\right)^2\left(39+18 M_4-9 M_8\right),
        \end{gather}
        \begin{align}
            \Fai[\boxast \varoast 2 \boxast](s) =& \frac{\left(1-M_4\right)^{2 s}}{8}\left(\left(1-M_4\right)^2\left(1+6 M_4+M_8\right)^{1-s}\right) \nonumber \\
            & +\frac{\left(1-M_4\right)^{4 s}}{128}\left(1+28 M_4+70 M_8+28 M_{12}+M_{16}\right)^{1-s} \nonumber \\
            & +\frac{\left(1-M_4\right)^{2 s}}{16}\left(\left(1-M_8\right)\left(1+14 M_4+M_8\right)^{1-s}\right), \\
            \Ba[\boxast \varoast 2 \boxast] &= 2\Fai[\boxast \varoast 2 \boxast](1/2) + \alpha_0^{\boxast \varoast 2 \boxast} \nonumber \\
                & =\frac{\left(1-M_4\right)^2}{64}\left(16\left(1-M_4\right) \left(1+6 M_4+M_4^2\right)^{1/2}\right. \nonumber \\
                & \quad+\left(1+28 M_4+70 M_4^2+28 M_4^3+M_4^4\right)^{1/2} \nonumber \\
                & \left.\quad+8\left(1+M_4\right) \left(1+14 M_4+M_4^2\right)^{1/2}+39+18 M_4-9 M_4^2\right).
        \end{align}
        \item for $\mathfrak{a}^{\boxast \varoast 3}$:
        \begin{equation}
            \Ba[\boxast \varoast 3] = \left(1-M_4\right)^4.
        \end{equation}
    \end{enumerate}
    \item[\textit{v})] for $\mathfrak{a}^{\varoast \boxast 2}$:
    \begin{align}
        \Fai[\varoast](s) &= A^2(s), \quad \alpha_0^{\varoast}= 2C, \\
        \Fai[\varoast](t)+ \Fai[\varoast](t+1) &= S_2\cdot \mathrm{e}^{2\Delta t}, \\
        w_1^{\varoast} &= \mathrm{e}^{2\Delta}, \quad z_1^{\varoast} = \frac{M_1}{S_2}, \quad \alpha_1^{\varoast} = S_2 \\
        \Ga[\varoast](\nu) &= (z_1^{\varoast})^{\nu}\cdot \alpha_1^{\varoast} = S_2^{1-\nu} M_1^\nu, \quad \alpha_n^{\varoast} \equiv 0, \\
        \Ga[\varoast \boxast](\nu) &= \Ga[\varoast]^2(\nu)=S_2^{2(1-\nu)} M_1^{2\nu}, \quad \alpha_0^{\varoast \boxast} = 4C\bar{C}, \quad \alpha_n^{\varoast \boxast} = 0, \\
        \Ga[\varoast \boxast 2](\nu) &= \Ga[\varoast \boxast]^2(\nu)=S_2^{4(1-\nu)} M_1^{4\nu}, \quad \alpha_0^{\varoast \boxast 2} = 1-S_2^4, \quad \alpha_n^{\varoast \boxast 2} = 0, \\
        \Ga[\varoast \boxast 2](t) &= \Ga[\varoast \boxast]^2(t)=S_2^{4(1-t)} M_1^{4t}, \quad z_1^{\varoast \boxast 2} = \left( \frac{M_1}{S_2} \right)^4, \quad \alpha_1^{\varoast \boxast 2} = S_2^4, \\
        \Fai[\varoast \boxast 2](s) &= \frac{1}{2}\left(\frac{1-z_1^{\varoast \boxast 2}}{1+z_1^{\varoast \boxast 2}}\right)^s\left(1+z_1^{\varoast \boxast 2}\right) \alpha_1^{\varoast \boxast 2}=\frac{1}{2}\left(S_2^4+M_4\right)^{1-s}\left(S_2^4-M_4\right)^s, \\
        \Ba[\varoast \boxast 2]  &= 2\Fai[\varoast \boxast 2](1/2) + \alpha_0^{\varoast \boxast 2} = \left(S_2^8-M_8\right)^{1/2}+1-S_2^4.
    \end{align}
    \begin{enumerate}
        \item for $\mathfrak{a}^{\varoast \boxast 3}$:
        \begin{align}
            \Ga[\varoast \boxast 3](\nu) &= \Ga[\varoast \boxast 2]^2(\nu)=S_2^{8(1-\nu)} M_1^{8\nu}, \quad \alpha_0^{\varoast \boxast 3} = 1-S_2^8, \quad \alpha_n^{\varoast \boxast 3} = 0, \\
            \Ga[\varoast \boxast 3](t) &= \Ga[\varoast \boxast 2]^2(t)=S_2^{8(1-t)} M_1^{8t},  \quad z_1^{\varoast \boxast 3} = \left( \frac{M_1}{S_2} \right)^8, \quad \alpha_1^{\varoast \boxast 3} = S_2^8, \\
            \Fai[\varoast \boxast 3](s) &= \frac{1}{2}\left(S_2^8+M_8\right)^{1-s}\left(S_2^8-M_8\right)^s, \\
            \Ba[\varoast \boxast 3] &= \left(S_2^{16}-M_{16} \right)^{1/2}+1-S_2^{8}.
        \end{align}
        \item for $\mathfrak{a}^{\varoast \boxast 2 \varoast}$:
        \begin{align}
            \tilphi[\varoast \boxast 2](u) &= \Fai[\varoast \boxast 2](\tau) \Fai[\varoast \boxast 2](\tau^*)  = \frac{1}{4}\left(S_2^8-M_8\right) \equiv \Fphi[\varoast \boxast 2], \\
            \Fpsii[\varoast \boxast 2](s) &= \Fpsiii[\varoast \boxast 2](s) \equiv 0, \\
            \Fai[\varoast \boxast 2 \varoast](s) &= \Fai[\varoast \boxast 2]^2(s) + 2\alpha_0^{\varoast \boxast 2} \Fai[\varoast \boxast 2](s) \nonumber \\
            &= \frac{1}{4}\left(S_2^4-M_4\right)^{2 s}\left(S_2^4+M_4\right)^{2(1-s)}+\left(1-S_2^4\right)\left(S_2^4-M_4\right)^s\left(S_2^4+M_4\right)^{1-s}, \\
            \alpha_0^{\varoast \boxast 2 \varoast} &= \left(\alpha_0^{\varoast \boxast 2}\right)^2+2 \Fphi[\varoast \boxast 2] = 1-2S_2^4+\frac{3}{2}S_2^8-\frac{1}{2}M_8, \\
            \Ga[\varoast \boxast 2 \varoast](\nu) &= M_4^\nu S_2^{-4\nu} \left( 2(1-S_2^4)S_2^4 + 2^{\nu-1} S_2^{8\nu}(M_4^2+S_2^8)^{1-\nu} \right), \\
            \Ba[\varoast \boxast 2 \varoast] &= \left(\left(S_2^8-M_8\right)^{1/2}+1-S_2^4\right)^2.
        \end{align}
    \end{enumerate}
    \item[\textit{vi})] for $\mathfrak{a}^{\varoast \boxast \varoast}$:
    \begin{align}
        \Ga[\varoast \boxast](t) &= \Ga[\varoast]^2(t)=S_2^{2(1-t)} M_1^{2t}, \quad z_1^{\varoast \boxast} = \left( \frac{M_1}{S_2} \right)^2, \quad \alpha_1^{\varoast \boxast} = S_2^2, \\
        \Fai[\varoast \boxast](s) &= \frac{1}{2}\left(\frac{1-z_1^{\varoast \boxast}}{1+z_1^{\varoast \boxast}}\right)^s\left(1+z_1^{\varoast \boxast}\right) \alpha_1^{\varoast \boxast}= \frac{1}{2}\left(S_2^2+M_2\right)^{1-s}\left(S_2^2-M_2\right)^s, \\
        \tilphi[\varoast \boxast](u) &= \Fai[\varoast \boxast](\tau)\Fai[\varoast \boxast](\tau^*) = \frac{1}{4}\left(S_2^4-M_4\right)\equiv \Fphi[\varoast \boxast], \\
        \Fpsii[\varoast \boxast \varoast](s) &= \Fpsiii[\varoast \boxast \varoast](s) \equiv 0, \\
        \Fai[\varoast \boxast \varoast](s) &= \Fai[\varoast \boxast]^2(s) + 2 \alpha_0^{\varoast \boxast} \Fai[\varoast \boxast](s) \nonumber \\ 
            &= \frac{1}{4}\left(S_2^2+M_2\right)^{2(1-s)}\left(S_2^2-M_2\right)^{2s}+ 4C\bar{C}\left(S_2^2+M_2\right)^{1-s}\left(S_2^2-M_2\right)^s, \\
        \alpha_0^{\varoast \boxast \varoast} &= (\alpha_0^{\varoast \boxast})^2 + 2 \Fphi[\varoast \boxast] = 16C^2\bar{C}^2+\frac{1}{2}\left(S_2^4-M_4\right) \\
        \Ba[\varoast \boxast \varoast] &= 2\Fai[\varoast \boxast \varoast](1/2) + \alpha_0^{\varoast \boxast \varoast} = \left(\left(S_2^4-M_4\right)^{1/2}+4C\bar{C}\right)^2.
    \end{align}
    \begin{enumerate}
        \item for $\mathfrak{a}^{\varoast \boxast \varoast \boxast}$:
        \begin{gather}
            \Fai[\varoast \boxast \varoast](t) + \Fai[\varoast \boxast \varoast](s)(t+1) = \frac{1}{2}\left(S_2^4+M_4\right)\left(\frac{S_2^2-M_2}{S_2^2+M_2}\right)^{2 t}+8 C \bar{C} S_2^2\left(\frac{S_2^2-M_2}{S_2^2+M_2}\right)^t, \\
            w_1^{\varoast \boxast \varoast} = \frac{S_2^2-M_2}{S_2^2+M_2}, \quad  z_1^{\varoast \boxast \varoast} =\frac{M_2}{S_2^2}, \quad \alpha_1^{\varoast \boxast \varoast} = 8 C \bar{C} S_2^2, \\
            w_2^{\varoast \boxast \varoast} = \left(\frac{S_2^2-M_2}{S_2^2+M_2}\right)^2, \quad  z_2^{\varoast \boxast \varoast} =\frac{2 S_2^2 M_2}{S_2^4+M_4}, \quad \alpha_2^{\varoast \boxast \varoast} = \frac{1}{2}\left(S_2^4+M_4\right), \\
            \Ga[\varoast \boxast \varoast](\nu) = M_2^\nu\left(8 C \bar{C} S_2^{2(1-\nu)}+2^{\nu-1} S_2^{2 \nu}\left(S_2^4+M_4\right)^{1-\nu}\right), \quad \alpha_n^{\varoast \boxast \varoast} \equiv 0,
        \end{gather}
        \begin{align}
            &\Ga[\varoast \boxast \varoast \boxast](\nu) =  M_2^{2\nu}\left(64 C^2 \bar{C}^2 S_2^{4(1-\nu)} + 2^{\nu+3} C \bar{C} S_2^2\left(S_2^4+M_4\right)^{1-\nu} + 2^{2 \nu-2} S_2^{4 \nu}\left(S_2^4+M_4\right)^{2(1-\nu)} \right), \\
            &\alpha_0^{\varoast \boxast \varoast \boxast} = 32C^2\bar{C}^2-M_4+S_2^4-\frac{1}{4}\left(32C^2\bar{C}^2-M_4+S_2^4\right)^2, \\
            &z_1^{\varoast \boxast \varoast \boxast}= \left(\frac{M_2}{S_2^2}\right)^2, \quad z_2^{\varoast \boxast \varoast \boxast}=\frac{2 M_2^2}{S_2^4+M_4}, \quad z_3^{\varoast \boxast \varoast \boxast}=\left(\frac{2 M_2 S_2^2}{S_2^4+M_4}\right)^2, \\
            &\alpha_1^{\varoast \boxast \varoast \boxast}= 64 C^2 \bar{C}^2 S_2^4, \quad \alpha_2^{\varoast \boxast \varoast \boxast}=8 C \bar{C} S_2^2 \left(S_2^4+M_4\right), \quad \alpha_3^{\varoast \boxast \varoast \boxast}=\frac{1}{4}\left(S_2^4+M_4\right)^2, \\
            &\Fai[\varoast \boxast \varoast \boxast](s)=32 C^2 \bar{C}^2\left(S_2^4-M_4\right)^s\left(S_2^4+M_4\right)^{1-s}+4 C \bar{C} S_2^2\left(S_2^4-M_4\right)^s\left(S_2^4+3 M_4\right)^{1-s} \nonumber \\
            &\qquad  \qquad + \frac{1}{8}\left(S_2^4-M_4\right)^{2 s}\left(S_2^8+6 S_2^4 M_4+M_8\right)^{1-s}, \\
            &\Ba[\varoast \boxast \varoast \boxast] = 64 C^2 \bar{C}^2 \left( S_2^8-M_8 \right)^{1/2} + 8 C \bar{C} S_2^2 \left(S_2^4-M_4\right)^{1/2}\left(S_2^4+3 M_4\right)^{1/2} \nonumber \\
            &\qquad  \qquad + \frac{1}{4}\left(S_2^4-M_4\right)\left( S_2^8+6 S_2^4 M_4+M_8 \right)^{1/2} \nonumber \\
            &\qquad  \qquad + 32C^2\bar{C}^2-M_4+S_2^4-\frac{1}{4}\left(32C^2\bar{C}^2-M_4+S_2^4\right)^2.
        \end{align}
        \item for $\mathfrak{a}^{\varoast \boxast \varoast 2}$:
        \begin{align}
            \tilphi[\varoast \boxast \varoast](u) &= \Fai[\varoast \boxast \varoast](\tau)\Fai[\varoast \boxast \varoast](\tau^*)= \frac{1}{16}\left(S_2^4-M_4\right)^2 + 16 C^2 \bar{C}^2\left(S_2^4-M_4\right) \nonumber \\
            &+ 2 C \bar{C}\left(S_2^4-M_4\right)^{3 / 2} \cos \left( \log \left(\frac{S_2^2+M_2}{S_2^2-M_2}\right) u\right), \\
            \widehat{\varphi}^{\varoast \boxast \varoast}(\omega) &= 2\pi\left( \Fphi[\varoast \boxast \varoast] \cdot \delta(\omega)+ C \bar{C}\left(S_2^4-M_4\right)^{3 / 2} \delta\left(\omega\pm\omega_1\right)\right), \\
            \Fphi[\varoast \boxast \varoast] &= \frac{1}{16}\left(S_2^4-M_4\right)^2 + 16 C^2 \bar{C}^2\left(S_2^4-M_4\right), \quad \omega_1 = \log \left(\frac{S_2^2+M_2}{S_2^2-M_2}\right), \\
            \Fpsii[\varoast \boxast \varoast](s) &= C \bar{C}\left(S_2^2+M_2\right)^{2-s}\left(S_2^2-M_2\right)^{s+1}, \\
            \Fai[\varoast \boxast \varoast 2](s) &= \Fai[\varoast \boxast \varoast]^2(s) + 2\alpha_0^{\varoast \boxast \varoast}\Fai[\varoast \boxast \varoast](s)+2\Fpsii[\varoast \boxast \varoast](s) \nonumber \\
                &= \frac{1}{16}\left(S_2^2+M_2\right)^{4(1-s)}\left(S_2^2-M_2\right)^{4 s}+2 C \bar{C}\left(S_2^2+M_2\right)^{3(1-s)}\left(S_2^2-M_2\right)^{3 s} \nonumber \\
                & +24 C^2 \bar{C}^2\left(S_2^2+M_2\right)^{2(1-s)}\left(S_2^2-M_2\right)^{2 s}+\frac{1}{4}\left(S_2^2+M_2\right)^{3-2 s}\left(S_2^2-M_2\right)^{1+2 s} \nonumber \\
                & +128 C^3 \bar{C}^3\left(S_2^2+M_2\right)^{1-s}\left(S_2^2-M_2\right)^s+2 C \bar{C}\left(S_2^2+M_2\right)^{1+s}\left(S_2^2-M_2\right)^{2-s} \nonumber \\
                & +4 C \bar{C}\left(S_2^2+M_2\right)^{2-s}\left(S_2^2-M_2\right)^{1+s}, \\
            \alpha_0^{\varoast \boxast \varoast 2} &= \left( \alpha_0^{\varoast \boxast \varoast} \right)^2 + 2\Fphi[\varoast \boxast \varoast] = 256 C^4 \bar{C}^4+48 C^2 \bar{C}^2\left(S_2^4-M_4\right)+\frac{3}{8}\left(S_2^4-M_4\right)^2, \\
            \Ba[\varoast \boxast \varoast 2] &= \left(\left(S_2^4-M_4\right)^{1/2}+4C\bar{C}\right)^4.
        \end{align}
    \end{enumerate}
    \item[\textit{vii})] for $\mathfrak{a}^{\varoast 2 \boxast}$: from Table~\ref{tab:bsc_variable_evo} and Example~\ref{exmp:bsc_1101}, we have
    \begin{align}
        \mathscr{F}_{\mathfrak{a}^{\varoast 2},1}(s) &= A^4(s)+4CA^2(s)=4 C^2 (S_{2}-D_{2})^{s-\frac{1}{2}} (D_{2}+S_{2})^{\frac{1}{2}-s}+C^2 (S_{4}-D_{4})^{s-\frac{1}{2}} (D_{4}+S_{4})^{\frac{1}{2}-s}, \\
        \Fphi[\varoast 2] &= 17C^4, \quad \Fpsii[\varoast 2](s)= 4C^3A^2(s), \\
        \alpha_0^{\varoast 2} &= 6C^2, \\
        \Ga[\varoast 2](\nu) &=4 C D_{2}^{\nu} S_{2}^{1-\nu}+D_{4}^{\nu} S_{4}^{1-\nu}, \\
        \Ga[\varoast 2 \boxast](\nu) &= 16 C^2 D_2^{2 \nu} S_2^{2(1-\nu)}+8 C\left(D_2 D_4\right)^\nu\left(S_2 S_4\right)^{1-\nu}+\left(D_4\right)^{2 \nu} S_4^{2(1-\nu)}, \\
        \Fai[\varoast 2 \boxast](s) &= 2^{s+4} C^{2 (s+1)} S_{4}^{1-s}+8 C^{2 s+1} S_{2}^{s} S_{6}^{1-s}+2^s C^{4 s} S_{8}^{1-s}, \\
        \alpha_0^{\varoast 2 \boxast} &= 2\alpha_0^{\varoast 2}-\left(\alpha_0^{\varoast 2}\right)^2 = 12 C^2\left(1-3 C^2\right),\\
        \Ba[\varoast 2 \boxast] &= 2\Fai[\varoast 2 \boxast](1/2) + \alpha_0^{\varoast 2 \boxast} \nonumber \\
            &=  32 \sqrt{2} C^3 S_4^{1/2} + 16 C^{2} S_2^{1/2} S_6^{1/2}  +2 \sqrt{2} C^2 S_8^{1/2}+12 C^2-36 C^4 .
    \end{align}
    \begin{enumerate}
        \item for $\mathfrak{a}^{\varoast 2 \boxast 2}$:
        \begin{align}
            \Ga[\varoast 2 \boxast 2](\nu) &= 256 C^4 D_2^{4 \nu} S_2^{4(1-\nu)} +96 C^2 D_2^{2 \nu} D_4^{2 \nu} S_2^{2(1-\nu)} S_4^{2(1-\nu)} +256 C^3 D_2^{3 \nu} D_4^\nu S_2^{3(1-\nu)} S_4^{1-\nu} \nonumber \\
                & +D_4^{4 \nu} S_4^{4(1-\nu)} +16 C D_2^\nu D_4^{3 \nu} S_2^{1-\nu} S_4^{3(1-\nu)}, \\
            z_1^{\varoast 2 \boxast 2} = \frac{D_2^4}{S_2^4} &, z_2^{\varoast 2 \boxast 2} = \frac{D_2^3 D_4}{S_2^3 S_4}, z_3^{\varoast 2 \boxast 2} = \frac{D_2^2 D_4^2}{S_2^2 S_4^2}, z_4^{\varoast 2 \boxast 2} = \frac{D_2 D_4^3}{S_2 S_4^3}, z_5^{\varoast 2 \boxast 2} = \frac{D_4^4}{S_4^4}, \\
            & \alpha_1^{\varoast 2 \boxast 2} = 256 C^4 S_2^4, \ \alpha_2^{\varoast 2 \boxast 2} = 256 C^3 S_2^3 S_4, \ \alpha_3^{\varoast 2 \boxast 2} = 96 C^2 S_2^2 S_4^2, \\
            & \alpha_4^{\varoast 2 \boxast 2} = 16 C S_2 S_4^3, \ \alpha_5^{\varoast 2 \boxast 2} = S_4^4, \\
            & \alpha_0^{\varoast 2 \boxast 2} = 24 C^2\left(1-9 C^2+36 C^4-54 C^6\right), \\
            \Fai[\varoast 2 \boxast 2](s) = & 128 C^4\left(S_2^4-D_2^4\right)^s\left(S_2^4+D_2^4\right)^{1-s} + 128 C^3\left(S_2^3 S_4-D_2^3 D_4\right)^s\left(S_2^3 S_4+D_2^3 D_4\right)^{1-s} \nonumber \\
                & + 48 C^2\left(S_2^2 S_4^2-D_2^2 D_4^2\right)^s\left(S_2^2 S_4^2+D_2^2 D_4^2\right)^{1-s}  \nonumber \\ 
                & + 8 C\left(S_2 S_4^3-D_2 D_4^3\right)^s\left(S_2 S_4^3+D_2 D_4^3\right)^{1-s} + \frac{1}{2}\left(S_4^4-D_4^4\right)^s\left(S_4^4+D_4^4\right)^{1-s}, \\
            \Ba[\varoast 2 \boxast 2] =& 256 C^4 \left(S_2^8-D_2^{8}\right)^{1/2}+256 C^3 \left(S_2^6S_4^2-D_2^{6}D_4^2\right)^{1/2} +96 C^2 \left(S_2^4 S_4^4-D_2^4 D_4^4\right)^{1/2} \nonumber \\
                & + 16 C \left(S_2^2 S_4^6-D_2^2 D_4^6\right)^{1/2} +\left(S_4^8-D_4^8\right)^{1/2} \nonumber \\
                & + 24 C^2\left(1-9 C^2+36 C^4-54 C^6\right).
        \end{align}
        \item for $\mathfrak{a}^{\varoast 2 \boxast \varoast}$:
        \begin{align}
            \tilphi[\varoast 2 \boxast](u) &= \Fai[\varoast 2 \boxast](\tau)\Fai[\varoast 2 \boxast](\tau^*) \nonumber\\
                & = 512 C^6 S_4+64 C^4 S_2 S_6+2 C^4 S_8 \nonumber\\
                & +256 \sqrt{2} C^5 S_2^{1/2}S_4^{1/2}S_6^{1/2} \cos \left(u \log \frac{2S_6}{S_2S_4}\right) \nonumber\\
                & +64 C^5 S_4^{1/2} S_8^{1/2} \cos \left(u \log \frac{S_8}{C^2 S_4}\right) \nonumber\\
                & +16 \sqrt{2} C^4 S_2^{1/2}S_6^{1/2}S_8^{1/2} \cos \left(u \log \frac{S_2S_8}{2 C^2S_6}\right), \\
            \Fphi[\varoast 2 \boxast] &= 512 C^6 S_4+64 C^4 S_2 S_6+2 C^4 S_8, \\
            \Fpsii[\varoast 2 \boxast](s) &= 2^{8-s} C^5 S_2^s S_4^s S_6^{1-s} + 32 C^{2 s+4} S_4^s S_8^{1-s} + 2^{s+3} C^{2 s+3} S_6^s S_2^{1-s} S_8^{1-s}, \\
            \Fai[\varoast 2 \boxast \varoast](s) =& \Fai[\varoast 2 \boxast]^2(s) + 2\alpha_0^{\varoast 2 \boxast}\Fai[\varoast 2 \boxast](s)+2\Fpsii[\varoast 2 \boxast](s) \nonumber \\     
                =& \left(2^{s+4} C^{2 (s+1)} S_{4}^{1-s}+8 C^{2 s+1} S_{2}^{s} S_{6}^{1-s}+2^s C^{4 s} S_{8}^{1-s}\right)^2 \nonumber \\
                & + 24 C^2\left(1-3 C^2\right)\left(2^{s+4} C^{2 (s+1)} S_{4}^{1-s}+8 C^{2 s+1} S_{2}^{s} S_{6}^{1-s}+2^s C^{4 s} S_{8}^{1-s}\right)  \nonumber \\
                & + 2^{9-s} C^5 S_2^s S_4^s S_6^{1-s} + 64 C^{2 s+4} S_4^s S_8^{1-s} + 2^{s+4} C^{2 s+3} S_6^s S_2^{1-s} S_8^{1-s},   \nonumber \\ 
                =& 2^{s+8} C^{4 s+3} S_{2}^{s} S_{6}^{1-s} S_{4}^{1-s}+2^{2 s+5} C^{6 s+2} S_{8}^{1-s} S_{4}^{1-s}+3\ 2^{s+7} C^{2 s+4} S_{4}^{1-s}   \nonumber \\
                & - 9\cdot 2^{s+7} C^{2 s+6} S_{4}^{1-s}+2^{9-s} C^5 S_{2}^{s} S_{6}^{1-s} S_{4}^{s}+64 C^{2 s+4} S_{8}^{1-s} S_{4}^{s}+4^{s+4} C^{4 s+4} S_{4}^{2-2 s}   \nonumber \\
                & + 192 C^{2 s+3} S_{2}^{s} S_{6}^{1-s}-576 C^{2 s+5} S_{2}^{s} S_{6}^{1-s}+2^{s+4} C^{6 s+1} S_{2}^{s} S_{6}^{1-s} S_{8}^{1-s}  \nonumber \\
                & + 2^{s+4} C^{2 s+3} S_{2}^{1-s} S_{6}^{s} S_{8}^{1-s}+3\cdot 2^{s+3} C^{4 s+2} S_{8}^{1-s}-9\cdot 2^{s+3} C^{4 s+4} S_{8}^{1-s}  \nonumber \\
                & + 64 C^{4 s+2} S_{2}^{2 s} S_{6}^{2-2 s}+4^s C^{8 s} S_{8}^{2-2 s}, \\
            \Ga[\varoast 2 \boxast \varoast](\nu) =& -1152 C^6 \left(S_{4}-2 C^2\right)^{\nu } \left(2 C^2+S_{4}\right)^{1-\nu }+384 C^4 \left(S_{4}-2 C^2\right)^{\nu } \left(2 C^2+S_{4}\right)^{1-\nu } \nonumber \\
            & + 256 C^4 \left(S_{4}^{2}-4 C^4\right)^{\nu } \left(4 C^4+S_{4}^{2}\right)^{1-\nu }-576 C^5 \left(S_{6}-C^2 S_{2}\right)^{\nu } \left(S_{2} C^2+S_{6}\right)^{1-\nu } \nonumber \\
            & + 192 C^3 \left(S_{6}-C^2 S_{2}\right)^{\nu } \left(S_{2} C^2+S_{6}\right)^{1-\nu }+256 C^5 (2 S_{6}-S_{2} S_{4})^{\nu } (S_{2} S_{4}+2 S_{6})^{1-\nu } \nonumber \\
            & + 256 C^3 \left(S_{4} S_{6}-2 C^4 S_{2}\right)^{\nu } \left(2 S_{2} C^4+S_{4} S_{6}\right)^{1-\nu }+64 C^2 \left(S_{6}^{2}-C^4 S_{2}^{2}\right)^{\nu } \left(S_{2}^{2} C^4+S_{6}^{2}\right)^{1-\nu } \nonumber \\
            & - 72 C^4 \left(S_{8}-2 C^4\right)^{\nu } \left(2 C^4+S_{8}\right)^{1-\nu }+24 C^2 \left(S_{8}-2 C^4\right)^{\nu } \left(2 C^4+S_{8}\right)^{1-\nu } \nonumber \\
            & + 64 C^4 \left(S_{8}-C^2 S_{4}\right)^{\nu } \left(S_{4} C^2+S_{8}\right)^{1-\nu }+16 C^3 \left(S_{2} S_{8}-2 C^2 S_{6}\right)^{\nu } \left(2 S_{6} C^2+S_{2} S_{8}\right)^{1-\nu }  \nonumber \\
            & + 32 C^2 \left(S_{4} S_{8}-4 C^6\right)^{\nu } \left(4 C^6+S_{4} S_{8}\right)^{1-\nu }+16 C \left(S_{6} S_{8}-2 C^6 S_{2}\right)^{\nu } \left(2 S_{2} C^6+S_{6} S_{8}\right)^{1-\nu }  \nonumber \\
            & + \left(S_{8}^{2}-4 C^8\right)^{\nu } \left(4 C^8+S_{8}^{2}\right)^{1-\nu }, \\
            \alpha_0^{\varoast 2 \boxast \varoast} &= \left(\alpha_0^{\varoast 2 \boxast}\right)^2 + 2\Fphi[\varoast 2 \boxast] \nonumber \\
                &= 4 C^4 \left(324 C^4+8 C^2(32 S_{4}-27) +32 S_{2} S_{6}+S_{8}+36\right), \\
            \Ba[\varoast 2 \boxast \varoast] =& \left(32 \sqrt{2} C^3 S_4^{1/2} + 16 C^{2 s+1} S_2^{1/2} S_6^{1/2}  +2 \sqrt{2} C^2 S_8^{1/2}+12 C^2-36 C^4\right)^2.
        \end{align}
    \end{enumerate}
    \item[\textit{viii})] for $\mathfrak{a}^{\varoast 3}$: from Table~\ref{tab:bsc_variable_evo}, we have
    \begin{align}
        \mathscr{F}_{\mathfrak{a}^{\varoast 3},1}(s) &= A^8(s)+8 C A^6(s)+28 C^2 A^4(s)+56 C^3 A^2(s), \nonumber \\
        &=56 C^4 (S_{2}-D_{2})^{s-\frac{1}{2}} (D_{2}+S_{2})^{\frac{1}{2}-s}+28 C^4 (S_{4}-D_{4})^{s-\frac{1}{2}} (D_{4}+S_{4})^{\frac{1}{2}-s}\nonumber \\
        &+8 C^4 (S_{6}-D_{6})^{s-\frac{1}{2}} (D_{6}+S_{6})^{\frac{1}{2}-s}+C^4 (S_{8}-D_{8})^{s-\frac{1}{2}} (D_{8}+S_{8})^{\frac{1}{2}-s}, \\
        \alpha_{0}^{\varoast 3} &= 70C^4, \\
        \Ga[\varoast 3](\nu) &= 56 C^3 D_{2}^{\nu} S_{2}^{1-\nu}+28 C^2 D_{4}^{\nu} S_{4}^{1-\nu}+8 C D_{6}^{\nu} S_{6}^{1-\nu}+D_{8}^{\nu} S_{8}^{1-\nu}, \\
        \Ba[\varoast 3] &= 2\Fai[\varoast 3](1/2) + \alpha_0^{\varoast 3} = 256C^4. ~\label{eq:Ba_111}
    \end{align}
    \begin{enumerate}
        \item for $\mathfrak{a}^{\varoast 3 \boxast}$:
        \begin{gather}
            \mathscr{F}_{\mathfrak{a}^{\varoast 3},1}(t)+\mathscr{F}_{\mathfrak{a}^{\varoast 3},1}(t+1) = S_8 \mathrm{e}^{8 \Delta t}+8 C S_6 \mathrm{e}^{6 \Delta t}+28 C^2 S_4 \mathrm{e}^{4 \Delta t}+56 C^3 S_2 \mathrm{e}^{2 \Delta t}, \\
            z_1^{\varoast 3}= D_2/S_2, \quad z_2^{\varoast 3}= D_4/S_4, \quad z_3^{\varoast 3}= D_6/S_6, \quad z_4^{\varoast 3}= D_8/S_8, \\
            \alpha_1^{\varoast 3}= 56 C^3 S_2, \quad \alpha_2^{\varoast 3}= 28 C^2 S_4, \quad \alpha_3^{\varoast 3}= 8 C S_6, \quad \alpha_4^{\varoast 3}= S_8, \\
            \Ga[\varoast 3](\nu) = 56C^3 D_2^\nu S_2^{1-\nu} + 28C^2 D_4^\nu S_4^{1-\nu} + 8C D_6^\nu S_6^{1-\nu} +  D_8^\nu S_8^{1-\nu},
        \end{gather}
        \begin{align}
            \Ga[\varoast 3 \boxast](\nu) &= 3136 C^6 D_{2}^{2 \nu} S_{2}^{2-2 \nu}+3136 C^5 D_{2}^{\nu} D_{4}^{\nu} S_{2}^{1-\nu} S_{4}^{1-\nu}+784 C^4 D_{4}^{2 \nu} S_{4}^{2-2 \nu}+896 C^4 D_{2}^{\nu} D_{6}^{\nu} S_{2}^{1-\nu} S_{6}^{1-\nu} \nonumber \\
            &+448 C^3 D_{4}^{\nu} D_{6}^{\nu} S_{6}^{1-\nu} S_{4}^{1-\nu}+112 C^3 D_{2}^{\nu} D_{8}^{\nu} S_{2}^{1-\nu} S_{8}^{1-\nu}+56 C^2 D_{4}^{\nu} D_{8}^{\nu} S_{8}^{1-\nu} S_{4}^{1-\nu} \nonumber \\
            &+64 C^2 D_{6}^{2 \nu} S_{6}^{2-2 \nu}+16 C D_{6}^{\nu} D_{8}^{\nu} S_{6}^{1-\nu} S_{8}^{1-\nu}+D_{8}^{2 \nu} S_{8}^{2-2 \nu},
        \end{align}
        \begin{gather}
            z_1^{\varoast 3 \boxast} = \frac{D_2^2}{S_2^2}, \quad
            z_2^{\varoast 3 \boxast} = \frac{D_2D_4}{S_2S_4}, \quad
            z_3^{\varoast 3 \boxast} = \frac{D_2D_6}{S_2S_6}, \quad
            z_4^{\varoast 3 \boxast} = \frac{D_2D_8}{S_2S_8}, \quad
            z_5^{\varoast 3 \boxast} = \frac{D_4^2}{S_4^2}, \\
            z_6^{\varoast 3 \boxast} = \frac{D_4D_6}{S_4S_6}, \quad
            z_7^{\varoast 3 \boxast} = \frac{D_4D_8}{S_4S_8}, \quad
            z_8^{\varoast 3 \boxast} = \frac{D_6^2}{S_6^2}, \quad
            z_9^{\varoast 3 \boxast} = \frac{D_6D_8}{S_6S_8}, \quad
            z_{10}^{\varoast 3 \boxast} = \frac{D_8^2}{S_8^2}, \\
            \alpha_1^{\varoast 3 \boxast} = 3136 C^6 S_2^2,
            \alpha_2^{\varoast 3 \boxast} = 3136 C^5 S_2 S_4,
            \alpha_3^{\varoast 3 \boxast} = 896 C^4 S_2 S_6,
            \alpha_4^{\varoast 3 \boxast} = 112 C^3 S_2 S_8,
            \alpha_5^{\varoast 3 \boxast} = 784 C^4 S_4^2, \\
            \alpha_6^{\varoast 3 \boxast} = 448 C^3 S_4 S_6, \
            \alpha_7^{\varoast 3 \boxast} = 56 C^2 S_4 S_8, \
            \alpha_8^{\varoast 3 \boxast} = 64 C^2 S_6^2, \
            \alpha_9^{\varoast 3 \boxast} = 16 C S_6 S_8, \
            \alpha_{10}^{\varoast 3 \boxast} = S_8^2, \\
            \alpha_0^{\varoast 3 \boxast} = 140 C^4-4900 C^8,
        \end{gather}
        \begin{align}
            \Fai[\varoast 3 \boxast](s) &= 49\cdot 2^{s+6} C^{2 s+6} S_{4}^{1-s}+896 C^{2 s+4} S_{8}^{1-s} S_{4}^{s}+56 C^{4 s+2} S_{12}^{1-s} S_{4}^{s}+3136 C^{2 s+5} S_{2}^{s} S_{6}^{1-s} \nonumber \\
            &+49\cdot 2^{s+4} C^{4 (s+1)} S_{8}^{1-s}+448 C^{4 s+3} S_{2}^{s} S_{10}^{1-s}+112 C^{2 s+3} S_{6}^{s} S_{10}^{1-s}+2^{s+6} C^{6 s+2} S_{12}^{1-s}\nonumber \\
            &+16 C^{6 s+1} S_{2}^{s} S_{14}^{1-s}+2^s C^{8 s} S_{16}^{1-s},
        \end{align}
        \begin{align}
            \Ba[\varoast 3 \boxast] = &6272 \sqrt{2}C^7 S_4^{1/2} \nonumber \\
                & +C^6\left(1568 \sqrt{2}S_8^{1/2}+6272 S_2^{1/2} S_6^{1/2}\right) \nonumber \\
                & +C^5\left(128 \sqrt{2}S_{12}^{1/2}+1792 S_4^{1/2} S_8^{1/2}+896 S_2^{1/2} S_{10}^{1/2}\right) \nonumber \\
                & +C^4\left(2 \sqrt{2}S_{16}^{1/2}+224 S_6^{1/2} S_{10}^{1/2}+112 S_4^{1/2} S_{12}^{1/2}+32 S_2^{1/2} S_{14}^{1/2}\right) \nonumber \\
                & +140 C^4-4900 C^8.
        \end{align}
        \item for $\mathfrak{a}^{\varoast 4}$:
        \begin{align}
            \Fai[\varoast 4](s) &= A^{16}(s)+16 C A^{14}(s)+120 C^2 A^{12}(s)+560 C^3 A^{10}(s)  \nonumber\\
                    & +1820 C^4 A^8(s) + 4368 C^5 A^6(s) + 8008 C^6 A^4(s) + 11440 C^7 A^2(s), \\
                    & = 11440 C^8 (S_{2}-D_{2})^{s-\frac{1}{2}} (D_{2}+S_{2})^{\frac{1}{2}-s}+8008 C^8 (S_{4}-D_{4})^{s-\frac{1}{2}} (D_{4}+S_{4})^{\frac{1}{2}-s} \nonumber\\
                    & + 4368 C^8 (S_{6}-D_{6})^{s-\frac{1}{2}} (D_{6}+S_{6})^{\frac{1}{2}-s}+1820 C^8 (S_{8}-D_{8})^{s-\frac{1}{2}} (D_{8}+S_{8})^{\frac{1}{2}-s} \nonumber\\
                    & + 560 C^8 (S_{10}-D_{10})^{s-\frac{1}{2}} (D_{10}+S_{10})^{\frac{1}{2}-s}+120 C^8 (S_{12}-D_{12})^{s-\frac{1}{2}} (D_{12}+S_{12})^{\frac{1}{2}-s} \nonumber\\
                    & + 16 C^8 (S_{14}-D_{14})^{s-\frac{1}{2}} (D_{14}+S_{14})^{\frac{1}{2}-s}+C^8 (S_{16}-D_{16})^{s-\frac{1}{2}} (D_{16}+S_{16})^{\frac{1}{2}-s}, \\
            \alpha_0^{\varoast 4} &= 12870 C^8, \\
            \Ga[\varoast 4](\nu) &= 11440 C^7 D_{2}^{\nu} S_{2}^{1-\nu}+8008 C^6 D_{4}^{\nu} S_{4}^{1-\nu}+4368 C^5 D_{6}^{\nu} S_{6}^{1-\nu}+1820 C^4 D_{8}^{\nu} S_{8}^{1-\nu} \nonumber\\
            &+560 C^3 D_{10}^{\nu} S_{10}^{1-\nu}+120 C^2 D_{12}^{\nu} S_{12}^{1-\nu}+16 C D_{14}^{\nu} S_{14}^{1-\nu}+D_{16}^{\nu} S_{16}^{1-\nu}, \\
            \Ba[\varoast 4] &= 65536 C^8.
        \end{align} 
    \end{enumerate}
\end{enumerate}
\end{proof}


\bibliographystyle{unsrt}  
\bibliography{references}

\end{document}